\documentclass[11pt]{article}
\usepackage{microtype}
\usepackage{amsmath,amsthm,amssymb,amsfonts}
\usepackage{latexsym}
\usepackage{subcaption}
\usepackage{graphicx}
\usepackage{siunitx}
\usepackage{tabu}
\usepackage{fullpage}
\usepackage[backref=page]{hyperref}
\usepackage{color}
\usepackage{wrapfig}
\usepackage{tikz}
\usetikzlibrary{decorations.pathreplacing}
\usepackage{setspace}
\usepackage{algorithm}
\usepackage{multirow}
\usepackage[noend]{algpseudocode}
\usepackage[framemethod=tikz]{mdframed}
\usepackage{xspace}
\usepackage{pgfplots}
\usepackage{framed}
\pgfplotsset{compat=1.5}

\renewcommand{\tilde}{\widetilde}

\newcommand{\score}{\textsc{Score}}

\newcommand{\T}{\mathbf{T}}
\renewcommand{\v}{\mathbf{v}}

\newcommand{\OL}{\mathsf{OL}}

\def \coreset    {\mdef{\textsc{Coreset}}}
\def \mspace    {\mdef{\mathsf{m_{space}}}}

\pgfplotsset{compat=1.5}

\newtheorem{theorem}{Theorem}[section]
\newtheorem{corollary}[theorem]{Corollary}
\newtheorem{lemma}[theorem]{Lemma}

\newtheorem{definition}[theorem]{Definition}

\newtheorem{framework}[theorem]{Framework}

\newcommand{\paren}[1]{\left( {#1} \right)}

\newenvironment{proofof}[1]{\begin{trivlist} \item {\bf Proof
#1:~~}}
  {\qed\end{trivlist}}

\newcommand{\namedref}[2]{\hyperref[#2]{#1~\ref*{#2}}}
\newcommand{\thmlab}[1]{\label{thm:#1}}
\newcommand{\thmref}[1]{\namedref{Theorem}{thm:#1}}
\newcommand{\lemlab}[1]{\label{lem:#1}}
\newcommand{\lemref}[1]{\namedref{Lemma}{lem:#1}}

\newcommand{\corlab}[1]{\label{cor:#1}}
\newcommand{\corref}[1]{\namedref{Corollary}{cor:#1}}
\newcommand{\seclab}[1]{\label{sec:#1}}
\newcommand{\secref}[1]{\namedref{Section}{sec:#1}}
\newcommand{\applab}[1]{\label{app:#1}}
\newcommand{\appref}[1]{\namedref{Appendix}{app:#1}}

\newcommand{\figlab}[1]{\label{fig:#1}}
\newcommand{\figref}[1]{\namedref{Figure}{fig:#1}}
\newcommand{\alglab}[1]{\label{alg:#1}}
\renewcommand{\algref}[1]{\namedref{Algorithm}{alg:#1}}

\newcommand{\deflab}[1]{\label{def:#1}}
\newcommand{\defref}[1]{\namedref{Definition}{def:#1}}

\newcommand{\frameref}[1]{\namedref{Framework}{frame:#1}}
\newcommand{\framelab}[1]{\label{frame:#1}}

\newenvironment{remindertheorem}[1]{\medskip\noindent {\bf Reminder of  #1. \newline}\em}{}

\newcommand{\COMMENTED}[1]{{}}

\newcommand{\PPr}[1]{\ensuremath{\mathbf{Pr}\left[#1\right]}}

\newcommand{\onlinepcp}{\textsc{OnlinePCP}\,}
\newcommand{\onlinelp}{\textsc{OnlineL1}\,}
\newcommand{\onlinerss}{\textsc{OnlineRSS}\,}
\newcommand{\onlinebss}{\textsc{OnlineBSS}\,}

\renewcommand{\a}{\mathbf{a}}
\renewcommand{\b}{\mathbf{b}}

\renewcommand{\r}{\mathbf{r}}
\renewcommand{\t}{\mathbf{t}}
\newcommand{\m}{\mathbf{m}}

\newcommand{\x}{\mathbf{x}}

\newcommand{\e}{\mathbf{e}}
\newcommand{\y}{\mathbf{y}}
\newcommand{\A}{\mathbf{A}}
\newcommand{\B}{\mathbf{B}}

\newcommand{\C}{\mathbf{C}}

\newcommand{\I}{\mathbb{I}}

\newcommand{\M}{\mathbf{M}}
\newcommand{\R}{\mathbf{R}}

\renewcommand{\S}{\mathbf{S}}
\renewcommand{\P}{\mathbf{P}}

\newcommand{\U}{\mathbf{U}}

\newcommand{\V}{\mathbf{V}}
\renewcommand{\u}{\mathbf{u}}
\newcommand{\W}{\mathbf{W}}
\newcommand{\X}{\mathbf{X}}
\newcommand{\Y}{\mathbf{Y}}
\newcommand{\Z}{\mathbf{Z}}

\newcommand{\barW}{\mathbf{\bar{W}}}
\newcommand{\barB}{\mathbf{\bar{B}}}
\newcommand{\s}{\textbf{s}}
\newcommand{\z}{\textbf{z}}

\newcommand{\norm}[1]{\left\lVert#1\right\rVert}

\newcommand{\EEx}[1]{\ensuremath{\mathbb{E}\left[#1\right]}}
\renewcommand{\O}[1]{\ensuremath{\mathcal{O}\left(#1\right)}}
\newcommand{\tO}[1]{\ensuremath{\tilde{\mathcal{O}}\left(#1\right)}}
\newcommand{\tT}[1]{\ensuremath{\tilde{\Theta}\left(#1\right)}}

\newcommand{\eps}{\varepsilon}
\renewcommand{\d}{\delta}

\newcommand{\mdef}[1]{{\ensuremath{#1}}\xspace}  
\DeclareMathOperator*{\argmin}{argmin}
\DeclareMathOperator*{\argmax}{argmax}
\DeclareMathOperator*{\polylog}{polylog}
\DeclareMathOperator*{\poly}{poly}
\DeclareMathOperator*{\rank}{rank}

\DeclareMathOperator*{\nnz}{nnz}

\newcommand{\superscript}[1]{\ensuremath{^{\mbox{\tiny{\textit{#1}}}}}\xspace}
\def \th {\superscript{th}}     
\def \etal{\,{\it et~al.}\,}

\newcommand{\ignore}[1]{}
\newif\ifnotes\notestrue 
\ifnotes
\newcommand{\samson}[1]{\textcolor{purple}{ \textbf{Samson: }{#1}} \marginpar{\tiny\bf
             \begin{minipage}[t]{0.5in}
               \raggedright S:
            \end{minipage}}} 
\else
\newcommand{\samson}[1]{}
\fi

\makeatletter
\renewcommand*{\@fnsymbol}[1]{\textcolor{mahogany}{\ensuremath{\ifcase#1\or *\or \dagger\or \ddagger\or
 \mathsection\or \mathparagraph\or \| \or \triangledown \or \bowtie \or **\or \dagger\dagger
   \or \ddagger\ddagger \else\@ctrerr\fi}}}
\makeatother

\providecommand{\email}[1]{\href{mailto:#1}{\nolinkurl{#1}\xspace}}

\definecolor{mahogany}{rgb}{0.75, 0.25, 0.0}
\definecolor{forestgreen}{rgb}{0.13, 0.55, 0.13}
\hypersetup{
     colorlinks   = true,
     citecolor    = forestgreen,
		 linkcolor		= mahogany,
		 urlcolor			= mahogany
}

\begin{document}
\begin{titlepage}
\title{Near Optimal Linear Algebra in the Online and \\Sliding Window Models\thanks{Considerably strengthens and subsumes earlier versions appearing on arXiv.}
}
\author{
Vladimir Braverman\thanks{Johns Hopkins University. 
E-mail: \email{vova@cs.jhu.edu}.}\\
\and
Petros Drineas\thanks{Purdue University. 
E-mail: \email{pdrineas@purdue.edu}}\\
\and
Cameron Musco\thanks{University of Massachusetts Amherst.
E-mail: \email{cmusco@cs.umass.edu}}\\
\and
Christopher Musco\thanks{New York University. 
E-mail: \email{cmusco@nyu.edu}}\\
\and
Jalaj Upadhyay\thanks{Apple, Inc. 
E-mail: \email{jalaj.kumar.upadhyay@gmail.com}}\\
\and
David P. Woodruff\thanks{Carnegie Mellon University. 
E-mail: \email{dwoodruf@cs.cmu.edu}}\\
\and
Samson Zhou\thanks{Carnegie Mellon University. 
E-mail: \email{samsonzhou@gmail.com}}
}

\date{}

\maketitle
\begin{abstract}
We initiate the study of numerical linear algebra in the sliding window model, where only the most recent $W$ updates in a stream form the underlying data set. 
Although many existing algorithms in the sliding window model use or borrow elements from the smooth histogram framework (Braverman and Ostrovsky, FOCS 2007), we show that many interesting linear-algebraic problems, including spectral and vector induced matrix norms, generalized regression, and low-rank approximation, are not amenable to this approach in the row-arrival model. 
To overcome this challenge, we first introduce a unified row-sampling based framework that gives \emph{randomized} algorithms for spectral approximation, low-rank approximation/projection-cost preservation, and $\ell_1$-subspace embeddings in the sliding window model, which often use nearly optimal space and achieve nearly input sparsity runtime. 
Our algorithms are based on ``reverse online'' versions of offline sampling distributions such as (ridge) leverage scores, $\ell_1$ sensitivities, and Lewis weights to quantify both the importance and the recency of a row; our structural results on these distributions may be of independent interest for future algorithmic design. 

Although our techniques initially address numerical linear algebra in the sliding window model, our row-sampling framework rather surprisingly implies connections to the well-studied online model; our structural results also give the first sample optimal (up to lower order terms) online algorithm for low-rank approximation/projection-cost preservation. 
Using this powerful primitive, we give online algorithms for column/row subset selection and principal component analysis that resolves the main open question of Bhaskara~\etal\,(FOCS 2019). 
We also give the first online algorithm for $\ell_1$-subspace embeddings. 
We further formalize the connection between the online model and the sliding window model by introducing an \emph{additional} unified framework for \emph{deterministic} algorithms using a merge and reduce paradigm and the concept of online coresets, which we define as a weighted subset of rows of the input matrix that can be used to compute a good approximation to some given function on all of its prefixes. 
Our sampling based algorithms in the row-arrival online model yield online coresets, giving deterministic algorithms for spectral approximation, low-rank approximation/projection-cost preservation, and $\ell_1$-subspace embeddings in the sliding window model that use nearly optimal space. 
\end{abstract}
\setcounter{page}{0}
\end{titlepage}

\section{Introduction}
The advent of big data has reinforced efforts to design and analyze algorithms in the \emph{streaming model}, where data arrives sequentially, can be observed in a small number of passes (ideally once), and the proposed algorithms are allowed to use space that is sublinear in the size of the input. 
For example in a typical e-commerce setup, the entries of a row represent the number of each item purchased by a customer in a transaction. 
As the transaction is completed, the advertiser receives an entire row of information as an update, which corresponds to the \emph{row-arrival model}. 
Then the underlying covariance matrix summarizes information about which items tend to be purchased together, while low-rank approximation identifies a representative subset of transactions. 

However, the streaming model does not fully address settings where the data is time-sensitive; the advertiser is not interested in the outdated behavior of customers.   
Thus one scenario that is not well-represented by the streaming model is when recent data is considered more accurate and important than data that arrived prior to a certain time window, as in applications such as network monitoring~\cite{CormodeM05,CormodeG08,Cormode13}, event detection in social media~\cite{OsborneEtAl2014}, and data summarization~\cite{ChenNZ16,EpastoLVZ17}. 
To model such settings, Datar\etal~\cite{DatarGIM02} introduced the \emph{sliding window model}, which is parametrized by the size $W$ of the window that represents the size of the \emph{active data} that we want to analyze, in contrast to the so-called ``expired data''. 
The objective is to compute or approximate statistics only on the active data using memory that is sublinear in the window size $W$. 

The sliding window model is more appropriate than the unbounded streaming model in a number of applications~\cite{BabcockBDMW02,MankuM12,PapapetrouGD15,WeiLLSDW16}. 
For example in large scale social media analysis, each row in a matrix can correspond to some online document, such as the content of a Twitter post, and given some corresponding time information. 
Although a streaming algorithm can analyze the data starting from a certain time, analysis with a recent time frame, e.g., the most recent week or month, could provide much more attractive information to advertisers or content providers. 
Similarly in the task of data summarization, the underlying data set is a matrix whose rows correspond to a number of subjects, while the columns correspond to a number of features.
Information on each subject arrives sequentially and the task is to select a small number of representative subjects, which is usually done through some kind of PCA~\cite{PapadimitriouY06,QahtanAWZ15}. 
However, if the behavior of the subjects has recently and indefinitely changed, we would like the summary to only be representative of the updated behavior, rather than the outdated information. 

Another time-sensitive scenario that is not well-represented by the streaming model is when irreversible decisions must be made upon the arrival of each update in the stream, which enables further actions downstream, such as in scheduling, facility location, and data structures. 
The goal of the \emph{online model} is to address such settings by requiring immediate and permanent actions on each element of the stream as it arrives, while still remaining competitive with an optimal offline solution that has full knowledge of the entire input. 
We specifically study the case where the online model must also use space sublinear in the size of the input, though this restriction is not always enforced across algorithms in the online model for other problems.  
In the context of online PCA, an algorithm receives a stream of input vectors and must immediately project each input vector into a lower dimension space of its choice. 
The projected vector can then be used as input to some downstream rotationally invariant algorithm, such as classification, clustering, or regression, which would run more efficiently due to the lower dimensional input. 
Moreover, PCA serves as a popular preprocessing step because it often actually \emph{improves} the quality of the solution by removing isotropic noise~\cite{BoutsidisGKL15} from the data, so that in applications such as clustering, the denoised projection can perform better than the original input. 
The online model has also been extensively used in a number of other applications, such as learning~\cite{BlumKRW03,HazanK16,BalcanDV18}, (prophet) secretary problems~\cite{Rubinstein16,EsfandiariHLM17,EhsaniHKS18}, ad allocation~\cite{NaorW18}, and a variety of graph applications~\cite{CohenW18,GamlathKMSW19,CohenPW19}. 

Generally, the sliding window model and the online model do not seem related, resulting in a different set of techniques being developed for each problem and each setting. 
Surprisingly, our results exhibit a seemingly unexplored and interesting connection between the sliding window model and the online model. 
Our observation is that an online algorithm should be correct on all prefixes of the input, in case the stream terminates at that prefix; on the other hand, a sliding window algorithm should be correct on all suffixes of the input, in case the previous elements expire leaving only the suffix (and perhaps a bunch of ``dummy'' elements). 
Then can we gain something by viewing each update to a sliding window algorithm as an update to an online algorithm \emph{in reverse}? 
At first glance, the answer might seem to be no; we cannot simulate an online algorithm with the stream in reverse order because it would have access to the entire stream whereas a sliding window algorithm only maintains a sketch of the previous elements upon each update. 
However, it turns out that in the row-arrival model, a sketch of the previous elements often suffices to approximately simulate the entire stream input to the online algorithm. 
Indeed, we show that any row-sampling based online algorithm for the problems of spectral approximation, low-rank approximation/projection-cost preservation, and $\ell_1$-subspace embedding automatically implies a corresponding deterministic sliding window algorithm for the problem! 

Formally, we study the following numerical linear algebraic problems in the row-arrival online and sliding window models:
\paragraph{Spectral Approximation.}
Given a matrix $\A\in\mathbb{R}^{n\times d}$ and an approximation parameter $\eps>0$, we wish to find a matrix $\M\in\mathbb{R}^{m\times d}$ with $m\ll n$ that is a spectral approximation (or interchangeably, an $\ell_2$-subspace embedding) of $\A$. 
That is, with high probability, our output matrix $\M$ should satisfy $(1-\eps)\norm{\A\x}_2\le\norm{\M\x}_2\le(1+\eps)\norm{\A\x}_2$ for all $\x\in\mathbb{R}^d$. 
Equivalently, we require $(1-\eps)\A^\top\A\preceq\M^\top\M\preceq(1+\eps)\A^\top\A$. 

\paragraph{Low-Rank Approximation/Projection-Cost Preservation.}
In the low-rank approximation problem, we are given a matrix $\A\in\mathbb{R}^{n\times d}$, a rank parameter $k>0$, and an approximation parameter $\eps>0$, and we wish to find a matrix $\M\in\mathbb{R}^{m\times d}$ with $m\ll n$ such that $(1-\eps)\norm{\A-\A_{(k)}}_F^2\le\norm{\M-\M_{(k)}}_F^2\le(1+\eps)\norm{\A-\A_{(k)}}_F^2$, where $\A_{(k)}$ for a matrix $\A$ represents the best rank $k$ approximation to $\A$. 
A stronger notion is a projection-cost preservation\footnote{Rank $k$ projection-cost preservation was originally formulated as an additive-multiplicative guarantee by~\cite{CohenEMMP15} but subsequent literature uses the purely multiplicative form as in~\defref{def:pcp}.}:
\begin{definition}[Rank $k$ Projection-Cost Preservation~\cite{CohenMM17}]
\deflab{def:pcp}
For $m<n$, a matrix $\M\in\mathbb{R}^{m\times d}$ of rescaled rows of $\A\in\mathbb{R}^{n\times d}$ is a $(1+\eps)$ projection-cost preservation if, for all rank $k$ orthogonal projection matrices $\P\in\mathbb{R}^{d\times d}$,
\[(1-\eps)\norm{\A-\A\P}_F^2\le\norm{\M-\M\P}_F^2\le(1+\eps)\norm{\A-\A\P}_F^2.\]
\end{definition}
\noindent
Note if $\M$ is a projection-cost preservation of $\A$, then its best low-rank approximation can be used to find a projection matrix that gives an approximation of the best low-rank approximation to $\A$. 

\paragraph{$\ell_1$-Subspace Embedding.}
The $\ell_1$-subspace embedding problem has similar demands to the spectral approximation problem, but the structural differences between the $\ell_1$ and $\ell_2$ norms require vastly different techniques. 
Given a matrix $\A\in\mathbb{R}^{n\times d}$ and an approximation parameter $\eps>0$, we wish to find a matrix $\M\in\mathbb{R}^{m\times d}$ with $m\ll n$ so that, with high probability, we have $(1-\eps)\norm{\A\x}_1\le\norm{\M\x}_1\le(1+\eps)\norm{\A\x}_1$ for all $\x\in\mathbb{R}^d$. 

\paragraph{Row/Column Subset Selection.}
The row/column subset selection problems are symmetric, depending on whether the arrival model is rows or columns; we address the row subset selection problem in the row-arrival model. 
Given a matrix $\A\in\mathbb{R}^{n\times d}$, a rank parameter $k>0$, and an approximation parameter $\eps>0$, the goal is to select $k$ rows of $\A$ to form a matrix $\M$ to minimize $\norm{\A-\A\M^\dagger\M}_F$.  
Since the matrix $\M$ has rank at most $k$, then $\norm{\A-\A\M^\dagger\M}_F\ge\norm{\A-\A_{(k)}}_F$, but we would ideally like to obtain some guarantee for $\norm{\A-\A\M^\dagger\M}_F$ relative to $\norm{\A-\A_{(k)}}_F$, where $\A_{(k)}$ is the best rank $k$ approximation to $\A$. 

\paragraph{Principal Component Analysis.}
In the (online) PCA problem, rows of the matrix $\A\in\mathbb{R}^{n\times d}$ arrive sequentially in a data stream and after each row $\a_i$ arrives, the goal is to immediately output a row $\m_i\in\mathbb{R}^{m}$ with $m\ll d$ such that at the end of the stream, there exists a low-rank matrix $\X\in\mathbb{R}^{m\times d}$ with
\[\norm{\A-\M\X}_F^2\le(1+\eps)\norm{\A-\A_{(k)}}_F^2,\]
where $m$ should be minimal and $\A_{(k)}$ is the best rank $k$ approximation to $\A$. 
Here the matrix $\M=\m_1\circ\ldots\circ\m_n$ is formed by the concatenation of the rows that we output at each time.

\subsection{Our Contributions}
We initiate and perform a comprehensive study for both randomized and deterministic algorithms in the sliding window model. 
We first present a randomized row sampling framework for spectral approximation, low-rank approximation/projection-cost preservation, and $\ell_1$-subspace embeddings in the sliding window model. Most of our results are space or time optimal, up to lower order terms.  
Our sliding window structural results imply structural results for the online setting, which we use to give algorithms for row/column subset selection, PCA, projection-cost preservation, and subspace embeddings in the online model. 
Our online algorithms are simple and intuitive, yet they either are novel for the particular problem or improve upon the state-of-the-art, e.g., Bhaskara~\etal (FOCS 2019)~\cite{BhaskaraLVZ19}. 
Finally, we formalize a surprising connection between online algorithms and sliding window algorithms by describing a unified framework for deterministic algorithms in the sliding window model based on the merge-and-reduce paradigm and the concept of online coresets, which are provably generated by online algorithms. 
We now describe our results in greater detail.

\paragraph{Row Sampling Framework for the Sliding Window Model.}
One may ask whether existing algorithms in the sliding window model can be generalized to problems for numerical linear algebra. 
For example, it is known that elementary linear-algebraic problems, such as estimating the Frobenius norm, can be addressed in the sliding window model using the smooth histogram framework~\cite{BravermanO07} (we refer to \appref{sec:smooth-background} for general background on the smooth histogram framework). 
We show that it is not the case in general. 
In \appref{sec:counters}, we give counterexamples showing that various linear-algebraic functions, including the spectral norm, vector induced matrix norms, generalized regression, and low-rank approximation, are not smooth according to the definitions of~\cite{BravermanO07} and therefore cannot be used in the smooth histogram framework. 
This motivates the need for new frameworks for problems of linear algebra in the sliding window model. 
We first give a row sampling based framework for space and runtime efficient randomized algorithms for numerical linear algebra in the sliding window model. 
\begin{framework}[Row Sampling Framework for the Sliding Window Model]
\framelab{frame:random}
There exists a row sampling based framework in the sliding window model that upon the arrival of each new row of the stream with condition number $\kappa$ chooses whether to keep or discard each previously stored row, according to some predefined probability distribution for each problem. 
Using the appropriate probability distribution, we obtain for any approximation parameter $\eps>0$:
\begin{enumerate}
\item
A randomized algorithm for spectral approximation in the sliding window model that with high probability, outputs a matrix $\M$ that is a subset of (rescaled) rows of an input matrix $\A\in\mathbb{R}^{W\times d}$ such that $(1-\eps)\A^\top\A \preceq \M^\top\M\preceq (1+\eps)\A^\top\A$, while storing $\O{\frac{d}{\eps^2}\log n\log\kappa}$ rows at any time and using nearly input sparsity time. 
(See \thmref{thm:sw:spectral}.) 
\item
A randomized algorithm for low-rank approximation/projection-cost preservation in the sliding window model that with high probability, outputs a matrix $\M$ that is a subset of (rescaled) rows of an input matrix $\A\in\mathbb{R}^{W\times d}$ such that for all rank $k$ orthogonal projection matrices $\P\in\mathbb{R}^{d\times d}$,
\[(1-\eps)\norm{\A-\A\P}_F^2\le\norm{\M-\M\P}_F^2\le(1+\eps)\norm{\A-\A\P}_F^2,\]
while storing $\O{\frac{k}{\eps^2}\log n\log^2\kappa}$ rows at any time and using nearly input sparsity time. 
(See \thmref{thm:sw:pcp}.)
\item
A randomized algorithm for $\ell_1$-subspace embeddings in the sliding window model that with high probability, outputs a matrix $\M$ that is a subset of (rescaled) rows of an input matrix $\A\in\mathbb{R}^{W\times d}$ such that $(1-\eps)\norm{\A\x}_1\le\norm{\M\x}_1\le(1+\eps)\norm{\A\x}_1$ for all $\x\in\mathbb{R}^d$, while storing $\O{\frac{d^2}{\eps^2}\log^2 n\log(\kappa' nd)}$ rows at any time, where $\kappa'$ is the subset condition number of $\A$. 
(See \thmref{thm:sw:lp}.) 
\end{enumerate}
\end{framework}
Here we say the stream has condition number $\kappa$ if the ratio of the largest to smallest nonzero singular values of any matrix formed by consecutive rows of the stream is at most $\kappa$ and we say the matrix $\A$  has subset condition number $\kappa'$ if the largest condition number of any matrix formed from a subset of rows of $\A$ has condition number $\kappa'$. 

We further show that for low-rank approximation/projection-cost preservation, we can further improve the polylogarithmic factors from $\O{\frac{k}{\eps^2}\log n\log^2\kappa}$ rows stored at any time to $\O{\frac{k}{\eps^2}\log^2 n}$, under the assumption that the entries of the underlying matrix are integers with magnitude at most $\poly(n)$ even though $\log\kappa$ can be as large as $\O{d\log n}$ with these assumptions. 
To the best of our knowledge, not only are our contributions in \frameref{frame:random} the first such algorithms for these problems in the sliding window model, but also \thmref{thm:sw:spectral} and \thmref{thm:sw:pcp} are both space and runtime optimal up to lower order terms, even compared to row sampling algorithms in the offline setting for most reasonable regime of parameters~\cite{AndoniCKQWZ16, DeshpandeV06}. 

\paragraph{Numerical Linear Algebra in the Online Model.}
An important step in the analysis of our row sampling framework for numerical linear algebra in the sliding window model is bounding the sum of the sampling probabilities for each row. 
In particular, we provide a tight bound on the sum of the online ridge leverage scores that was previously unexplored. 
We show that our bounds along with the paradigm of row sampling with respect to online ridge leverage scores offer simple online algorithms that improve upon the state-of-the-art across broad applications. 

\begin{theorem}[Online Rank $k$ Projection-Cost Preservation]
Given parameters $\eps>0$, $k>0$, and a matrix $\A\in\mathbb{R}^{n\times d}$, whose rows $\a_1,\ldots,\a_n$ arrive sequentially in a stream with condition number $\kappa$, there exists an online algorithm that with high probability, outputs a matrix $\M$ that has $\O{\frac{k}{\eps^2}\log n\log^2\kappa}$ (rescaled) rows of $\A$ and for all rank $k$ orthogonal projection matrices $\P\in\mathbb{R}^{d\times d}$,
\[(1-\eps)\norm{\A-\A\P}_F^2\le\norm{\M-\M\P}_F^2\le(1+\eps)\norm{\A-\A\P}_F^2.\]
(See \thmref{thm:online:pcp}.)
\end{theorem}

\thmref{thm:online:pcp} immediately yields improvements on the two online algorithms recently developed by Bhaskara~\etal (FOCS 2019) for online row subset selection and online PCA~\cite{BhaskaraLVZ19}. 
\begin{theorem}[Online Row Subset Selection]
Given parameters $\eps>0$, $k>0$, and a matrix $\A\in\mathbb{R}^{n\times d}$, whose rows $\a_1,\ldots,\a_n$ arrive sequentially in a stream with condition number $\kappa$, there exists an online algorithm that with high probability, outputs a matrix $\M$ with $\O{\frac{k}{\eps}\log n\log^2\kappa}$ rows that contains a matrix $\T$ of $k$ rows such that 
\[\norm{\A-\A\T^{-1}\T}_F^2\le(1+\eps)\norm{\A-\A_{(k)}}_F^2.\] 
(See \thmref{thm:online:rss}.)
\end{theorem}
By comparison, the online row subset selection algorithm of~\cite{BhaskaraLVZ19} stores $\O{\frac{k}{\eps^2}\log n\log^2\kappa}$ rows to succeed with high probability. 
Moreover, our algorithm provides the guarantee of the existence of a subset $\T$ of $k$ rows that provides a $(1+\eps)$-approximation to the best rank $k$ solution, whereas \cite{BhaskaraLVZ19} promises the bicriteria result that their matrix with rank $\O{\frac{k}{\eps^2}\log n\log^2\kappa}$ is a $(1+\eps)$-approximation to the best rank $k$ solution. 

The online PCA algorithm of \cite{BhaskaraLVZ19} also offers this bicriteria guarantee; for an input matrix $\A\in\mathbb{R}^{n\times d}$, they give an algorithm that outputs a matrix $\M\in\mathbb{R}^{n\times m}$, where $m=\O{\frac{k}{\eps^2}(\log n+\log\kappa)^4}$ and a matrix $\X$ of rank $m$, so that $\norm{\A-\M\X}_F^2\le(1+\eps)\norm{\A-\A_{(k)}}_F^2$. 
Our online row subset selection can also be adjoined with the online PCA algorithm of \cite{BhaskaraLVZ19} to offer the promise of the existence of a submatrix $\Y\in\mathbb{R}^{k\times d}$ within $\X$ such that $\norm{\A-\B\Y}_F^2\le(1+\eps)\norm{\A-\A_{(k)}}_F^2$. 
\begin{theorem}[Online Principal Component Analysis]
Given parameters $n,d,k,\eps>0$ and a matrix $\A\in\mathbb{R}^{n\times d}$ whose rows arrive sequentially in a stream with condition number $\kappa$, let $m=\O{\frac{k}{\eps^2}(\log n+\log\kappa)^4}$.
There exists an algorithm for online PCA that immediately outputs a row $\m_i\in\mathbb{R}^{m}$ after seeing row $\a_i\in\mathbb{R}^{d}$ and with high probability, outputs a matrix $\X\in\mathbb{R}^{m\times d}$ at the end of the stream such that
\[\norm{\A-\M\X}_F^2\le(1+\eps)\norm{\A-\A_{(k)}}_F^2,\]
where $\A_{(k)}$ is the best rank $k$ approximation to $\A$. 
Moreover, $\X$ contains a submatrix $\Y\in\mathbb{R}^{k\times d}$ such that there exists a matrix $\B$ such that
\[\norm{\A-\B\Y}_F^2\le(1+\eps)\norm{\A-\A_{(k)}}_F^2.\]
(See \thmref{thm:online:pca}.)
\end{theorem}

Our sliding window algorithm for $\ell_1$-subspace embeddings also uses an online $\ell_1$-subspace embedding algorithm that we develop. 
\begin{theorem}[Online $\ell_1$-Subspace Embedding]
Given $\eps>\frac{1}{n}$ and a matrix $\A\in\mathbb{R}^{n\times d}$ with subset condition number $\kappa$, whose rows $\a_1,\ldots,\a_n$ arrive sequentially in a stream, there exists an online algorithm that outputs a matrix $\M$ with $\O{\frac{d^2}{\eps^2}\log^2 n\log(\kappa nd)}$ (rescaled) rows of $\A$ such that
\[(1-\eps)\norm{\A\x}_1\le\norm{\M\x}_1\le(1+\eps)\norm{\A\x}_1,\]
for all $\x\in\mathbb{R}^d$ with high probability.  
(See \thmref{thm:online:lp}.)
\end{theorem}
\noindent
We summarize these results in \figref{fig:online:results}.
\begin{figure*}[!htb]
\begin{center}
{\tabulinesep=1.2mm
\begin{tabu}{|c|c|c|}\hline
Online Algorithms & Rows Sampled/Output Dimension & Notes \\\hline\hline
Rank $k$ Approximation & $\O{\frac{k}{\eps^2}\log n\log^2\kappa}$ (\thmref{thm:online:pcp}) & \\\hline
\multirow{2}{*}{Row Subset Selection} & $\O{\frac{k}{\eps^2}\log n\log^2\kappa}$~\cite{BhaskaraLVZ19} & Bicriteria \\\cline{2-3} 
& $\O{\frac{k}{\eps}\log n\log^2\kappa}$ (\thmref{thm:online:rss}) & \\\hline
\multirow{2}{*}{Principal Component Analysis} & $\O{\frac{k}{\eps^2}(\log n+\log\kappa)^4}$~\cite{BhaskaraLVZ19} & Bicriteria \\\cline{2-3}
& $\O{\frac{k}{\eps^2}(\log n+\log\kappa)^4}$ (\thmref{thm:online:pca}) & \\\hline
$\ell_1$-Subspace Embedding & $\O{\frac{d^2}{\eps^2}\log^2 n\log\kappa}$ (\thmref{thm:sw:lp}) &  \\\hline
\end{tabu}
}
\end{center}
\vspace{-0.2in}
\caption{Online algorithms for an input matrix of dimension $n\times d$. 
$\kappa$ is the subset condition number for $\ell_1$-subspace embedding and condition number of the stream otherwise. 
Bicriteria denotes that the rank of the matrix within $(1+\eps)$ approximation of the best rank $k$ solution need not have rank at most $k$.}
\figlab{fig:online:results}
\end{figure*}

\paragraph{A Coreset Framework for Deterministic Sliding Window Algorithms.}
To formalize a connection between online algorithms and sliding window algorithms, we give a framework for deterministic sliding window algorithms based on the merge-and-reduce paradigm and the concept of an online coreset, which we define as a weighted subset of rows of $\A$ that can be used to compute a good approximation to some given function on all prefixes of $\A$. 
On the other hand, observe that a row-sampling based online algorithm does not know when the input might terminate, so it must output a good approximation to any prefix of the input, which is exactly the requirement of an online coreset! 
Moreover, an online cannot revoke any of its decisions, so the history of its decisions are fully observable.  
Indeed, each of our online algorithms imply the existence of an online coreset for the corresponding problem. 
Intuition for our framework is presented in \figref{fig:sw:coreset}. 

\begin{figure}[!htb]
\centering
\begin{tikzpicture}
\foreach \i in {-10,...,29}{
\draw (\i/5,0.5)--(\i/5,0);
}
\foreach \i in {30,...,50}{
\draw[green](\i/5,0.5)--(\i/5,0);
}
\draw (5.9cm,-0.2cm) rectangle+(4.2cm,0.9cm);
\draw[->](-2.1cm,-0.3cm) -- (10.1cm,-0.3cm) ;
\node at (4cm,-0.6cm){Stream of rows, active rows in sliding window};

\draw[decorate,decoration={brace}](-2.1cm,1.8cm) -- (3.7cm,1.8cm);
\node at (0.8cm,2.2cm){$\B_3$};
\draw (-2.1cm,0.8cm) rectangle+(5.8cm,0.9cm);
\draw[red] (-0.4,1)--(-0.4,1.5);
\draw[red] (1.0,1)--(1.0,1.5);
\draw[red] (1.2,1)--(1.2,1.5);
\draw[red] (2.2,1)--(2.2,1.5);
\draw[red] (3.4,1)--(3.4,1.5);

\draw[decorate,decoration={brace}](3.7cm,1.8cm) -- (6.9cm,1.8cm);
\node at (5.3cm,2.2cm){$\B_2$};
\draw (3.7cm,0.8cm) rectangle+(3.2cm,0.9cm);
\draw[red] (4.4,1)--(4.4,1.5);
\draw[red] (5.2,1)--(5.2,1.5);
\draw[red] (5.6,1)--(5.6,1.5);
\draw[blue] (6.2,1)--(6.2,1.5);
\draw[blue] (6.8,1)--(6.8,1.5);

\draw[decorate,decoration={brace}](6.9cm,1.8cm) -- (8.7cm,1.8cm);
\node at (7.8cm,2.2cm){$\B_1$};
\draw (6.9cm,0.8cm) rectangle+(1.8cm,0.9cm);
\draw[blue] (7,1)--(7,1.5);
\draw[blue] (7.2,1)--(7.2,1.5);
\draw[blue] (7.8,1)--(7.8,1.5);
\draw[blue] (8.2,1)--(8.2,1.5);
\draw[blue] (8.4,1)--(8.4,1.5);

\draw[decorate,decoration={brace}](8.7cm,1.8cm) -- (10.1cm,1.8cm);
\node at (9.4cm,2.2cm){$\B_0$};
\draw (8.7cm,0.8cm) rectangle+(1.4cm,0.9cm);
\foreach \i in {44,...,50}{
\draw[blue] (\i/5,1)--(\i/5,1.5);
}
\end{tikzpicture}
\caption{Merge and reduce framework for deterministic sliding window algorithms via coresets that accurately approximate \emph{any} suffix of the input. 
The stream of rows proceeds from left to right, with the active rows of the sliding window in green. 
The rows sampled by the coresets $\B_0,\B_1,\ldots$ are in color above, with the blue rows used to approximate the active rows and red rows approximating expired portions of the stream, even if the size of the sliding window is given after the stream.}
\figlab{fig:sw:coreset}
\end{figure}

\begin{framework}[Coreset Framework for Deterministic Sliding Window Algorithms]
\framelab{frame:det}
There exists a merge-and-reduce framework for numerical linear algebra in the sliding window model using online coresets. 
If the input stream has condition number $\kappa$, then for approximation parameter $\eps>\frac{1}{n}$, the framework gives:
\begin{enumerate}
\item
A deterministic algorithm for spectral approximation in the sliding window model that outputs a matrix $\M$ that is a subset of (rescaled) rows of an input matrix $\A\in\mathbb{R}^{W\times d}$ such that $(1-\eps)\A^\top\A \preceq \M^\top\M\preceq (1+\eps)\A^\top\A$, while storing $\O{\frac{d}{\eps^2}\log^4 n\log\kappa}$ rows at any time. 
(See \thmref{thm:det:spectral}). 
\item
A deterministic algorithm for low-rank approximation/projection-cost preservation in the sliding window model that outputs a matrix $\M$ that is a subset of (rescaled) rows of an input matrix $\A\in\mathbb{R}^{W\times d}$ such that for all rank $k$ orthogonal projection matrices $\P\in\mathbb{R}^{d\times d}$,
\[(1-\eps)\norm{\A-\A\P}_F^2\le\norm{\M-\M\P}_F^2\le(1+\eps)\norm{\A-\A\P}_F^2,\]
while storing $\O{\frac{k}{\eps^2}\log^4 n\log^2\kappa}$ rows at any time. 
(See \thmref{thm:det:pcp}.)
\item
A deterministic algorithm for $\ell_1$-subspace embeddings in the sliding window model that outputs a matrix $\M$ that is a subset of (rescaled) rows of an input matrix $\A\in\mathbb{R}^{W\times d}$ such that $(1-\eps)\norm{\A\x}_1\le\norm{\M\x}_1\le(1+\eps)\norm{\A\x}_1$ for all $\x\in\mathbb{R}^d$, while storing $\O{\frac{d}{\eps^2}\log^4 n\log(\kappa' nd)}$ rows at any time, where $\kappa'$ is the subset condition number of $\A$. 
(See \thmref{thm:det:lp}). 
\end{enumerate}
\end{framework}
All of the results presented using \frameref{frame:det} are space optimal, up to lower order terms~\cite{AndoniCKQWZ16, DeshpandeV06, CohenP15}. 
Again we have the property for low-rank approximation/projection-cost preservation that the number of sampled rows can be improved from $\O{\frac{k}{\eps^2}\log n\log^2\kappa}$ to $\O{\frac{k}{\eps^2}\log^2 n}$, under the assumption that the entries of the underlying matrix are integers at most $\poly(n)$ in magnitude.

We remark that neither our randomized framework \frameref{frame:random} nor our deterministic framework \frameref{frame:det} requires the sliding window parameter $W$ as input during the processing of the stream. 
Instead, they create \emph{oblivious} data structures from which approximations for any window can be computed after processing the stream. 
Intuitively, this can be visualized by \figref{fig:sw:coreset}. 
We summarize our sliding window algorithms in \figref{fig:sw:results}.
\begin{figure*}[!htb]
\begin{center}
{\tabulinesep=1.2mm
\begin{tabu}{|c|c|c|}\hline
Sliding Window Algorithms & Rows Sampled (Randomized) & Rows Sampled (Deterministic) \\\hline\hline
Spectral Approximation & $\tT{\frac{d}{\eps^2}}$ (\thmref{thm:sw:spectral}) & $\tT{\frac{d}{\eps^2}}$ (\thmref{thm:det:spectral}) \\\hline
Rank $k$ Approximation & $\tT{\frac{k}{\eps^2}}$ (\thmref{thm:sw:pcp}) & $\tT{\frac{k}{\eps^2}}$ (\thmref{thm:det:pcp}) \\\hline
$\ell_1$-Subspace Embedding & $\tO{\frac{d^2}{\eps^2}}$ (\thmref{thm:sw:lp}) & $\tT{\frac{d}{\eps^2}}$ (\thmref{thm:det:lp}) \\\hline
\end{tabu}
}
\end{center}
\vspace{-0.2in}
\caption{Sliding window algorithms using \frameref{frame:random} for randomized algorithms and \frameref{frame:det} for deterministic algorithms for a stream of length $n$ of dimension $d$ rows.  
We omit dependencies on $\log n$ and $\log\kappa$, where $\kappa$ is the condition number of the stream or the subset condition number. 
Furthermore, \thmref{thm:sw:spectral} and \thmref{thm:sw:pcp} run in nearly input sparsity time. 
}
\figlab{fig:sw:results}
\end{figure*}

\subsection{Overview of Our Techniques}
\seclab{sec:techniques}
The design and analysis of many existing algorithms in the sliding window model use either the exponential histogram framework~\cite{DatarGIM02} or the smooth histogram framework~\cite{BravermanO07}. 
Unfortunately, we show in \appref{sec:counters} that these frameworks cannot be applied to many interesting linear-algebraic functions, such as approximating the spectral norm or vector induced matrix norms, generalized regression, and low-rank approximation. 
This motivates the need for new frameworks for problems of linear algebra in the sliding window model.  

Our first observation is that as additional rows arrive in the stream, the singular values of the underlying matrix cannot decrease. 
Then we must have the Loewner ordering $\M_1^\top\M_1\succeq\ldots\succeq\M_n^\top\M_n$, where $\M_i$ is the matrix formed by the rows that have arrived since time $i$. 
For spectral approximation, we can reduce the number of outer products stored, since we only care when some singular value of a matrix has increased by $(1+\eps)$. 
Thus we can repeatedly delete an index $i$ if $\M_{i-1}^\top\M_{i-1}\preceq(1+\eps)\M_{i+1}\M_{i+1}$ and relabel the indices, since there are no significant differences between the singular values of $\M_{i-1}$ and $\M_{i+1}$. 
We can repeatedly delete matrices until there are about $\O{\frac{d}{\eps}\log\kappa}$ matrices remaining, where $d$ is the dimension of each row and $\kappa$ is the condition number of the stream. 
Now if $A$ and $B$ are substreams, where $B$ is a suffix of $A$, corresponding to matrices $\A$ and $\B$ and $(1-\eps)\A^\top\A\preceq\B^\top\B\preceq\A^\top\A$, then
\[(1-\eps)(\A^\top\A+\C^\top\C)\preceq\B^\top\B+\C^\top\C\preceq\A^\top\A+\C^\top\C\]
for any matrix $\C$ that represents a substream $C$ that arrives right after the substream $A$. 
Hence, if we have a $(1+\eps)$ spectral approximation to some suffix of the stream, it will remain a $(1+\eps)$ spectral approximation upon the arrival of new rows. 
In particular, the matrix represented by the active rows in the sliding window will be sandwiched between two matrices maintained by the algorithm and thus be well-approximated. 
In fact, this approach can easily be seen as a generalization of the smooth histogram framework to matrix functions and is formalized in \appref{app:histogram}. 

Alas, not only is this approach not space optimal, but it does not seem to generalize to other linear algebraic problems in the sliding window model. 
Nevertheless, this warm-up algorithm crucially gives insight into a more space efficient spectral approximation that \emph{does} generalize to other algorithms. 
Observe that the rows of $\M_i$ are a subset of the rows of $\M_j$ for any $i<j$, so $\M_i^\top\M_i$ and $\M_j^\top\M_j$ be storing a lot of redundant information, which suggests a row sampling approach for more space efficient algorithms.  

\paragraph{Spectral Approximation via Row Sampling in the Sliding Window Model.}
The challenge for row sampling approaches in the sliding window model results from two conflicting forces. 
Suppose we have a good approximation $\M$ to the matrix $\A$ consisting of the rows that have already arrived in the stream. 
When a new row $\r$ arrives, we would like to sample $\r$ with high probability if $\r$ is ``important'' based on the rows in $\M$. 
Namely, if $\r$ has high norm or a different direction than the rows of $\M$, then we would like to capture that information by sampling $\r$. 
On the other hand, if $\r$ has low importance based on the existing rows of $\M$, then it seems like we should not sample $\r$. 

However, the sliding window model also emphasizes the more recent rows. 
For example, it may be possible that the rows that follow $\r$ all contain only zeroes and that all rows of before $\r$ are expired at the time of query. 
Since all rows of $\M$ have expired at the time of query, then we would be left with no information about the underlying matrix if we did not sample $\r$. 
This implies that we must \emph{always} store the most recent row and similarly place greater emphasis on more recent rows. 
Although the (ridge) leverage score of a row quantifies the ``uniqueness'' or ``importance'' of a row with respect to all other rows in the matrix, there is no measure that combines both uniqueness and recency. 

We first consider ``online'' versions of sampling distributions, which quantifies the importance of a row with respect to the \emph{previous} rows in the matrix. 
For example, \cite{CohenMP16} introduces online (ridge) leverage scores for a row sampling approach to online spectral approximation. 
By contrast, the sliding window model seems to value the importance of a row with respect to the \emph{following} rows in the matrix. 
Thus we introduce the concept of \emph{reverse online leverage scores} for spectral approximation in the sliding window model. 
Given rows $\r_1,\ldots,\r_t$, we say the reverse online leverage score of row $\r_i$ is the leverage score of $\r_i$ with respect to the matrix formed by the concatenation of the following rows $\r_{i+1}\circ\ldots\circ\r_t$\footnote{We define the reverse online leverage score to be $1$ when a row is not in the span of the following rows.}. 
Note specifically that the reverse online leverage score of the most recent row is always $1$, which matches the previous observation that we should always sample the most recent row in the sliding window model. 

On the other hand, we cannot compute the reverse online leverage scores of each row without storing the entire stream. 
We can compute the reverse online leverage score of a row with respect to the rows that we have sampled as a $(1+\eps)$ approximation to the true score, but a possible concern is that the error at each step compounds and the error at the end could be as large as $(1+\eps)^n$. 
To resolve this issue, we use an idea of~\cite{CohenLMMPS15} that shows if we have a $(1+\eps)$ approximation for each score but then oversample each row by a factor $C>(1+\O{\eps})$, then the error \emph{will not} compound and the resulting matrix will still be a $(1+\eps)$ approximation. 
We also use a matrix martingale argument similar to~\cite{CohenMP16} to avoid known row sampling dependencies~\cite{KelnerL13}. 

To bound the rows sampled by our algorithm, we note that each row is sampled with probability proportional to the reverse online leverage score, and~\cite{CohenMP16} bounds the sum of the online leverage scores by $\O{d\log\kappa}$ assuming bounded entries in the underlying matrix, which must also bound the sum of the reverse online leverage scores. 
We also note that by batching until a certain number of rows have arrived before choosing to discard rows, we can amortize the runtime needed to approximate each of the reverse online leverage scores and obtain nearly input sparsity time by using standard projection techniques to embed each of the sampled rows into an $\O{\log\frac{n}{\eps}}$ dimensional subspace. 

\paragraph{From Spectral Approximation to a Row Sampling Framework.} 
The spectral approximation algorithm suggests a simple row sampling framework for sliding window algorithms. 
Suppose at each time, we have stored a matrix $\M$ that serves as a good approximation for some function, e.g., a low-rank approximation or an $\ell_1$-subspace embedding, on input matrix $\A$. 
When a new row $\r$ arrives, we add $\r$ to $\M$ and then we start from the most recent row in $\M$ and iteratively choose whether to keep and rescale each row of $\M$ based on some ``reverse online'' sampling distribution that is monotonic -- if the sampling probability of a row increases as additional rows arrive, we can no longer guarantee that we sampled the row with sufficient probability. 
Then generally our analysis must first show correctness of the sampling probabilities and then bound the number of sampled rows. 

\emph{Low-Rank Approximation/Projection-Cost Preservation.} 
For low-rank approximation, the same matrix martingale argument shows that a reverse online version of ridge leverage scores~\cite{AlaouiM15, CohenEMMP15, CohenMM17} with the proper regularization provides a rank $k$ projection-cost preservation of the underlying matrix, and thus a low-rank approximation. 
We then provide a tighter bound on the sum of the (reverse) online ridge leverage scores with regularization parameter $\lambda=\frac{\norm{\A-\A_{(k)}}_F^2}{k}$, which shows that the number of sampled rows will be proportional to $k$ rather than $d$. 
However, this \emph{still} does not suffice for our purpose; we do not know $\frac{\norm{\A-\A_{(k)}}_F^2}{k}$ in advance and thus we cannot regularize. 
Fortunately, the fact that our algorithm is a rank $k$ projection-cost preservation means that we have an $(1+\eps)$ approximation to the regularization parameter at all times. 
Thus we can again oversample by a factor of say $2$ to compensate and avoid compounding errors. 
Similar to our spectral approximation algorithm, our low-rank approximation algorithm has input sparsity runtime, up to lower order factors. 

\emph{$\ell_1$-Subspace Embedding.} 
For $\ell_1$-sampling, we define the reverse online $\ell_1$-sensitivity of a row $\a_i\in\mathbb{R}^d$ as $\max_{\x\in\mathbb{R}^d}\frac{|\a_i^\top\x|}{\norm{\Z_i\x}_1}$, where $\Z_i$ is the matrix that consists of the rows including and following $\a_i$. 
The analysis showing correctness of this probability distribution is straightforward; it follows from a scalar martingale concentration inequality to show that $\norm{\M\x}_1\approx\norm{\A\x}_1$ for all points $\x$ in an $\eps$-net, where $\A$ is the input matrix and $\M$ is the sampled matrix. 
By a simple iterative argument, it then follows that $\norm{\M\x}_1\approx\norm{\A\x}_1$ for all $\x\in\mathbb{R}^d$. 
It follows that given the correctness of $\M$ at all times, it is straightforward for our algorithm to approximate the reverse online $\ell_1$ sensitivity of each row.  
We could use a reverse online version of the $\ell_1$ leverage scores~\cite{DasguptaDHKM08}, but these quantities would result in a higher number of sampled rows. 
We could also try a reverse online version of the Lewis weights~\cite{CohenP15}, but it does seem apparent how to approximate these quantities.

The challenge is bounding the sum of the (reverse) online $\ell_1$ sensitivities. 
A natural approach would be adapting the approach~\cite{CohenMP16}, who relate the sum of the online $\ell_2$ leverage scores to the evolution of the determinant of $\A^\top\A$ as additional rows of $\A$ arrive in the stream. 
A similar geometric argument of relating the sum of the online $\ell_1$ sensitivities to the change in volume of some polytope induced by $\A$ does not seem obvious. 
A primary reason for this is that, unlike for $\ell_2$, the unit ball $\{x\,|\,\norm{\A\x}_1 \leq 1\}$ is not an ellipsoid and it is not clear how the John ellipsoid or the $\ell_1$ sensitivities of this polytope change when a new row is added. 

Instead, we first show that if the online $\ell_2$ leverage scores are uniformly bounded by roughly $\frac{Cd\log\kappa}{n}$ given some regularization of the matrix for any constant $C>1$, then the online $\ell_1$ sensitivities must also be uniformly bounded by $\frac{Cd\log\kappa}{n}$. 
Now even under our regularization, the input matrix $\A$ does not have uniformly bounded online $\ell_2$ leverage scores. 
We thus use a reweighting idea of~\cite{CohenLMMPS15} to argue that we can remove half of the rows of $\A$, while increasing the online $\ell_2$ leverage scores of the other half of the rows of $\A$ to at most $\frac{Cd\log\kappa}{n}$, given our regularization. 
Since sensitivities only increase with the removal of rows, it follows that the sum of the (regularized) online $\ell_1$ sensitivities of the half of the rows that remain with respect to $\A$ must be $\O{d\log\kappa}$. 
We then induct on the matrix formed by the removed rows to argue that the total sum of the (regularized) online $\ell_1$ sensitivities of $\A$ must be $\O{d\log n\log\kappa}$, which implies a bound on the total number of sampled rows.  
This also implies the first online algorithm for an $\ell_1$-subspace embedding. In fact, our analysis can be used to improve a number of other online algorithms. 

\paragraph{Rank Constrained Online Algorithms.}
Our analysis for the sum of the online ridge leverage scores with regularization $\lambda=\frac{\norm{\A-\A_{(k)}}_F^2}{k}$ immediately gives a nearly space optimal algorithm for low-rank approximation/projection-cost preservation in the online model. 
As in the sliding window setting, we do not know the value of $\lambda$ in advance, but since our sketch is a rank $k$ projection-cost preservation, we can track the evolution of $\frac{\norm{\A-\A_{(k)}}_F^2}{k}$ as additional rows of $\A$ arrive in the stream. 
Correctness follows from the same matrix martingale argument as~\cite{CohenMP16} and the improved space bounds follow from our analysis bounding the sum of the online ridge leverage scores. 
We then show that an online algorithm providing a rank-$k$ projection-cost preservation is a powerful primitive that can be used in conjunction with existing techniques to improve previous work.  

\emph{Online Row Subset Selection.} 
For row subset selection in the online model, our starting point is an offline algorithm by \cite{CohenMM17}, who observe that given a matrix $\Z$ that is a constant factor low-rank approximation to the underlying matrix $\A$, a theorem by \cite{DeshpandeRVW06} shows that adaptive sampling $\O{\frac{k}{\eps}}$ additional rows $\S$ of $\A$ against the rows of $\Z$ suffices for $\Z\cup\S$ to contain a $(1+\eps)$ factor approximation to the online row subset selection problem. 
Moreover, $\Z\cup\S$ contains a submatrix $\T$ of $k$ rows that is a good low-rank approximation to $\A$.  
Since our online projection-cost preservation algorithm can output such a matrix $\Z$, we adapt the theorem of \cite{DeshpandeRVW06} to the streaming model, showing that if $\Z$ is given, adaptive sampling can also be performed on data streams to obtain a different but valid $\S$ by oversampling each row of $\A$ by an $\O{\frac{k}{\eps}}$ factor. 

We do not know $\Z$ until the end of the stream, so we cannot immediately perform adaptive sampling in the online model. 
Fortunately, if we adaptively sample against the current output of the online projection-cost preservation algorithm, then we will only oversample with respect to the true sampling probability. 
\cite{CohenMM17} shows that the adaptive sampling probabilities can be upper bounded by the $\lambda$-ridge leverage scores, where $\lambda=\frac{\norm{\A-\A_k}_F^2}{k}$. 
Since the $\lambda$-ridge leverage scores are at most the online $\lambda$-ridge leverage scores, we can again sample rows proportional to their online $\lambda$-ridge leverage scores. 
It then suffices to again use our bound on the sum of the online ridge leverage scores to bound the number of total rows sampled by this algorithm. 

\emph{Online Principal Component Analysis.} 
Our starting point is a recent algorithm by~\cite{BhaskaraLVZ19} that maintains and updates a matrix $\X$ throughout the data stream. 
After the arrival of row $\a_i$, the matrix $\X$ is updated using a combination of residual based sampling and a black-box theorem of Boutsidis et al.~\cite{BoutsidisGKL15}. 
Row $\m_i$ is then output as the embedding of $\a_i$ into $\X$ by $\m_i=\a_i\X^{(i)}$, where $\X^{(i)}$ is the matrix $\X$ after row $i$ has been processed by $\X$. 
$\X$ has the property that no rows from $\X$ are ever removed across the duration of the algorithm, so then $\m_i$ is only an upper bound on the best embedding of $\a_i$. 
However, this matrix $\X$ does not contain a good rank $k$ approximation to $\A$. 
On the other hand, the output to our online row subset selection algorithm \emph{does} contain a submatrix $\Y$ of $k$ rows that is a good approximation to $\A$. 
Thus, our online row subset selection (RSS) algorithm can be combined with the algorithm of \cite{BhaskaraLVZ19} to output matrices $\M$ and $\W$ that is a good approximation to the online PCA problem but also so that $\W$ contains a submatrix $\Y$ of $k$ rows that is a good rank $k$ approximation of $\A$. 
Namely, the algorithm of \cite{BhaskaraLVZ19} can be run to produce a matrix $\X^{(i)}$ after the arrival of each row $\a_i$. 
Simultaneously, our online RSS algorithm produces a matrix $\Z^{(i)}$ after $\a_i$ arrives. 
We then immediately output the embedding $\m_i=\a_i\W^{(i)}$, where $\W^{(i)}$ appends the new rows of $\X^{(i)}$ and $\Z^{(i)}$ to $\W^{(i-1)}$. 
Since the dimension of $\W$ is the sum of the dimensions of $\X$ and $\Z$ but $\X$ has smaller dimension than $\Z$, we do not suffer additional asymptotic space over the algorithm of \cite{BhaskaraLVZ19}.

\paragraph{Online Coresets for Deterministic Sliding Window Algorithms.} 
Our algorithms have repeatedly hinted at a connection between row sampling algorithms for the online model and the sliding window model; if there exists an ``online'' probability distribution for a linear algebraic problem, then there seems to be a corresponding ``reverse online'' probability distribution. 
Online algorithms demand correctness on prefixes; sliding window algorithms demand correctness on suffixes. 
We formalize this intuition by showing a framework for \emph{deterministic} sliding window algorithms using a merge-and-reduce framework based on online coresets, which we define to be a weighted subset of rows of an input matrix $\A$ that can be used to compute a good approximation to some given function on all prefixes of $\A$. 
Observe that any online algorithm in a row-arrival stream that succeeds with high probability \emph{must} yield an online coreset. 
By a simple union bound, it must output the correct answer at all times. 
Thus, it must be correct for all prefixes of the stream. 
Since an online algorithm cannot revoke any of its sampled rows, it follows that we can just consider the rows output at the end of the algorithm and consider the subset of rows that were sampled prior to a certain time to obtain a good approximation to the corresponding prefix of $\A$. 
Hence, online algorithms imply the existence of an online coreset for the corresponding problem. 

Given an online coreset for a particular problem, we obtain a corresponding deterministic sliding window algorithm, as in \figref{fig:sw:coreset}. 
The idea is to store up to the most recent $m$ rows in a block $\B_0$, for some parameter $m$ related to the coreset size. 
When $\B_0$ becomes full, we create a $\left(1+\frac{\eps}{\log n}\right)$ online coreset $\B_1$ for $\B_0$ starting with the most recent row. 
We then reset $\B_0$ to empty and start adding rows to $\B_0$ again. 
At some point $\B_0$ will contain $m$ rows again. 
We then merge all of the rows $\B_0,\B_1,\ldots,\B_i$ where $\B_{i+1}$ is the first empty block. 
It can be shown that the merged rows form a $\left(1+\frac{\eps}{\log n}\right)^i$ coreset for the most recent $2^i\cdot m$ rows, which we can reduce back down to $m$ rows to form block $\B_{i+1}$. 
From a simple induction, it follows that using $\O{\log n}$ blocks will give a $\left(1+\frac{\eps}{\log n}\right)^{\log n}$ coreset, starting with the most recent row. 
Rescaling $\eps$, this gives a merge-and-reduce based framework for $(1+\eps)$ deterministic sliding window algorithms based on online coresets.

\emph{$\ell_1$-Subspace Embedding.} 
For this framework, we can actually use a probability distribution for $\ell_1$-subspace embeddings that is more space efficient than $\ell_1$ sensitivities. 
The Lewis weights~\cite{Lewis78,CohenP15} have been shown to be space optimal in the offline setting, up to lower order terms. 
However, their properties are less understood; rather than impossibility results, the challenge in using online Lewis weights in the row sampling framework is a gap in analysis. 
It did not seem evident how to approximate the online Lewis weight of $\A$ for a row $\r$, given a matrix $\M$ such that $\norm{\M\x}_1\approx\norm{\A\x}_1$ for all $\x\in\mathbb{R}^d$. 
For the merge-and-reduce framework, we have access to each of the rows stored by an online coreset, so we can compute their online Lewis weights through an iterative process. 

It then remains to bound the sum of the online Lewis weights to bound the size of the online coreset. 
We first prove elementary properties of the (online) Lewis weights, such as monotonicity to reweighting and a splitting invariance. 
We also show that if the online $\ell_2$ leverage scores are uniformly bounded by roughly $\frac{Cd\log\kappa}{n}$ given some regularization of the matrix for any constant $C>1$, then the online Lewis weights must also be uniformly bounded by $\frac{Cd\log\kappa}{n}$. 
We then use these properties to bound the online Lewis weights using the same uniformity of reweighting technique used to bound the online $\ell_1$ sensitivities. 
Namely, we again argue that we can remove half of the rows of $\A$, while increasing the online Lewis weights of the other half of the rows of $\A$ to at most $\frac{Cd\log\kappa}{n}$, given our regularization. 
Since Lewis weights only increase with the removal of rows, the sum of the (regularized) online Lewis weights of the half of the rows that remain with respect to $\A$ must be $\O{d\log\kappa}$. 
An inductive argument then shows that the total sum of the (regularized) online Lewis weights of $\A$ must be $\O{d\log n\log\kappa}$, which implies a bound on the online Lewis weights and thus the coreset size. 
This gives a more space efficient algorithm for $\ell_1$-subspace embedding in the sliding window model compared to the previous row sampling framework. 

\subsection{Organization and Open Questions}
In \secref{sec:sample}, we introduce a general framework for space and time efficient randomized matrix algorithms in the sliding window model. 
We use the framework to give algorithms for spectral approximation and low-rank approximation that are nearly space optimal and input sparsity runtime, up to lower order terms. 
We also provide a structural result that bounds the number of rows sampled by a probability distribution induced by online ridge leverage scores, which were recently introduced, but not fully explored. 
Our framework implicitly demonstrates connections between the sliding window and online models through the probability distribution used to sample each row and the corresponding analysis and we make this connection explicit in later sections. 

In \secref{sec:online}, we show that the paradigm of row sampling using online ridge leverage scores along with our structural result achieves simple and intuitive online algorithms for low-rank approximation, row subset selection, and principal component analysis that nevertheless improve on the state-of-the-art. 
In \secref{sec:lp}, we characterize and analyze an intuitive probability distribution for the $\ell_1$-subspace embedding problem and show that it can be used to obtain both online and sliding window algorithms. 

Finally, in \secref{sec:coreset}, we give a general framework for deterministic matrix algorithms in the sliding window model using a merge-and-reduce paradigm for the concept of online coresets, which are generated by all of our online algorithms. 
We define and analyze a space optimal distribution for the $\ell_1$-subspace embedding problem, so that in all, we obtain nearly space optimal deterministic algorithms for spectral approximation, low-rank approximation, and $\ell_1$-subspace embeddings. 

In~\appref{sec:smooth-histogram}, we give background on the popular smooth histogram framework for sliding window algorithms and then give counterexamples showing the approach is not amenable to many interesting linear algebraic functions.

Our work leads to a number of interesting open questions. 
First, note that although we give the first online algorithm for $\ell_1$-subspace embedding, the number of sampled rows is $\O{\frac{d^2}{\eps^2}\log^2 n\log\kappa}$ from using online $\ell_1$ sensitivities to determine the sampling probability for each row. 
This is because even though we show that the sum of the online $\ell_1$ sensitivities is $\O{d\log n\log\kappa}$, the corresponding analysis requires an exponentially small probability of failure due to the $\eps$-net argument. 
We show that the sum of the online Lewis weights is $\O{d\log n\log\kappa}$ and the analysis of \cite{CohenP15} that uses offline Lewis weights does not require the $\eps$-net, but the challenge is approximating the online Lewis weight of a row $\a_i$ given a sketch $\M$ for $\A$. 
Thus, a natural question is whether the online Lewis weights can be used to improve the sample complexity for online $\ell_1$-subspace embedding. 

Another interesting question is whether assumptions on the bit complexity of the underlying matrix $\A$ can remove the dependency on $\log\kappa$ for spectral approximation. 
We showed this is possible for low-rank approximation/projection-cost preservation by separately considering the case when $\A$ has rank at most $2k$, since we can efficiently upper bound $\frac{\norm{\A}_F}{\norm{\A-\A_{(k)}}_F}$ when $\A$ has rank at least $2k$. 

Finally, we describe how to construct online coresets in polynomial time for spectral approximation by generalizing an online version of \cite{BSS12} by \cite{CohenMP16}. 
Do there exist corresponding fast deterministic constructions of online coresets for low-rank approximation/projection-cost preservation and $\ell_1$-subspace embeddings?

\subsection{Preliminaries}
For a positive integer $n$, we use $[n]$ to represent the set $\{1,\ldots,n\}$. 
We use $\frac{1}{\poly(n)}$ to denote some arbitrary degree polynomial in $n$, generally some failure event that can be avoided by fixing sufficiently large constants. 
When an event has probability $1-\frac{1}{\poly(n)}$ of occurring, we say the event occurs with high probability. 
We use $\polylog(n)$ to omit terms that are polynomial in $\log n$ and write $\exp(n)$ to denote $e^n$. 

In the row-arrival model, the stream has length $n$ and the $i\th$ update in the stream is precisely a row $\r_i$. 
In the sliding window model, the input matrix $\A$ is implicitly defined through the stream and a parameter $W>0$ that represents the window size, so that $\A$ is the matrix consisting of the last $W$ rows of the stream, $\A=\r_{n-W+1}\circ\ldots\circ\r_n$, where $\a\circ\b$ denotes the vertical concatenation of rows: $\begin{bmatrix}\a\\\b\end{bmatrix}$. 
In the online model, the input matrix $\A$ is the matrix that consists of all $n$ rows, so that $\A=\r_1\circ\ldots\circ\r_n$. 

We use $\I_n$ to denote the $n\times n$ identity matrix, but drop the subscript when the dimensions are clear from context. 
We use the notation $\A^\top$ to denote the transpose of $\A$ and $\A^{-1}$ to denote the Moore-Penrose pseudoinverse of $\A$ so that $\A\A^{-1}\A=\A$, $\A^{-1}\A\A^{-1}=\A^{-1}$, $(\A\A^{-1})^\top=\A\A^{-1}$, and $(\A^{-1}\A)^\top=\A^{-1}\A$. 
A symmetric matrix $\A\in\R^{n \times n}$ is positive semidefinite if $\x^\top\A\x\ge 0$ for all $\x\in\mathbb{R}^n$, in which case we say $0\preceq\A$. 
Then the Loewner partial ordering of matrices has $\A\preceq\B$ if and only if $0\preceq\B-\A$. 
If $\A$ has rank $r$, then we write its nonzero singular values as $\sigma_{\max}(\A)=\sigma_1(\A)\ge\ldots\ge\sigma_r(\A)=\sigma_{\min}(\A)$. 
We define the condition number of $\A$ by $\frac{\sigma_{\max}(\A)}{\sigma_{\min}(\A)}$ and the operator norm of $\A$ by $\norm{\A}_2=\sigma{\max}(\A)$. 
We use $\kappa$ to represent the condition number of the stream; for the online model, the condition number of the stream is the maximum condition number across any matrix formed by prefixes of the stream, but for the sliding window model, the condition number of the stream is the maximum condition number across any matrix formed by up to $W$ consecutive rows in the stream.

For a vector $\x\in\mathbb{R}^n$, we have the Euclidean norm $\norm{\x}_2=\sqrt{\sum_{i=1}^n v_i^2}$ and more generally, $\norm{\x}_p=\left(\sum_{i=1}^n |x_i|^p\right)^{\frac{1}{p}}$. 
For a matrix $\A\in\mathbb{R}^{n\times d}$, we denote its Frobenius norm by $\norm{\A}_F=\sqrt{\sum_{i=1}^n\sum_{j=1}^d A_{i,j}^2}$. 
We denote $\A_{(k)}=\argmin_{\rank(\X)\le k}\norm{\A-\X}_F^2$ to be the best rank $k$ approximation to $\A$.

We introduce and use multiple formulations, variants, and generalizations of (ridge) leverage scores, which we will define in their various sections, but we use the following definition of (ridge) leverage scores throughout.
\begin{definition}[(Ridge) Leverage Scores]
For a matrix $\A=\a_1\circ\ldots\circ\a_n\in\mathbb{R}^{n\times d}$ and a regularization parameter $\lambda\ge0$, the ridge leverage score of row $\a_i$ for each $i\in[n]$ is the quantity $\a_i(\A^\top\A+\lambda\I)^{-1}\a_i^\top$. 
When $\lambda=0$, we refer to the quantity as the leverage score of row $\a_i$. 
\end{definition}
Informally, the (ridge) leverage score of row $\a_i$ quantifies how ``important'' or ``unique'' a row is. 
When the regularization parameter is $\frac{\norm{\A-\A_{(k)}}_F^2}{k}$, the intuition is that the importance of a row is only considered with respect to the top directions. 
Sampling rows with probability proportional to their leverage scores has been for spectral approximations~\cite{MahoneyD09,DrineasMMW12} while sampling rows with probability proportional to their ridge leverage scores has been used for low-rank approximations~\cite{AlaouiM15, CohenEMMP15, CohenMM17}. 

\section{Row Sampling Framework for the Sliding Window Model}
\seclab{sec:sample}
In this section, we give space and time efficient algorithms for matrix functions in the sliding window model. 
Our general approach will be to use the following framework. 
As the stream $\r_1,\ldots,\r_n\in\mathbb{R}^d$ arrives, we shall maintain a weighted subset of these rows at each time. 
Suppose at some time $t$, we have a matrix $\M_t=\r_{t,1}\circ\ldots\circ\r_{t,m_t}$ of \emph{weighted} rows of the stream that can be used to give a good approximation to the function applied to any \emph{suffix} of the stream. 
Upon the arrival of row $t+1$, we first set $\M_{t+1}=\r_{t+1}$. 
Then starting with $i=m_t$ and moving backwards toward $i=1$, we repeatedly prepend a weighted version of $\r_{t,i}$ to $\M_{t+1}$ with some probability that depends on $\r_{t,i}$, $\M_{t+1}$, and the matrix function to be approximated.  
Once the rows of $\M_t$ have each been either added to $\M_{t+1}$ or discarded, we proceed to row $t+2$. 

Note that the matrices $\M_t$ serve no real purpose other than for presentation; the framework is just storing a subset of weighted rows at each time and repeatedly performing online row sampling, \emph{starting with the most recent row}. 
Since an online algorithm must be correct on all prefixes of the input, then our framework must be correct on all suffixes of the input and in particular, on the sliding window. 
This observation demonstrates a connection between online algorithms and sliding window algorithms that we explore in greater detail in future sections. 
We give our framework in \algref{alg:sampling:framework}. 

\begin{algorithm}[t]
\caption{Row sampling framework for matrix algorithms in the sliding window model}
\alglab{alg:sampling:framework}
\begin{algorithmic}[1]
\Require{A stream of rows $\r_1,\ldots,\r_n\in\mathbb{R}^d$, window size $W$, and an accuracy parameter $\eps>0$}
\Ensure{A $(1+\eps)$ approximation for various matrix functions in the sliding window model.}
\State{$\M_0\gets\emptyset$.}
\State{$\alpha\gets\frac{C}{\eps^2}\log n$ with sufficiently large constant $C>0$}
\For{each row $\r_t$}
\Comment{Process stream}
\State{$\M_t=\r_t$}
\Comment{Keep timestamps of all rows}
\State{Let $\M_{t-1}=\m_1\circ\ldots\circ\m_{m_{t-1}}$.}
\For{$i=m_{t-1}$ down to $i=1$}
\State{$\tau_i\gets\score(\m_i,\M_t)$}
\Comment{Importance of row $i$ based on matrix function}
\State{$p_i\gets\min(1, \alpha\tau_i)$}
\State{With probability $p_i$, $\M_t\gets\frac{\m_i}{\sqrt{p_i}}\circ\M_t$}
\Comment{Downsample and rescale row}
\EndFor
\State{Delete $\M_{t-1}$.}
\EndFor
\State{$\M\gets\emptyset$}
\Comment{Return rows relevant to sliding window}
\State{Let $\M_n=\m_1\circ\ldots\circ\m_{m_n}$.}
\For{$i=1$ to $i=m_n$}
\If{timestamp of $\m_i$ is at least $n-W+1$}
\State{$\M\gets\M\circ\m_i$}
\EndIf
\EndFor
\State{\Return $\M$}
\end{algorithmic}
\end{algorithm}

\subsection{$\ell_2$-Subspace Embedding}
\seclab{sec:l2:randomized}
We first give a randomized algorithm for spectral approximation in the sliding window model that is both space and time efficient. 
\cite{CohenMP16} define the concept of online (ridge) leverage scores and show that by sampling each row of a matrix $\A$ with probability proportional to its online leverage score, the weighted sample at the end of the stream provides a $(1+\eps)$ spectral approximation to $\A$.  
We recall the definition of online ridge leverage scores of a matrix from~\cite{CohenMP16}, as well as introduce reverse online ridge leverage scores. 
\begin{definition}[Online/Reverse Online (Ridge) Leverage Scores]
\deflab{def:rols}
For a matrix $\A=\a_1\circ\ldots\circ\a_n\in\mathbb{R}^{n\times d}$, let $\A_i=\a_1\circ\ldots\circ\a_i$ and $\Z_i=\a_n\circ\ldots\circ\a_i$. 
Let $\lambda\ge 0$. 
The \emph{online $\lambda$-ridge leverage score} of row $\a_i$ is defined to be $\min(1,\a_i(\A_{i-1}^\top\A_{i-1}+\lambda\I)^{-1}\a_i^\top)$, while the \emph{reverse online $\lambda$-ridge leverage score} of row $\a_i$ is defined to be $\min(1,\a_i(\Z_{i+1}^\top\Z_{i+1}+\lambda\I)^{-1}\a_i^\top)$. 
The (reverse) online leverage scores are defined respectively by setting $\lambda=0$, though we use the convention that if the (reverse) online leverage score of $\a_i$ is $1$ if $\rank(\A_i)>\rank(\A_{i-1})$ (respectively if $\rank(\Z_i)>\rank(\Z_{i+1})$).
\end{definition}
Intuitively, the online leverage score quantifies how important row $\a_i$ is, with respect to the previous rows, while the reverse online leverage score quantifies how important row $\a_i$ is, with respect to the following rows, and the ridge leverage scores are regularized versions of these quantities. 
As the name suggests, online (ridge) leverage scores seem appropriate for online algorithms while reverse online (ridge) leverage scores seem appropriate for sliding window algorithms, where more recency is an emphasis. 
Hence, we use reverse online leverage scores in computing the sampling probability of each particular row in \algref{alg:l2:score} that serves as our customized $\score$ function in \algref{alg:sampling:framework} for spectral approximation. 
\begin{algorithm}[t]
\caption{$\score(\r,\A)$ function for spectral approximation}
\alglab{alg:l2:score}
\begin{algorithmic}[1]
\Require{A row $\r\in\mathbb{R}^d$ and a matrix $\A\in\mathbb{R}^{m\times d}$.}
\Ensure{Scaled leverage score of $\r$ with respect to $\A$.}
\If{$\rank(\A)=\rank(\A\circ\r)$}
\State{\Return $2\r(\A^\top\A)^{-1}\r^\top$}
\Comment{In particular, {\color{black}$\score(\m_i,\M_t)$} in {\color{black}\algref{alg:sampling:framework}} is $2\m_i(\M_t^\top\M_t)\m_i^\top$.}
\Else
\State{\Return $1$}
\EndIf
\end{algorithmic}
\end{algorithm}

However, these quantities are related; \cite{CohenMP16} provide an asymptotic bound on the sum of the online ridge leverage scores of any matrix, which also implies a bound on the sum of the reverse online ridge leverage scores, by reversing the order of the rows in a matrix.
\begin{lemma}[Bound on Sum of Online Ridge Leverage Scores]
\cite{CohenMP16}
\lemlab{lem:online:space}
Let the rows of $\A=\a_1\circ\ldots\circ\a_n\in\mathbb{R}^{n\times d}$ arrive in a stream with condition number $\kappa$ and let $\ell_i$ be the online (ridge) leverage score of $\a_i$ with regularization $\lambda$. 
Then $\sum_{i=1}^n\ell_i=\O{d\log\frac{\norm{\A}_2}{\lambda}}$ for $\lambda>\sigma_{\min}(\A)$ and $\sum_{i=1}^n\ell_i=\O{d\log\kappa}$ for $\lambda\le\sigma_{\min}(\A)$. 
It follows that if $\tau_i$ is the reverse online ridge leverage score of $\a_i$, then $\sum_{i=1}^n\tau_i=\O{d\log\frac{\norm{\A}_2}{\lambda}}$ for $\lambda>\sigma_{\min}(\A)$ and $\sum_{i=1}^n\tau_i=\O{d\log\kappa}$ for $\lambda\le\sigma_{\min}(\A)$. 
\end{lemma}
From the definition, it is evident that the reverse online (ridge) leverage scores are monotonic; whenever a new row is added to $\A$, the scores of existing rows cannot increase. 
\begin{lemma}[Monotonicity of Reverse Online (Ridge) Leverage Scores]
\lemlab{lem:scores:monotonic}
For a matrix $\A=\a_1\circ\ldots\circ\a_n\in\mathbb{R}^{n\times d}$, $\lambda\ge0$, and $i\in[n]$, let $\tau_i(\A)$ denote the reverse online $\lambda$-ridge leverage score of row $\a_i$ with respect to $\A$ and let $\tau_i(\B)$ denote the reverse online $\lambda$-ridge leverage score of row $\a_i$ with respect to $\B:=\A\circ\r$ for any row $\r\in\mathbb{R}^{d}$. 
Then $\tau_i(\A)\ge\tau_i(\B)$. 
\end{lemma}
The proof follows immediately from \defref{def:rols} and the fact that $\A_i^\top\A_i+\lambda\I\preceq(\A_i^\top\A_i+\r^\top\r+\lambda\I)$ for any $\A_i=\a_1\circ\ldots\circ\a_i$ and vector $\r\in\mathbb{R}^{d}$. 

We also require the following version of the matrix Freedman concentration inequality:
\begin{theorem}[Matrix Freedman Inequality]
\cite{Tropp2011}
\thmlab{thm:freedman}
Let $\U_0,\ldots,\U_i\in\mathbb{R}^{d\times d}$ be a matrix martingale of symmetric matrices with difference sequence $\W_1,\ldots,\W_i$, where $\W_j=\U_j-\U_{j-1}$. 
If $\norm{\W_j}_2\le R$ for each $j\in[i]$ with high probability and $\norm{\sum_{j=1}^i\EEx{\W_j^2}}_2\le\sigma^2$, then for all $\eps>0$,
\[\PPr{\norm{\U_i}_2\ge\eps}\le d\cdot\exp\left(-\frac{\eps^2/2}{\sigma^2+R\eps/3}\right).\]
\end{theorem}

We now show that at any time $t$, \algref{alg:sampling:framework} using the $\score$ function of \algref{alg:l2:score} stores a matrix $\M_t$ whose rows with timestamp after a time $i\in[t]$ provides a $(1+\eps)$ spectral approximation to any matrix $\Z_i=\r_t\circ\ldots\circ\r_i$. 
This statement shows a good approximation to \emph{any} suffix of the stream at all times and in particular for $t=n$ and $i=n-W+1$, shows that \algref{alg:sampling:framework} using the $\score$ function of \algref{alg:l2:score} outputs a spectral approximation for the matrix induced by the sliding window model. 
\begin{lemma}[Spectral Approximation Guarantee, Bounds on Sampling Probabilities]
\lemlab{lem:l2:downsample}
Let $t\in[n]$, $\lambda\ge 0$ and $\eps>0$. 
For $i\in[t]$, let $q_i=\min(1, \alpha\cdot\r_i(\Z_{i+1}\top\Z_{i+1})^{-1}\r_i^\top)$, where $\Z_{i+1}=\r_t\circ\ldots\circ\r_{i+1}$. 
Then with high probability after the arrival of row $\r_t$, \algref{alg:sampling:framework} using the $\score$ function of \algref{alg:l2:score} will have sampled each row $\r_i$ with probability at least $q_i$ and probability at most $4q_i$. 
Moreover, if $\Y$ is the suffix of $\M_t$ consisting of the (scaled) rows whose timestamps are at least $i$, then 
\[(1-\eps)(\Z_i^\top\Z_i+\lambda\I)\preceq\Y^\top\Y+\lambda\I\preceq(1+\eps)(\Z_i^\top\Z_i+\lambda\I).\]
\end{lemma}
\begin{proof}
We assume $\eps\in\left(0,\frac{1}{2}\right)$ and give the proof by induction on $t$. 
For $t=1$ and ignoring the trivial case where $\r_1$ is the all zeros row, then the input consists of a single nonzero row $\r_1$ whose reverse online leverage score is $1$, so that \algref{alg:sampling:framework} using the $\score$ function of \algref{alg:l2:score} will store $\r_1$, which completes the base case. 

Now we suppose that the conditions hold for $t-1$ and prove they must also for $t$ with high probability. 
We show the conditions hold for each $i\in[t]$ by following the sampling process from $\r_t$ down to $\r_i$. 
We also assume that \algref{alg:sampling:framework} has sampled each row $\r_i$ with probability $\widehat{p_i}$ at least $\frac{1}{2}q_i$ and probability at most $2q_i$. 
Let $k=t-i+1$ and $\U_0,\ldots,\U_k\in\mathbb{R}^{d\times d}$ be a matrix martingale so that $\U_0$ is the all zeros matrix. 
We define $\W_j=\U_j-\U_{j-1}$ for each $j\in[k]$. 
For $j\ge1$, we set $\W_j$ to be the all zeros matrix if $\norm{\U_{j-1}}_2\ge\eps$ and otherwise if $\norm{\U_{j-1}}_2<\eps$, we set the random matrix variable
\[\W_j=
\begin{cases}
\left(\frac{1}{\widehat{p_j}}-1\right) \s_j^\top \s_j & \text{ if }\r_{t-j+1}\text{ is sampled in }\Y\\
- \s_j^\top \s_j  & \text{otherwise,}
\end{cases}
\]
where $\s_j := \r_{t-j+1}(\Z_i^\top\Z_i+\lambda\I)^{-1/2}$ for each $j\in[k]$. 

Thus, the difference sequence $\W_1,\ldots,\W_k$ defines 
\[\U_{j-1}=(\Z_i^\top\Z_i+\lambda\I)^{-1/2}(\Y_{j-1}^\top\Y_{j-1}-\Z_{j-1}^\top\Z_{j-1})(\Z_i^\top\Z_i+\lambda\I)^{-1/2},\]
where $\Y_{j-1}$ are the rows of $\Y$ with timestamp at least $j-1$. 
The predictable quadratic variation process of the martingale $\{\U_j\}$ is hence defined by $\sum_{q=1}^j\EEx{\W_q^2}$ for $1\le j\le k$. 

The remainder of the argument proceeds in the same manner as Lemma 3.3 of \cite{CohenMP16}. 
Specifically, we first note that since $\norm{\U_{j-1}}_2<\eps$ then either $\W_j$ is the all zeros matrix if $\widehat{p_j}=1$ or if $\widehat{p_j}<1$, then $\norm{\W_j}_2\le\frac{2}{\alpha}$ by the assumption that $\widehat{p_j}\ge\frac{1}{2}q_j$ and the definition of $q_j$.  
Thus $\EEx{\W_j^2}\preceq\frac{2}{\alpha}\s_j^\top\s_j$ and so we can use the $\W_j$ matrices to bound the spectral norm of the predictable quadratic variation process $\norm{\sum_{x=1}^j\EEx{\W_x^2}}_2$ of the martingale $\{\U_j\}$ by 
\begin{align*}
\norm{\sum_{x=1}^j\EEx{\W_x^2}}_2\le\norm{\sum_{x=1}^j\frac{2}{\alpha}\s_j^\top\s_j}_2\le\frac{2}{\alpha}\norm{\sum_{x=1}^j\r_{t-j+1}^\top(\Z_i^\top\Z_i+\lambda\I)^{-1}\r_{t-j+1}}_2\le\frac{2}{\alpha},
\end{align*}
where the last inequality follows from the fact that $\sum_{x=1}^j\r_{t-j+1}^\top\r_{t-j+1}=\Z_i^\top\Z_i$. 

By applying the Matrix Freedman inequality (\thmref{thm:freedman}), then it follows that $\norm{\U_k}_2<\eps$ with probability at least $1-d\cdot\exp\left(\frac{-\eps^2/2}{2/\alpha+2\eps/(3\alpha)}\right)$. 
Since $\alpha=\frac{C}{\eps^2}\log n$, then for sufficiently large $C$, $\norm{\U_k}_2<\eps$ with high probability so that
\[\norm{(\Z_i^\top\Z_i+\lambda\I)^{-1/2}(\Y^\top\Y+\lambda\I)(\Z_i^\top\Z_i+\lambda\I)^{-1/2}-\I}_2\le\eps,\]
which implies that
\[(1-\eps)(\Z_i^\top\Z_i+\lambda\I)\preceq\Y^\top\Y+\lambda\I\preceq(1+\eps)(\Z_i^\top\Z_i+\lambda\I),\]
which completes a single step of the second part of the claim.  

Finally, consider the sampling probability of row $\r_{i-1}$. 
Conditioned on $\Y$ being a $(1+\eps)$-spectral approximation to $\Z_i$, then $\r_{i-1}$ is in $\M_{t-1}$ with some probability $\gamma\le 1$ but has been rescaled to $\frac{1}{\sqrt{\gamma}}\r_{i-1}$. 
Crucially, the online ridge leverage scores are monotonic by \lemref{lem:scores:monotonic} so that $\gamma\ge\min(1,\alpha(\r_{i-1}(\Z_i^\top\Z_i+\lambda\I)^{-1}\r_{i-1}^\top))$. 
Although we only have a $(1+\eps)$-spectral approximation to $\Z_i$, the $\score$ function of \algref{alg:l2:score} increases the sampling probability by an extra factor of $2$ to compensate. 
Thus conditioned on $\r_{i-1}$ being in $\M_{t-1}$, then $\r_{i-1}$ remains in $\M_t$ with probability 
\[\min\left(1,\frac{2}{\gamma}\alpha(\r_{i-1}(\Y^\top\Y+\lambda\I)^{-1}\r_{i-1}^\top)\right).\]

Hence, the overall probability that $\r_{i-1}$ is retained is at most $\min(1,4\alpha(\r_{i-1}(\Z_i^\top\Z_i+\lambda\I)^{-1}\r_{i-1}^\top))$ and at least $\min(1,\alpha(\r_{i-1}(\Z_i^\top\Z_i+\lambda\I)^{-1}\r_{i-1}^\top))$ for $\eps<\frac{1}{2}$, which completes a single step of the first part of the claim. 
Thus, both parts of the claim hold by induction. 
\end{proof}

\begin{theorem}[Randomized Spectral Approximation Sliding Window Algorithm]
\thmlab{thm:sw:spectral}
Let $\r_1,\ldots,\r_n\in\mathbb{R}^{d}$ be a stream of rows and $\kappa$ be the condition number of the stream. 
Let $W>0$ be a window size parameter and $\A=\r_{n-W+1}\circ\ldots\circ\r_n$ be the matrix consisting of the $W$ most recent rows. 
Given a parameter $\eps>0$, there exists an algorithm that outputs a matrix $\M$ with a subset of (rescaled) rows of $\A$ such that $(1-\eps)\A^\top\A \preceq \M^\top\M\preceq (1+\eps)\A^\top\A$ and stores $\O{\frac{d}{\eps^2}\log n\log\kappa}$ rows at any time, with high probability.
\end{theorem}
\begin{proof}
The fact that $(1-\eps)\A^\top\A \preceq \M^\top\M\preceq (1+\eps)\A^\top\A$ holds immediately from \lemref{lem:l2:downsample} using $t=n$ and $i=n-W+1$. 
Moreover, \lemref{lem:l2:downsample} implies that each row $\r_i$ is sampled with probability at most $4\alpha\tau_i$, where $\tau_i$ is the reverse online leverage score of row $i$.   
By \lemref{lem:online:space}, we have $\sum_{i=1}^n\tau_i=\O{d\log\kappa}$ and since $\alpha=\O{\frac{1}{\eps^2}\log n}$, then the space complexity of the algorithm follows from a coupling argument and standard Chernoff bounds.
\end{proof}

\paragraph{Nearly Input Sparsity Runtime.} 
We remark that the amortized running time per arriving row can be improved by batching, i.e. processing $\frac{d}{\eps^2}\log n\log\kappa$ rows at a time.  
Since \algref{alg:sampling:framework} using the $\score$ function of \algref{alg:l2:score} stores $\O{\frac{d}{\eps^2}\log n\log\kappa}$ rows, the asymptotic space used by the algorithm will remain the same. 
Observe that it suffices to obtain some constant factor of the reverse online leverage score, since $\alpha$ can just be scaled accordingly. 

To compute a constant factor approximation of the reverse online leverage score of any sampled row $\a_i$, we use standard projection tricks \cite{SpielmanS11, CohenMP16, CohenMM17} embedding each of the $\O{\frac{d}{\eps^2}\log n\log\kappa}$ rows into a $\O{\log\frac{n}{\eps}}$ dimension subspace by applying a Johnson-Lindenstrauss transform. 
Applying this subspace embedding to each of the rows takes $\O{\log\frac{n}{\eps}\cdot\overline{\nnz}}$ time, where $\overline{\nnz}$ is the input sparsity of the batch. 
Subsequently, all operations are in $\O{\log\frac{n}{\eps}}$ dimension subspace, so approximating each reverse online leverage score requires $\polylog\frac{n}{\eps}$ time for each of the $\O{\frac{d}{\eps^2}\log n\log\kappa}$ scores. 
Hence, the amortized time across each batch of $\O{\frac{d}{\eps^2}\log n\log\kappa}$ rows is $\polylog\frac{n}{\eps}\cdot\overline{\nnz}$ so the total runtime is $\polylog\frac{n}{\eps}\cdot\nnz$, where $\nnz$ is the input sparsity of the stream. 
One might observe that the algorithm does not actually know $\kappa$ in advance, but for algorithmic purposes it suffices to downsample when the number of newly arrived rows in the batch equals the number of stored rows after the previous downsampling procedure.

\subsection{Low-Rank Approximation}
\seclab{sec:pcp}
In this section, we give a randomized algorithm for low-rank approximation in the sliding window model that is both space and time optimal, up to lower order terms. 
Throughout this section, we use $\A_{(k)}$ to denote the best rank $k$ approximation to a matrix $\A\in\mathbb{R}^{n\times d}$ so that $\A_{(k)}:=\argmin_{\rank(\X)\le k}\norm{\A-\X}_F^2$. 
Recall the definition of a projection-cost preservation in \defref{def:pcp}, from which it follows that obtaining a projection-cost preservation of $\A$ suffices to produce a low-rank approximation of $\A$. 
\cite{CohenMM17} show that an additive-multiplicative spectral approximation of a matrix $\A$ along with an additional moderate condition that holds for ridge leverage score sampling gives a projection-cost preservation of $\A$.
\begin{lemma}
\cite{CohenMM17}
\lemlab{lem:addmult:pcp}
Let $\A=\a_1\circ\ldots\circ\a_n\in\mathbb{R}^{n\times d}$ and $\lambda\le\frac{\norm{\A-\A_{(k)}}_F^2}{k}$. 
Let $p$ be the largest integer such that $\sigma_p(\A)^2\ge\lambda$ and let $\X=\A-\A_{(p)}$. 
Let $\S\in\mathbb{R}^{m\times n}$ be a sampling matrix so that $\M=\A\S$ is a subset of scaled rows of $\A$. 
If $(1-\eps)(\A^\top\A+\lambda\I)\preceq\M^\top\M+\lambda\I\preceq(1+\eps)(\A^\top\A+\lambda\I)$ and $\left|\norm{\S\X}_F^2-\norm{\X}_F^2\right|\le\eps\norm{\A-\A_{(k)}}_F^2$, then $\M$ is a rank $k$ projection-cost preservation of $\A$ with approximation parameter $24\eps$. 
\end{lemma}
We focus our discussion on the additive-multiplicatve spectral approximation since the same argument of \cite{CohenMM17} with Freedman's inequality rather than Chernoff bounds shows sampling matrices generated from ridge leverage scores satisfy the condition $\left|\norm{\S\X}_F^2-\norm{\X}_F^2\right|\le\eps\norm{\A-\A_{(k)}}_F^2$ with high probability, even when the entries of $\S$ are not independent. 
\begin{lemma}
\cite{CohenMM17}
\lemlab{lem:moderate:pcp}
Let $\A=\a_1\circ\ldots\circ\a_n\in\mathbb{R}^{n\times d}$, $\lambda=\frac{\norm{\A-\A_{(k)}}_F^2}{k}$, and $\tau_i$ be the ridge leverage score of $\a_i$ with regularization $\lambda$. 
Let $p$ be the largest integer such that $\sigma_p(\A)^2\ge\lambda$ and let $\X=\A-\A_{(p)}$. 
Let $\S\in\mathbb{R}^{m\times n}$ be a sampling matrix so that row $\a_i$ is sampled by $\S$, not necessarily independently, with probability at least $\min\left(1,\frac{C\tau_i}{\eps^2}\log n\right)$ for sufficiently large constant $C$. 
Then $\left|\norm{\S\X}_F^2-\norm{\X}_F^2\right|\le\eps\norm{\A-\A_{(k)}}_F^2$ with high probability. 
\end{lemma}

On the other hand, our space analysis in \secref{sec:l2:randomized} relied on bounding the sum of the online leverage scores by $\O{d\log\kappa}$ through \lemref{lem:online:space}; a better bound is not known if we set $\lambda=\frac{\norm{\A-\A_{(k)}}_F^2}{k}$. 
This gap provides a barrier for algorithmic design not only in the sliding window model but also in the online model. 
We show a tighter analysis showing that the sum of the online ridge leverage scores for $\lambda=\frac{\norm{\A-\A_{(k)}}_F^2}{k}$ is $\O{k\log\kappa}$. 

Now if we knew the value of $\frac{\norm{\A-\A_{(k)}}_F^2}{k}$ a priori, we could set $\lambda\le\frac{\norm{\A-\A_{(k)}}_F^2}{k}$ and immediately apply \lemref{lem:l2:downsample} to show that the output of \algref{alg:sampling:framework} with a $\score$ function that uses the $\lambda$ regularization outputs a matrix $\M$ that is a rank $k$ projection-cost preservation of $\A$. 

Initially, even a constant factor approximation to $\frac{\norm{\A-\A_{(k)}}_F^2}{k}$ seems challenging because the quantity is not smooth. 
This issue can be circumvented using additional procedures, such as spectral approximation on rows with reduced dimension~\cite{CohenEMMP15}. 
Even simpler, observe that sampling with any regularization factor $\lambda<\frac{\norm{\A-\A_{(k)}}_F^2}{k}$ would still provide the guarantees of \lemref{lem:addmult:pcp}. 

We could set $\lambda=0$ and still obtain a rank $k$ projection-cost preservation of $\A$, but larger values of $\lambda$ correspond to smaller number of sampled rows and the total number of sampled rows for $\lambda=0$ would be proportional to $d$, as opposed to our goal of $k$.  
Instead, observe that if $\B$ is any prefix or suffix of rows of $\A$, then $\norm{\B-\B_{(k)}}_F^2\le\norm{\B-\B_{(k)}}_F^2$. 
In other words, we can use the rows that have already been sampled to give a constant factor approximation to $\norm{\A-\A_{(k)}}_F^2$ as it evolves, i.e., as more rows of $\A$ arrive. 
We \emph{again} pay for the underestimate to $\norm{\A-\A_{(k)}}_F^2$ by sampling an additional number of rows, but we show that we cannot sample too many rows before our approximation to $\norm{\A-\A_{(k)}}_F^2$ doubles, which only incurs an additional $\O{\log\kappa}$ factor in the number of sampled rows. 
We give the $\score$ function for low-rank approximation in \algref{alg:pcp:score}. 

\begin{algorithm}[t]
\caption{$\score(\r,\A)$ function for rank $k$ projection-cost preservation}
\alglab{alg:pcp:score}
\begin{algorithmic}[1]
\Require{A row $\r\in\mathbb{R}^d$ and a matrix $\A\in\mathbb{R}^{m\times d}$.}
\Ensure{Scaled ridge leverage score of $\r$ with respect to $\A$.}
\State{$\lambda\gets\frac{1}{k}\norm{\A-\A_{(k)}}_F^2$}
\If{$\lambda\neq 0$ or $\rank(\A)=\rank(\A\circ\r)$}
\State{\Return $2\r(\A^\top\A+\lambda\I)^{-1}\r^\top$}
\Else
\State{\Return $1$}
\EndIf
\end{algorithmic}
\end{algorithm}

We first bound the probability that each row is sampled, analogous to \lemref{lem:l2:downsample}. 
\begin{lemma}[Projection-Cost Preservation Guarantee, Bounds on Sampling Probabilities]
\lemlab{lem:pcp:downsample}
Let $t\in[n]$ be fixed and for each $i\in[t]$, let $\Z_i=\r_t\circ\ldots\circ\r_i$. 
Let $\lambda=\frac{\norm{\Z_{i+1}-(\Z_{i+1})_{(k)}}_F^2}{k}$ and $\eps>0$. 
Let $q_i=\min(1, \alpha\cdot\r_i(\Z_{i+1}\top\Z_{i+1}+\lambda\I)^{-1}\r_i^\top)$. 
Then with high probability after the arrival of row $\r_t$, \algref{alg:sampling:framework} using the $\score$ function of \algref{alg:pcp:score} will have sampled row $\r_i$ with probability at least $q_i$ and probability at most $2q_i$. 
Moreover, if $\Y$ is the suffix of $\M_t$ consisting of the (scaled) rows whose timestamps are at least $i$, then 
\[(1-\eps)(\Z_i^\top\Z_i+\lambda\I)\preceq\Y^\top\Y+\lambda\I\preceq(1+\eps)(\Z_i^\top\Z_i+\lambda\I).\]
\end{lemma}
\begin{proof}
Consider \algref{alg:sampling:framework} using the $\score$ function of \algref{alg:pcp:score}. 
We again assume $\eps\in\left(0,\frac{1}{2}\right)$ and prove the statement by induction. 
Recall the convention that if $\lambda=0$ and $\r_1$ is nonzero, then the online leverage score of $\r_1$ is $1$ rather than $0$. 
Thus, the base case holds for $t=1$ because any nonzero will be sampled. 

Now suppose the claim holds for $t-1$ and all indices $i\in[t-1]$ after the arrival of row $\r_{t-1}$. 
We show it holds for $i=t$ down to $i=1$ after the arrival of row $\r_t$. 
Consider a fixed $i\in[t]$. 
By \lemref{lem:addmult:pcp} and \lemref{lem:moderate:pcp}, the additive-multiplicative spectral guarantee of the claim along with the arrival of row $\r_t$ gives a $(1+24\eps)$ approximation to $\norm{\Z_{i+1}-(\Z_{i+1})_{(k)}}_F^2$. 
Note that the error does not compound since the $\score$ function scales the sampling probability by $2$ and thus the $\score$ function of \algref{alg:pcp:score} effectively computes the reverse online ridge leverage score of each row $\r_j$ with $j\ge t-i+1$ using an underestimate of $\lambda=\frac{\norm{\Z_{i+1}-(\Z_{i+1})_{(k)}}_F^2}{k}$. 
The matrix martingale argument in \lemref{lem:l2:downsample} then shows that 
\[(1-\eps)(\Z_i^\top\Z_i+\lambda\I)\preceq\Y^\top\Y+\lambda\I\preceq(1+\eps)(\Z_i^\top\Z_i+\lambda\I),\]
which completes a single step of the second part of the claim. 
By \lemref{lem:addmult:pcp} and \lemref{lem:moderate:pcp}, it follows that $\norm{\Y-\Y_{(k)}}_F^2$ is a $(1+\O{\eps})$-approximation to $\norm{\Z_{i+1}-(\Z_{i+1})_{(k)}}_F^2$. 
The bounds on the sampling probability of row $\r_i$ then follows from the fact that the reverse online ridge leverage scores are monotonic by \lemref{lem:scores:monotonic}. 
Though we only have a $(1+\eps)$ additive-multiplicative spectral approximation to $\Z_i$, which translates to a $(1+24\eps)$ approximation of the regularization, the $\score$ function of \algref{alg:pcp:score} has again compensated by increasing the sampling probability by a factor of $2$. 
The entire claim then follows by induction. 
\end{proof}

We now give a tighter analysis of the sum of the online $\lambda$-ridge leverage scores $l_i$ for $\lambda=\frac{\norm{\A-\A_{(k)}}^2_2}{k}$. 
We require an upper bound from \cite{CohenMP16} on the sum of the online $\lambda$-ridge leverage scores, proven using the matrix determinant lemma.
\begin{lemma}\cite{CohenMP16}
\lemlab{lem:det:lower}
For a matrix $\A=\a_1\circ\ldots\circ\a_n\in\mathbb{R}^{n\times d}$, let $l_i$ denote the online $\lambda$-ridge leverage score of $\a_i$, for each $i\in[n]$. 
Then $\det(\A^\top\A+\lambda\I)\ge\lambda^d e^{\sum l_i/2}$. 
\end{lemma}
 
\begin{lemma}[Bound on Sum of Online Ridge Leverage Scores]
\lemlab{lem:pcp:scores}
Let $\A\in\mathbb{R}^{n\times d}$ have condition number $\kappa$. 
Let $\beta\ge 1$, $k\ge 1$ be constants and $\lambda=\frac{\norm{\A-\A_{(k)}}^2_2}{\beta k}$. 
Then $\sum_{i=1}^n l_i=\O{k\log\kappa}$.
\end{lemma}
\begin{proof}
Since $\A^\top\A+\lambda\I \succeq \lambda \I$, let $\sigma_1\ge\sigma_2\ge\ldots\ge\sigma_d\ge\lambda$ be the singular values of $\A^\top\A+\lambda\I$. 
Then $\det(\A^\top\A+\lambda\I)=\prod_{i=1}^d\sigma_i$. 
Observe that $\sigma_{k+1}+\ldots+\sigma_d=\norm{\A-\A_{(k)}}^2_F+(d-k)\lambda$ by the Eckart-Young-Mirsky theorem. 
By the AM-GM inequality, we have 
\[\prod_{i=k+1}^d\sigma_i \leq \paren{\frac{\norm{\A-\A_{(k)}}^2_F+(d-k)\lambda}{d-k}}^{d-k}.\]
\noindent
Combining with the fact that $\sigma_i\le\norm{\A}_2^2+\lambda$ for $1\le i\le k$, we have for $\lambda=\frac{\norm{\A-\A_{(k)}}_F^2}{\beta k}$,
\begin{align*}
\det(\A^\top\A+\lambda\I) &=\prod_{i=1}^d\sigma_i \le(\norm{\A}_2^2+\lambda)^k\left(\frac{\norm{\A-\A_{(k)}}^2_F}{d-k}+\lambda\right)^{d-k}\\
&\le(\norm{\A}_2^2+\lambda)^k\lambda^{d-k}\left(\frac{\beta k}{d-k}+1\right)^{d-k}\le(2\norm{\A}_F^2)^k\lambda^{d-k}e^{\beta k}.
\end{align*}
Combining with \lemref{lem:det:lower} and taking logarithms, it follows that $d\log\lambda+\sum\frac{l_i}{2}\le k+2k\log2\norm{\A}_F+(d-k)\log\lambda+\beta k$. 
Since $\lambda=\frac{\norm{\A-\A_{(k)}}_F^2}{\beta k}$, then $\sum_{i} l_i\le 2(1+\beta)k+4k\log2\kappa$. 
\end{proof}
The sum of the reverse online $\lambda$-ridge leverage scores is bounded by the same quantity, since the rows of the input matrix can simply be considered in reverse order. 
We now show that \algref{alg:sampling:framework} using the $\score$ function of \algref{alg:pcp:score} gives a relative error low-rank approximation with efficient space usage. 

\begin{theorem}[Randomized Low-Rank Approximation Sliding Window Algorithm]
\thmlab{thm:sw:pcp}
Let $\r_1,\ldots,\r_n\in\mathbb{R}^{d}$ be a stream of rows and $\kappa$ be the condition number of the matrix $\r_1\circ\ldots\circ\r_n$. 
Let $W>0$ be a window size parameter and $\A=\r_{n-W+1}\circ\ldots\circ\r_n$ be the matrix consisting of the $W$ most recent rows. 
Given a parameter $\eps>0$, there exists an algorithm that outputs a matrix $\M$ that is a $(1+\eps)$ rank $k$ projection-cost preservation of $\A$ and stores $\O{\frac{k}{\eps^2}\log n\log^2\kappa}$ rows at any time, with high probability.
\end{theorem}
\begin{proof}
Let $\lambda=\frac{\norm{\A-\A_{(k)}}_F^2}{k}$. 
Then $(1-\eps)(\A^\top\A+\lambda\I) \preceq\M^\top\M+\lambda\I\preceq (1+\eps)(\A^\top\A+\lambda\I)$ holds immediately from \lemref{lem:pcp:downsample} using $t=n$ and $i=n-W+1$. 
Thus by \lemref{lem:addmult:pcp}, \lemref{lem:moderate:pcp}, and rescaling $\eps$, it follows that $\M$ is a $(1+\eps)$ rank $k$ projection-cost preservation of $\A$. 

To bound the number of sampled rows, observe that we can sample at most $k$ linearly independent rows before the regularization becomes nonzero. 
Then by \lemref{lem:pcp:downsample} each row $\r_i$ is sampled with probability at most $\min(1,2c\cdot\r_i(\Z_{i+1}^\top\Z_{i+1}+\lambda_i\I)^{-1}\r_i^\top)$, where $\Z_i=\r_i\circ\ldots\circ\r_n$ and $\lambda_i=\frac{\norm{\Z_{i+1}-(\Z_{i+1})_{(k)}}_F^2}{k}$ for each $i\in[n]$. 
For the purposes of analysis, set $u_0=n$ and for each $i\ge1$, let $u_i$ be the largest index $j$ such that $\lambda_j>2\lambda_{u_{i-1}}$. 
This partitions the stream into breakpoints $u_i$ starting from the most recent row, so that for each $u_{i+1}<t\le u_i$, we have $\lambda_{u_i}\le\lambda_t<2\lambda_{u_i}$. 
Thus by \lemref{lem:pcp:scores}, the sum of the reverse online ridge leverage scores for the rows $\r_t$ with $u_{i+1}<t\le u_i$ is $\O{k\log\kappa}$. 
\algref{alg:sampling:framework} scales the reverse online ridge leverage score by a factor of $\alpha=\O{\frac{1}{\eps^2}\log n}$ to determine the sampling probability, so by a coupling argument and standard Chernoff bounds, the number of sampled rows $\r_t$ with $u_{i+1}<t\le u_i$ is $\O{\frac{k}{\eps^2}\log n\log\kappa}$. 
Since there are $\O{\log\kappa}$ such breakpoints $u_i$, the total number of sampled rows is $\O{\frac{k}{\eps^2}\log n\log^2\kappa}$.
\end{proof}

\paragraph{Nearly Input Sparsity Runtime.}
Nearly input sparsity runtime can be achieved through similar batching and dimensionality reduction techniques as \secref{sec:l2:randomized}. 
Since \algref{alg:sampling:framework} using the $\score$ function of \algref{alg:pcp:score} stores $\O{\frac{k}{\eps^2}\log n\log^2\kappa}$ rows, the asymptotic space used by the algorithm will remain the same if it processes $\O{\frac{k}{\eps^2}\log n\log^2\kappa}$ rows at a time. 
To compute a constant factor approximation of the reverse online ridge leverage score of any sampled row $\a_i$, each of the $\O{\frac{k}{\eps^2}\log n\log^2\kappa}$ rows into a $\O{\log\frac{n}{\eps}}$ dimension subspace by applying a Johnson-Lindenstrauss transform~\cite{SpielmanS11, CohenMP16, CohenMM17}. 
Applying this subspace embedding to each of the rows takes $\O{\log\frac{n}{\eps}\cdot\overline{\nnz}}$ time, where $\overline{\nnz}$ is the input sparsity of the batch. 
Subsequently, all operations are in $\O{\log\frac{n}{\eps}}$ dimension subspace, so approximating the each regularization parameter $\lambda$ and each reverse online leverage score requires $\polylog\frac{n}{\eps}$ time for each of the $\O{\frac{k}{\eps^2}\log n\log^2\kappa}$ scores. 
Hence, the total runtime is $\polylog\frac{n}{\eps}\cdot\nnz$, where $\nnz$ is the input sparsity of the stream. 
Again the algorithm need not know $\kappa$ in advance since it suffices to process a batch when the number of newly arrived rows in the batch equals the number of stored rows after the process.

\paragraph*{Bounded Precision and Condition Number.}
In typical applications, the entries of the underlying matrix will have some bounded precision, say $\O{\log n}$ bits. 
From invariance through scaling the entries, we assume that the entries of the matrix are integers whose magnitudes are bounded by $\poly(n)$. 
Nevertheless, the condition number $\kappa$ of the underlying matrix can be as large as $\poly(n)^d$, which can cause the $\log\kappa$ term to incur an additional factor of $\O{\d\log n}$. 
We show this is not necessary by reducing the $\log\kappa$ factor down to $\log n$ in some cases and removing the $\log\kappa$ factor altogether in other cases. 
We first relate matrices $\A$ with bounded entries and sufficiently large rank with a bound on $\frac{\norm{\A}_F}{\norm{\A-\A_{(k)}}_F}$. 
\begin{lemma}
\cite{ClarksonW09}
\lemlab{lem:rank:condition}
If $\A\in\mathbb{R}^{n\times d}$ has integer entries bounded in magnitude by $\poly(n)$ and rank at least $2k$, then $\frac{\norm{\A}_F}{\norm{\A-\A_{(k)}}_F}\le\poly(n,d)$ as $nd\to\infty$. 
\end{lemma}
By \lemref{lem:rank:condition}, we have $\O{\log\kappa}=\O{\log n}$ if $\rank(\A)\ge 2k$. 
Moreover, we do not incur the extra $\O{\log n}$ factor from Chernoff bounds to achieve high probability over all times in the stream, so the total number of sampled rows is $\O{\frac{k}{\eps^2}\log^2 n}$. 
Thus it suffices to maintain a rank $k$ projection-cost preservation in the case that $\rank(\A)\le 2k$. 

Observe that we can decompose $\A^\top\A=\R^\top\U\R$, where $\R\in\mathbb{R}^{2k\times d}$ is the row span of $\A$.  
Now the row span of any suffix of $\A$ will be contained within $\R$. 
Thus it suffices to maintain a sequence of matrices $\U_1,\U_2,\ldots\in\mathbb{R}^{2k\times 2k}$ so that $\R^\top\U_1\R,\R^\top\U_2\R,\ldots$ provides a $(1+\eps)$ spectral approximation to all suffixes of $\A$.  
We give the algorithm in \algref{alg:pcp:low} and the guarantees in \thmref{thm:pcp:low}, but defer the full proof to \appref{app:histogram}, as it is a special case of a more general data structure that we introduce for deterministic algorithms for numerical linear algebra in the sliding window model. 
\begin{algorithm}[!htb]
\caption{Projection-cost preservation for low-rank matrices in the sliding window model}
\alglab{alg:pcp:low}
\begin{algorithmic}[1]
\Require{A stream of rows $\r_1,\ldots,\r_n\in\mathbb{R}^d$, window size $W$, and an accuracy parameter $\eps>0$}
\Ensure{Rank $k$ projection-cost preservation in the sliding window model.}
\State{$\M_0\gets 0^{d\times d}$}
\For{each row $\r_t$}
\State{Suppose $\M_0,\M_1,\ldots,\M_s$ are defined}
\State{$\R_{s+1}\gets\r_t$, $\M_{s+1}\gets\R_{s+1}^\top\R_{s+1}$}
\Comment{Keep timestamp and decomposition for $\M_{s+1}$}
\For{$i=s$ to $i=1$}
\Comment{Update sketches with $\r_t$}
\State{Let $\M_i=\R_i^\top\U_i\R_i$ and $\B_i$ be a basis for $\R_i\setminus\R_{i+1}$.}
\State{$\R_i\gets\R_{i+1}\circ\B_i$}
\Comment{Ensures $\R_{i+1}\subseteq\R_i$}
\State{Set $\U_i$ so that $\R_i^\top\U_i\R_i=\M_i+\r_t^\top\r_t$.}
\EndFor
\For{$i=s$ to $i=2$}
\If{$\M_{i-1}^\top\M_{i-1}\preceq(1+\eps)\M_{i+1}^\top\M_{i+1}$}
\State{Delete $\M_i$ and relabel indices}
\EndIf
\EndFor
\If{$\M_2$ has expired}
\Comment{Timestamp is at most $t-W+1$}
\State{Delete $\M_1$ and relabel indices}
\EndIf
\EndFor
\State{\Return $\M_1$}
\end{algorithmic}
\end{algorithm}
The key observation is that a new matrix $\U_i$ is needed only when some singular value of $\A$ increase by $(1+\eps)$. 
Thus the number of matrices $\U_i$ that are needed is $\O{\frac{1}{\eps}\log\kappa}$, where $\kappa$ is the condition number of $\A$. 
Since the entries of $\A$ are bounded integers, then the characteristic polynomial of $\A^\top\A$ has integer coefficients bounded by $(\poly(n))^{\O{k}}$. 
Similarly, the largest eigenvalue of $\A^\top\A$ is bounded by its Frobenius norm, which is bounded by $\poly(n)$, so then $\log\kappa=\O{k\log n}$.
As each matrix $\U_i$ has dimension $2k\times 2k$, the space bound follows. 
\newcommand{\thmpcplow}{There exists a deterministic algorithm in the sliding window model that outputs a rank $k$ projection-cost preservation of an input matrix $\A$ whose rank is at most $2k$. 
If the entries of $\A$ are integers that are bounded in magnitude by $\poly(\A)$, then this algorithm stores $\O{k}$ rows of $\A$ and uses $\O{\frac{k^4}{\eps}\log n}$ additional words of space.}
\begin{theorem}
\thmlab{thm:pcp:low}
\thmpcplow
\end{theorem}
Given \thmref{thm:pcp:low} to handle $\rank(\A)\le 2k$ and \lemref{lem:rank:condition} to handle $\rank(\A)\ge 2k$ along with \thmref{thm:sw:pcp} and the previous observation that we no longer need to incur an extra $\O{\log n}$ factor for Chernoff bounds to achieve high probability over all times in the stream, then we have the following guarantees for low-rank approximation/projection-cost preservation when the underlying matrix has integer entries that are bounded in magnitude by $\poly(n)$. 
\begin{corollary}
\thmlab{thm:sw:pcp:bounded}
Let $\r_1,\ldots,\r_n\in\mathbb{R}^{d}$ be a stream of rows of integers whose magnitudes are at most $\poly(n)$.  
Let $W>0$ be a window size parameter and $\A=\r_{n-W+1}\circ\ldots\circ\r_n$ be the matrix consisting of the $W$ most recent rows. 
Given a parameter $\eps>0$, there exists an algorithm that with high probability, outputs a matrix $\M$ that is a $(1+\eps)$ rank $k$ projection-cost preservation of $\A$, stores $\O{\frac{k}{\eps^2}\log^2 n}$ rows at any time, and uses $\O{\frac{k^4}{\eps}\log n}$ additional words of space. 
\end{corollary}
\section{Simple Rank Constrained Algorithms in the Online Model}
\seclab{sec:online}
In this section, we show that the paradigm of row sampling with respect to online ridge leverage scores offers simple analysis for a number of online algorithms that improve upon the state-of-the-art. 

\subsection{Online Projection-Cost Preservation}
\seclab{sec:online:pcp}
As a warm-up, we first demonstrate how our previous analysis bounding the sum of the online ridge leverage scores can be applied to analyze a natural online algorithm for producing a projection-cost preservation of a matrix $\A=\a_1\circ\ldots\circ\a_n\in\mathbb{R}^{n\times d}$.  
Our algorithm samples each row with probability equal to the online ridge leverage scores, where the regularization parameter $\lambda_i$ is computed at each step. 
Note that if $\A_i=\a_1\circ\ldots\circ\a_i$, then $\norm{\A_i-(\A_i)_{(k)}}_F^2\le\norm{\A_j-(\A_j)_{(k)}}_F^2$ for any $i<j$. 
Thus if $\lambda_i=\frac{\norm{\A_{i-1}-(\A_{i-1})_{(k)}}_F^2}{k}$, then sampling row $\a_i$ with online ridge leverage score regularized by $\lambda_i$ has a \emph{higher} probability than with ridge leverage score regularized by $\lambda=\frac{\norm{\A-\A_{(k)}}_F^2}{k}$. 
Although~\cite{CohenMP16} is only interested in spectral approximation and therefore set $\lambda=\eps\sigma_{\min}(\A)$, they nevertheless shows that online row sampling with any regularization $\lambda$ gives an additive-multiplicative spectral approximation to $\A$. 
Thus by setting $\lambda=\frac{\norm{\A-\A_{(k)}}_F^2}{k}$, our algorithm outputs a rank $k$ projection-cost preservation of $\A$ by \lemref{lem:addmult:pcp} and \lemref{lem:moderate:pcp}. 
Moreover, our bounds for the sum of the online ridge leverage scores in \lemref{lem:pcp:scores} show that our algorithm only samples a small number of rows, optimal up to lower order factors. 

\begin{algorithm}[!htb]
\caption{$\onlinepcp$: Online algorithm for projection-cost preservation}
\alglab{alg:online:pcp}
\begin{algorithmic}[1]
\Require{Stream of rows $\a_1,\ldots,\a_n\in\mathbb{R}^{1\times d}$, accuracy $\eps>0$, and parameter $k$.}
\Ensure{Rank $k$ projection-cost preservation of $\A:=\a_1\circ\ldots\circ\a_n$.}
\State{$\M\gets\emptyset$.}
\State{$\alpha\gets\frac{C}{\eps^2}\log n$ with sufficiently large constant $C>0$}
\For{each row $\a_t$}
\State{$\tilde{\lambda}_t\gets\frac{\norm{\M-\M_{(k)}}_F^2}{2k}$}
\State{$\tau_t\gets 2\a_t(\M^\top\M+\tilde{\lambda}_t\I)^{-1}\a_t^\top$}
\State{$p_t\gets\min(1,\alpha\tau_t)$}
\State{With probability $p_t$, $\M\gets\M\circ\frac{\a_t}{\sqrt{p_t}}$}
\EndFor
\State{\Return $\M$.}
\end{algorithmic}
\end{algorithm}

\begin{theorem}[Online Rank $k$ Projection-Cost Preservation]
\thmlab{thm:online:pcp}
Given parameters $\eps>0$, $k>0$, and a matrix $\A\in\mathbb{R}^{n\times d}$ whose rows $\a_1,\ldots,\a_n$ arrive sequentially in a stream with condition number $\kappa$, there exists an online algorithm that outputs a matrix $\M$ with $\O{\frac{k}{\eps^2}\log n\log^2\kappa}$ (rescaled) rows of $\A$ such that
\[(1-\eps)\norm{\A-\A_{(k)}}_F^2\le\norm{\M-\M_{(k)}}_F^2\le(1+\eps)\norm{\A-\A_{(k)}}_F^2,\]
and thus $\M$ is a rank $k$ projection-cost preservation of $\A$, with high probability. 
\end{theorem}
\begin{proof}
Consider \algref{alg:online:pcp}. 
For all $i\in[n]$, let $\M_i$ denote the matrix $\M$ at time $i$, which is used to compute $\lambda_{i+1}$. 
Let $\lambda=\frac{\norm{\A-\A_{(k)}}_F^2}{k}$. 
\cite{CohenMP16} show that sampling row $\a_i$ with probability at least $\min(1,\alpha\a_i(\M_{i-1}^\top\M_{i-1}+\lambda\I)^{-1}\a_i^\top)$ suffices to give a matrix $\M$ such that $(1-\eps)(\A^\top\A+\lambda\I)\preceq\M^\top\M+\lambda\I\preceq(1+\eps)(\A^\top\A+\lambda\I)$ with high probability. 
This argument also follows from the matrix martingale analysis in \lemref{lem:l2:downsample}, which is based on their argument. 
Since we sample each row $\a_i$ with probability $\min(1,2\alpha\a_i(\M_{i-1}^\top\M_{i-1}+\tilde{\lambda}_i\I)^{-1}\a_i^\top)$ instead, it suffices to show that $\tilde{\lambda}_i\le 2\lambda$. 
But this must hold since by a union bound, $(1-\eps)(\A_{i-1}^\top\A_{i-1}+\lambda_i\I)\preceq\M_{i-1}^\top\M_{i-1}+\lambda_i\I\preceq(1+\eps)(\A_{i-1}^\top\A_{i-1}+\lambda_i\I)$ for all $i\in[n]$. 
In particular, this implies by \lemref{lem:addmult:pcp} and \lemref{lem:moderate:pcp} that $\M_i$ is a rank $k$ projection-cost preservation for $\A_i$ and thus, \[\tilde{\lambda}_i=\frac{\norm{\M_{i-1}-(\M_{i-1})_{(k)}}_F^2}{k}\le\frac{2\norm{\A_{i-1}-(\A_{i-1})_{(k)}}_F^2}{k}\le\frac{2\norm{\A-\A_{(k)}}_F^2}{k}=2\lambda.\] 

To bound the number of rows sampled by $\M$, we use a similar approach to that in \thmref{thm:sw:spectral}. 
Let $u_0=1$ and for each $i\ge 1$, let $u_i$ be the smallest index $j$ such that $\lambda_j>2\lambda_{u_{i-1}}$, which partitions the stream into breakpoints $u_i$ starting from the beginning of the stream. 
For each $u_{i+1}<t\le u_i$, we have $\lambda_{u_i}\le\lambda_t<2\lambda_{u_i}$. 
Thus by \lemref{lem:pcp:scores}, the sum of the online ridge leverage scores for the rows $\a_t$ with $u_i\le t<u_{i+1}$ is $\O{k\log\kappa}$. 
\algref{alg:online:pcp} scales the online ridge leverage score by a factor of $\alpha=\O{\frac{1}{\eps^2}\log n}$ to determine the sampling probability, so by a coupling argument and standard Chernoff bounds, the number of sampled rows $\a_t$ in $\M$ with $u_i\le t<u_{i+1}$ is $\O{\frac{k}{\eps^2}\log n\log\kappa}$. 
Since there are $\O{\log\kappa}$ such breakpoints $u_i$, the total number of rows in $\M$ is $\O{\frac{k}{\eps^2}\log n\log^2\kappa}$.
Note that here we partition the stream from the beginning since we use online ridge leverage scores, whereas the argument in \thmref{thm:sw:spectral} forms the breakpoints beginning from the end of the stream since it considers \emph{reverse} online ridge leverage scores. 
\end{proof}
Our algorithm uses space optimal up to lower order terms for low-rank approximation, but we note that more efficient algorithms exist if the objective is only to approximate $\norm{\A-\A_{(k)}}_F^2$, rather than determining the entire projection matrix.  
For example, the Theorem 1.3 in \cite{AndoniN13} uses space roughly $\O{k^2}$ whereas \algref{alg:online:pcp} uses space roughly $\O{kd}$ to store each of the $\tO{k}$ sampled rows.  

\subsection{Online Row Subset Selection}
\seclab{sec:online:rss}
We next describe how to perform row subset selection in the online model. 
The algorithm and analysis are both relatively simple, but the algorithm only stores $\O{\frac{k}{\eps}\log n\log^2\kappa}$ rows. 
By comparison, the recent online row subset selection algorithm of~\cite{BhaskaraLVZ19} stores $\O{\frac{k}{\eps^2}\log n\log^2\kappa}$ rows to succeed with high probability. 
Moreover, our algorithm provides the stronger guarantee of the existence of a subset of $k$ rows that provides a $(1+\eps)$ factor approximation to the best rank $k$ solution. 

Our starting point is an offline algorithm by \cite{CohenMM17} and the paradigm of adaptive sampling~\cite{DeshpandeV06, DeshpandeRVW06, MahabadiRWZ20}, which is the procedure of repeatedly sampling rows of $\A$ with probability proportional to their squared distances to the subspace spanned by $\Z$. 
\cite{CohenMM17} first obtains a matrix $\Z$ that is a constant factor low-rank approximation to the underlying matrix $\A$ through ridge-leverage score sampling. 
They observe that a theorem by \cite{DeshpandeRVW06} shows that adaptive sampling $\O{\frac{k}{\eps}}$ additional rows $\S$ of $\A$ against the rows of $\Z$ suffices for $\Z\cup\S$ to contain a $(1+\eps)$ factor approximation to the online row subset selection problem. 

\cite{CohenMM17} then adapts this approach to the streaming model by maintaining a reservoir of $\O{\frac{k}{\eps}}$ rows, and replacing rows appropriately as new rows arrive and more information about $\Z$ is obtained. 
Alternatively, we modify the proof of \cite{DeshpandeRVW06} if $\Z$ is given, adaptive sampling can also be performed on data streams to obtain a different but valid $\S$ by oversampling each row of $\A$ by a $\O{\frac{k}{\eps}}$ factor. 
Moreover by running a low-rank approximation algorithm in parallel, downsampling can be performed as rows of $\Z$ arrive, so that the above approach can be performed in one stream. 

In the online model, we cannot downsample rows of $\S$ once they are selected, since $\Z$ evolves as the stream arrives.  
Fortunately, \cite{CohenMM17} shows that the adaptive sampling probabilities can be upper bounded by the $\lambda$-ridge leverage scores, where $\lambda=\frac{1}{k}\norm{\A-\A_{(k)}}_F^2$. 
Since the $\lambda$-ridge leverage scores are at most the online $\lambda$-ridge leverage scores, we can again sample rows proportional to their online $\lambda$-ridge leverage scores. 
It then suffices to again use \lemref{lem:pcp:scores} to bound the number of rows sampled in this manner. 

We first show an analog for Theorem 2.1 in \cite{DeshpandeRVW06} for adaptive sampling through oversampling each row by a $\O{\frac{k}{\eps}}$ factor. 
The proof is almost verbatim, but uses a different estimator since we do not perform offline sampling with replacement. 
As before, we use $\A_{(k)}$ to denote the best rank $k$ approximation to $\A$ in the Frobenius norm. 
\begin{theorem}
\thmlab{thm:adaptive:sampling}
Let $\A\in\mathbb{R}^{n\times d}$, $\Z$ be a set of vectors in $\mathbb{R}^d$, $k$ be a non-negative integer, and $\eps>0$ be a given constant.  
Let $\S$ be a matrix formed by selecting each row $i$ of $\A$ independently with probability at least
\[p_i=\min\left(1,\frac{k}{\eps}\frac{\norm{(\A-\A\Z^{-1}\Z)_i}_2^2}{\sum_{j=1}^n\norm{(\A-\A\Z^{-1}\Z)_j}_2^2}\right).\]
Then there exists a matrix $\T$ that contains $k$ rows of $\Z\cup\S$ such that
\[\EEx{\norm{\A-\A\T^{-1}\T}_F^2}\le\norm{\A-\A_{(k)}}_F^2+\eps\norm{\A-\A\Z^{-1}\Z}_F^2.\]
\end{theorem}
\begin{proof}
Define matrices $\U$ and $\V$ to be the left and right singular vectors of $\A$ so that $\A=\U\Sigma\V$. 
Let $\u^{(i)}\in\mathbb{R}^{1\times n}$ denote the transpose of the $i\th$ column of $\U$ and $\v^{(i)}\in\mathbb{R}^{1\times d}$ denote the $i\th$ row of $\V$. 
Denote the singular values of $\A$ by $\sigma_1\ge\sigma_2\ge\ldots$. 
We show there exist $\t_1,\ldots,\t_k$ in the span of $\Z\cup\T$ that are a good approximation of the span of top right singular vectors $\v^{(1)},\ldots,\v^{(k)}$. 

For $i\in[n]$ and $j\in[k]$, let $\x_i^{(j)}=\frac{\u_i^{(j)}}{p_i}(\A-\A\Z^{-1}\Z)_i$ with probability $p_i$ and the all zeros row vector otherwise so that $\x_i^{(j)}\in\mathbb{R}^{1\times d}$.  
Let $\x_j=\sum_{i=1}^n\x_i^{(j)}$ so that $\EEx{\x_j}=\u^{(j)}(\A-\A\Z^{-1}\Z)$. 
For each $j\in[k]$, we define each $\t_j\in\mathbb{R}^{1\times d}$ by $\t_j:=\u^{(j)}(\A\Z^{-1}\Z)+\x_j$ so that $\EEx{\t_j}=\sigma_j\v^{(j)}$. 

We first bound the variance of $\t_j$: 
\begin{align*}
\EEx{\norm{\t_j-\EEx{\t_j}}_2^2}&=\EEx{\norm{\t_j-\sigma_j\v^{(j)}}_2^2}\\
&=\EEx{\norm{\x_j-\u^{(j)}(\A-\A\Z^{-1}\Z)}_2^2}\\
&=\EEx{\norm{\x_j}_2^2}-2\EEx{\x_j}\cdot\u^{(j)}(\A-\A\Z^{-1}\Z)+\norm{\u^{(j)}(\A-\A\Z^{-1}\Z)}_2^2\\
&=\EEx{\norm{\x_j}_2^2}-\norm{\u^{(j)}(\A-\A\Z^{-1}\Z)}_2^2.
\end{align*}

To bound the expected squared norm of $\x_j$,
\begin{align*}
\EEx{\norm{\x_j}_2^2}&=\EEx{\norm{\sum_{i=1}^n\x_i^{(j)}}_2^2}=\sum_{i=1}^n\EEx{\norm{\x_i^{(j)}}_2^2}+\sum_{1\le i_1\neq i_2\le n}\EEx{x_{i_1}^{(j)}\cdot x_{i_2}^{(j)}}\\
&\le\sum_{i=1}^n\EEx{\norm{\x_i^{(j)}}_2^2}+\norm{\u^{(j)}(\A-\A\Z^{-1}\Z)}_2^2,
\end{align*}
where the last step holds since $\x_{i_1}^{(j)}$ and $\x_{i_2}^{(j)}$ are independent and $\EEx{\x_j}=\u_{(j)}(\A-\A\Z^{-1}\Z)$.

Thus by our choice of $p_i$,
\begin{align*}
\EEx{\norm{\t_j-\sigma_j\v^{(j)}}_2^2}&\le\sum_{i=1}^n\EEx{\norm{\x_i^{(j)}}_2^2}=\sum_{i=1}^n\frac{\norm{\u_i^{(j)}(\A-\A\Z^{-1}\Z)_i}_2^2}{p_i}\\
&\le\sum_{i=1}^n\frac{\eps}{k}\frac{\norm{\u_i^{(j)}(\A-\A\Z^{-1}\Z)_i}_2^2\norm{\A-\A\Z^{-1}\Z}_F^2}{\norm{(\A-\A\Z^{-1}\Z)_i}_2^2}\\
&=\frac{\eps}{k}\norm{\A-\A\Z^{-1}\Z}_F^2
\end{align*}
Define the matrix $\M=\A\sum_{j=1}^k\frac{1}{\sigma_j}(\v^{(j)})^\top\t_j$, so that the row space of $\M$ is spanned by the vectors $\t_1,\ldots,\t_k$. 
Let $\T$ be the matrix whose rows are $\t_1,\ldots,\t_k$, so that $\norm{\A-\A\T^{-1}\T}_F^2\le\norm{\A-\M}_F^2$. 
We now bound the expected value of the Frobenius norm of the difference. 
Since the Frobenius norm is invariant under rotations,
\begin{align*}
\EEx{\norm{\A-\A\T^{-1}\T}_F^2}&\le\EEx{\norm{\A-\M}_F^2}=\sum_{j=1}^{\rank(\A)}\EEx{\norm{\u_j(\A-\M)}_F^2}\\
&=\sum_{j=1}^k\EEx{\norm{\sigma_j\v^{(j)}-\t_j}_F^2}+\sum_{j=k+1}^{\rank(\A)}\sigma_j^2\le\eps\norm{\A-\A\Z^{-1}\Z}_F^2+\norm{\A-\A_{(k)}}_F^2.
\end{align*}
\end{proof}
By considering a matrix $\Z$ that is a constant $c$ factor low-rank approximation to $\A$, then we may choose the accuracy parameter in the statement of \thmref{thm:adaptive:sampling} to be $\frac{\eps}{c}$ so that $\frac{\eps}{c}\norm{\A-\A\Z^{-1}\Z}_F^2=\eps\norm{\A-\A_{(k)}}_F^2$. 
Moreover as \cite{CohenMM17} notes, the expectation can be converted to a high probability bound through $\log\frac{1}{\delta}$ repetitions and Markov's inequality. 
\begin{corollary}
\corlab{cor:adaptive:sampling}
Let $\A\in\mathbb{R}^{n\times d}$, $k$ be a non-negative integer, and $\eps,\delta>0$ be given constants. 
Let $\Z$ be a constant factor low-rank approximation to $\A$. 
Then there exists a constant $C$ such that sampling each row $i$ of $\A$ independently with probability at least
\[p_i=\frac{Ck\log(1/\delta)}{\eps}\frac{\norm{(\A-\A\Z^{-1}\Z)_i}_2^2}{\sum_{j=1}^n\norm{(\A-\A\Z^{-1}\Z)_j}_2^2}\]
forms a matrix $\S$ such that $\Z\cup\S$ contains a matrix $\T$ of $k$ rows such that
\[\PPr{\norm{\A-\A\T^{-1}\T}_F^2\le(1+\eps)\norm{\A-\A_{(k)}}_F^2}\ge 1-\delta.\]
\end{corollary}
\noindent
Before describing our algorithm, we need the following relationship between ridge leverage scores and adaptive sampling, shown by \cite{CohenMM17}. 
\begin{lemma}
\cite{CohenMM17}
\lemlab{lem:ridge:adaptive}
Let $\Z$ be a constant factor approximation to the best rank $k$ approximation to $\A$. 
There exists a universal constant $\gamma$ such that if $\tau_i(\A)$ is the $i\th$ ridge leverage score of a matrix $\A$, with regularization $\frac{\norm{\A-\A_{(k)}}_F^2}{k}$, then $\frac{\gamma\norm{\A-\A_{(k)}}_F^2}{k}\tau_i(\A)\ge\norm{(\A-\A\Z^{-1}\Z)_i}_2^2$.
\end{lemma}

\begin{algorithm}[!htb]
\caption{$\onlinerss$: Online algorithm for row subset selection}
\alglab{alg:online:rss}
\begin{algorithmic}[1]
\Require{Stream of rows $\a_1,\ldots,\a_n\in\mathbb{R}^{1\times d}$, accuracy $\eps\le\frac{1}{2}$, and parameter $k$.}
\Ensure{$(1+\eps)$-approximate row subset selection of $\A:=\a_1\circ\ldots\circ\a_n$.}
\State{Run a $2$-approximation $\onlinepcp$ for online low-rank approximation of $\A$ in parallel.}
\Comment{\algref{alg:online:pcp}}
\State{$\alpha\gets\frac{C}{\eps}\log n$ for sufficiently large $C>0$}
\State{$\S\gets\emptyset$}
\For{each row $\a_t$}
\State{Update \onlinepcp{} with $\a_t$.}
\State{Let $\Z_t$ be the rows stored by $\onlinepcp$.}
\State{$\tilde{\lambda}_t\gets\frac{\norm{\Z_t-(\Z_t)_{(k)}}_F^2}{2k}$}
\State{$p_t\gets\min(1,2\alpha\cdot\a_t(\Z_t^\top\Z_t+\tilde{\lambda}_t\I)^{-1}\a_t^\top)$}
\State{With probability $p_t$, $\S\gets\S\circ\a_t$}
\Comment{Independently sampled}
\EndFor
\State{Let $\Z$ be the output of $\onlinepcp$.}
\State{\Return $\Z\cup\S$.}
\end{algorithmic}
\end{algorithm}

\noindent
We now show the correctness and space bounds for \onlinerss.
\begin{theorem}[Online Row Subset Selection]
\thmlab{thm:online:rss}
Given parameters $0<\eps<\frac{1}{2}$, $k>0$, and a matrix $\A\in\mathbb{R}^{n\times d}$ whose rows $\a_1,\ldots,\a_n$ arrive sequentially in a stream with condition number $\kappa$, there exists an online algorithm that outputs a matrix $\M$ with $\O{\frac{k}{\eps}\log n\log^2\kappa}$ rows that contains a matrix $\T$ of $k$ rows such that 
\[\norm{\A-\A\T^{-1}\T}_F^2\le(1+\eps)\norm{\A-\A_{(k)}}_F^2,\]
with high probability. 
\end{theorem}
\begin{proof}
Consider $\onlinerss$ (\algref{alg:online:rss}). 
We set $\M$ as the output $\Z\cup\S$ from \onlinerss. 
Let $\tau_i$ denote the ridge leverage score of $\a_i$, with regularization $\frac{\norm{\A-\A_{(k)}}_F^2}{k}$. 
By choosing \onlinepcp{} to be a $2$-approximation to the best low-rank approximation to $\A$, we have that $\Z$ contains a rank $k$ matrix $\W\in\mathbb{R}^{k\times d}$ such that $\norm{\A-\A_{(k)}}_F^2\le\norm{\A-\A\W^{-1}\W}_F^2\le2\norm{\A-\A_{(k)}}_F^2$.
We condition on \onlinepcp{} to return such a matrix $\Z$ with high probability.  

By \lemref{lem:ridge:adaptive}, it follows that 
\[\frac{\alpha k\log n}{\eps}\frac{\norm{(\A-\A\W^{-1}\W)_i}_2^2}{\norm{\A-\A\W^{-1}\W}_F^2}\le\frac{\alpha\gamma\log n}{\eps}\tau_i(\A).\] 
Observe that each row $\a_i$ of $\A$ is sampled with probability $p_i$, which is a $2$-approximation to the online $\lambda$-ridge leverage score. 
Online ridge leverage scores are at least as large as ridge leverage scores, so $p_i\ge\frac{\alpha}{\eps}\tau_i$. 
Thus for sufficiently large $\alpha$, $p_i\ge\frac{\alpha k\log n}{\eps}\frac{\norm{(\A-\A\W^{-1}\W)_i}_2^2}{\norm{\A-\A\W^{-1}\W}_F^2}$. 
Moreover, the sampling probabilities $p_i$ to determine whether a row should be added to $\S$ depend only on $\onlinepcp$ and not on whether previous rows were added to $\S$. 
Hence, the sampling probabilities are independent and by \corref{cor:adaptive:sampling}, it follows that $\W\cup\S$ contains a matrix $\T$ of $k$ rows such that $\norm{\A-\A\T^{-1}\T}_F^2\le(1+\eps)\norm{\A-\A_{(k)}}_F^2$, with high probability. 
Since $\W$ is a subset of $\Z$, then $\T$ is also contained in $\Z\cup\S$. 

It remains to bound the number of rows in $\Z\cup\S$. 
Since $\Z$ uses \onlinepcp{} with a constant factor approximation, it stores $\O{k\log n\log^2\kappa}$ rows by \thmref{thm:online:pcp}. 
To bound the number of rows sampled by $\S$, we use the same idea from \thmref{thm:online:pcp}. 
Define $\lambda_i=\frac{\norm{\A_i-(\A_i)_{(k)}}_F^2}{k}$, where $\A_i=\a_1\circ\ldots\circ\a_i$, for each $i\in[n]$. 
We set $u_0=1$ and for each $i\ge 1$, let $u_i$ be the smallest index $j$ such that $\lambda_j>2\lambda_{u_{i-1}}$, which partitions the stream into breakpoints $u_i$ starting from the beginning of the stream. 
For each $u_{i+1}<t\le u_i$, we have $\lambda_{u_i}\le\lambda_t<2\lambda_{u_i}$ so the sum of the online ridge leverage scores for the rows $\a_t$ with $u_i\le t<u_{i+1}$ is $\O{k\log\kappa}$ by \lemref{lem:pcp:scores}.   
$\onlinerss$ scales the online ridge leverage score by a factor of $\alpha=\O{\frac{1}{\eps}\log n}$ to determine the sampling probability into $\S$, so by a coupling argument and standard Chernoff bounds, the number of sampled rows $\a_t$ in $\S$ with $u_i\le t<u_{i+1}$ is $\O{\frac{k}{\eps}\log n\log\kappa}$. 
Since there are $\O{\log\kappa}$ such breakpoints $u_i$, the total number of rows in $\S$ is $\O{\frac{k}{\eps}\log n\log^2\kappa}$.
Hence, the total number of rows in $\Z\cup\S$ is $\O{\frac{k}{\eps}\log n\log^2\kappa}$ with high probability. 
\end{proof}

\subsection{Online Principal Component Analysis}
\seclab{sec:online:pca}
Recall that in the online PCA problem, rows of the matrix $\A=\a_1\circ\ldots\circ\a_n\in\mathbb{R}^{n\times d}$ arrive sequentially in a data stream and after each row $\a_i$ arrives, the goal is to immediately output a row $\m_i\in\mathbb{R}^{1\times m}$ such that at the end of the stream, there exists a low-rank matrix $\X\in\mathbb{R}^{m\times d}$ such that
\[\norm{\A-\M\X}_F^2\le(1+\eps)\norm{\A-\A_{(k)}}_F^2,\]
where $\M=\m_1\circ\ldots\circ\m_n$ and $\A_{(k)}$ is again the best rank $k$ approximation to $\A$. 

\cite{BhaskaraLVZ19} gives an algorithm for the online PCA problem that embeds into a matrix $\X$ that has rank $m=\O{\frac{k}{\eps^2}(\log n+\log\kappa)^4}$ with high probability, where $\kappa$ is the condition number of $\A$. 
Their algorithm maintains and updates the matrix $\X$ throughout the data stream. 
After the arrival of row $\a_i$, the matrix $\X$ is updated using a combination of residual based sampling and a black-box theorem of \cite{BoutsidisGKL15}. 
Row $\m_i$ is then output as the embedding of $\a_i$ into $\X$ by $\m_i=\a_i\X^{(i)}$, where $\X^{(i)}$ is the matrix $\X$ after row $i$ has been processed by $\X$. 
$\X$ has the property that no rows from $\X$ are ever removed across the duration of the algorithm, so then $\m_i$ is only an upper bound on the best embedding of $\a_i$. 
However, this matrix $\X$ does that contain a good rank $k$ approximation to $\A$. 
That is, $\X$ does not contain a rank $k$ submatrix $\Y\in\mathbb{R}^{k\times d}$ such that there exists a matrix $\B$ such that
 \[\norm{\A-\B\Y}_F^2\le(1+\eps)\norm{\A-\A_{(k)}}_F^2.\]

Recall that our online row subset selection algorithm $\onlinerss$ (\algref{alg:online:rss}) returns a matrix $\Z\in\mathbb{R}$ with $\O{\frac{k}{\eps}\log n\log^2\kappa}$ rows of $\A$ that contains a submatrix $\Y$ of $k$ rows that is a good rank $k$ approximation of $\A$. 
Thus, our algorithm for online row subset selection algorithm can be combined with the algorithm of \cite{BhaskaraLVZ19} to output matrices $\M$ and $\W$ that is a good approximation to the online PCA problem but also so that $\W$ contains a submatrix $\Y$ of $k$ rows that is a good rank $k$ approximation of $\A$. 
Namely, the algorithm of \cite{BhaskaraLVZ19} can be run to produce a matrix $\X^{(i)}$ after the arrival of each row $\a_i$. 
Moreover, our online row subset selection algorithm \algref{alg:online:rss} can be run to produce a matrix $\Z^{(i)}$ after the arrival of each row $\a_i$. 
Let $\W^{(0)}=\mathbb{R}^{0\times d}$ and let $\W^{(i)}$ append the newly added rows of $\X^{(i)}$ and $\Z^{(i)}$ to $\W^{(i-1)}$. 
We then immediately output the embedding $\m_i=\a_i\W^{(i)}$.

\begin{theorem}[Online Principal Component Analysis]
\thmlab{thm:online:pca}
Given parameters $n,d,k,\eps>0$ and a matrix $\A\in\mathbb{R}^{n\times d}$ whose rows arrive sequentially in a stream with condition number $\kappa$, let $m=\O{\frac{k}{\eps^2}(\log n+\log\kappa)^4}$.
There exists an algorithm for online PCA that immediately outputs a row $\m_i\in\mathbb{R}^{m}$ after seeing row $\a_i\in\mathbb{R}^{d}$ and outputs a matrix $\X\in\mathbb{R}^{m\times d}$ at the end of the stream such that
\[\norm{\A-\M\X}_F^2\le(1+\eps)\norm{\A-\A_{(k)}}_F^2,\]
where $\A_{(k)}$ is the best rank $k$ approximation to $\A$. 
Moreover, $\X$ contains a submatrix $\Y\in\mathbb{R}^{k\times d}$ such that there exists a matrix $\B$ such that
\[\norm{\A-\B\Y}_F^2\le(1+\eps)\norm{\A-\A_{(k)}}_F^2.\]
\end{theorem}
\begin{proof}
Since $\W$ contains the matrix $\X$ from the algorithm of \cite{BhaskaraLVZ19}, it immediately follows that $\norm{\A-\M\W}_F^2\le(1+\eps\norm{\A-\A_{(k)}}_F^2$.
Moreover, since $\X$ contains the matrix $\Z$ from our online row subset selection algorithm, which also never removes any rows of $\Z$ once they are added, then $\X$ contains a rank $k$ submatrix $\Y\in\mathbb{R}^{k\times d}$ such that there exists a matrix $\B$ such that $\norm{\A-\B\Y}_F^2\le(1+\eps)\norm{\A-\A_{(k)}}_F^2$. 
Finally, we note that since $\X$ has $\O{\frac{k}{\eps^2}(\log n+\log\kappa)^4}$ rows and $\Z$ has $\O{\frac{k}{\eps}\log n\log^2\kappa}$ rows, then $\W$ has $\O{\frac{k}{\eps^2}(\log n+\log\kappa)^4}$ rows.
\end{proof}
\section{$\ell_1$-Subspace Embeddings}
\seclab{sec:lp}
In this section, we consider $\ell_1$-subspace embeddings in both the online model and the sliding window model. 
For $\ell_2$-subspace embeddings in \secref{sec:sample}, a central concept used to sample a row $\a_i\in\mathbb{R}^d$ in a matrix $\A_i=\a_1\circ\ldots\circ\a_i$ is the online leverage score, $\a_i(\A_{i-1}^\top\A_{i-1})^{-1}\a_i^\top$. 
An equivalent formulation for the (online) leverage score of $\a_i$ is the following:
\begin{definition}[Maximization Characterization of (Online) Leverage Scores]
For a matrix $\A=\a_1\circ\ldots\circ\a_n\in\mathbb{R}^{n\times d}$, the leverage score of $\a_i$ can also be written as $\max_{\x\in\mathbb{R}^d}\frac{|\a_i^\top\x|^2}{\norm{\A\x}_2^2}$. 
Similarly, the online leverage score of $\a_i$ can also be written as $\max_{\x\in\mathbb{R}^d}\frac{|\a_i^\top\x|^2}{\norm{\A_{i-1}\x}_2^2}$, where $\A_i=\a_1\circ\ldots\circ\a_i$. 
As by convention in \defref{def:rols}, the online leverage score is set to $1$ when the fraction is undefined. 
\end{definition}
Thus we define the analogous quantities for $\ell_1$-subspace embeddings that we shall use to govern sampling proabilities for each row:
\begin{definition}[(Online) $\ell_1$ Sensitivity]
For a matrix $\A=\a_1\circ\ldots\circ\a_n\in\mathbb{R}^{n\times d}$, we define the \emph{$\ell_1$ sensitivity} of $\a_i$ by $\max_{\x\in\mathbb{R}^d}\frac{|\a_i^\top\x|}{\norm{\A\x}_1}$ and the \emph{online $\ell_1$ sensitivity} of $\a_i$ by $\max_{\x\in\mathbb{R}^d}\frac{|\a_i^\top\x|}{\norm{\A_i\x}_1}$, where $\A_i=\a_1\circ\ldots\circ\a_i$. 
\end{definition}
We remark that the online $\ell_1$ sensitivity of $\a_i$ could be similarly defined by the minimum of $1$ and $\max_{\x\in\mathbb{R}^d}\frac{|\a_i^\top\x|}{\norm{\A_{i-1}\x}_1}$, again using the convention that the online $\ell_1$ sensitivity of $\a_i$ is $1$ if $\rank(\A_i)>\rank(\A_{i-1})$. 
The latter definition is more consistent with the definition of online leverage score by \cite{CohenMP16}, but the former will require slightly less notation in our analysis. 
Regardless, these qualities are within a constant factor. 

Note that from the definition, the online $\ell_1$ sensitivity of a row is at least as large as the $\ell_1$ sensitivity of the row. 
Similarly, the (online) $\ell_1$ sensitivity of each of the previous rows in $\A$ cannot increase when a new row $\r$ is added to $\A$. 
Thus we can use the online $\ell_1$ sensitivities to define an online algorithm for $\ell_1$-subspace embedding similar to \algref{alg:online:pcp}. 

\begin{algorithm}[!htb]
\caption{$\onlinelp$: Online algorithm for $\ell_1$-subspace embedding}
\alglab{alg:online:lp}
\begin{algorithmic}[1]
\Require{Stream of rows $\a_1,\ldots,\a_n\in\mathbb{R}^{1\times d}$, accuracy $\eps>\frac{1}{n}$, and parameter $k$.}
\Ensure{$\ell_1$-subspace embedding of $\A:=\a_1\circ\ldots\circ\a_n$.}
\State{$\M\gets\emptyset$}
\State{$\alpha\gets\frac{Cd}{\eps^2}\log n$ with sufficiently large constant $C>0$}
\For{each row $\a_t$}
\State{$\tau_t\gets4\cdot\max_{\x\in\mathbb{R}^d}\frac{|\a_t^\top\x|}{\norm{\M\x}_1}$}
\State{$p_t\gets\min(1,\alpha\tau_t)$}
\State{With probability $p_t$, $\M\gets\M\circ\frac{\a_t}{p_t}$}
\EndFor
\State{\Return $\M$.}
\end{algorithmic}
\end{algorithm}

We require the following concentration inequality on scalar martingales. 
\begin{theorem}[Freedman's inequality]
\cite{Freedman75}
\thmlab{thm:scalar:freedman}
Let $Y_0,Y_1,\ldots,Y_n$ be a scalar martingale with difference sequence $X_1,\ldots,X_n$. 
Suppose $|X_t|\le R$ for all $t\in[n]$ with high probability and let the predictable quadratic variation process of the martingale be defined by $w_k:= \sum_{t=1}^k\underset{t-1}{\mathbb{E}}\left[X_t^2\right]$, for $k\in[n]$. 
Then for all $\eps\ge 0$ and $\sigma^2 > 0$, and every $k \in [n]$, 
\[\PPr{\max_{t \in [k]} |Y_t|>\eps\text{ and } w_k \le \sigma^2}\le 2\exp\left(-\frac{\eps^2/2}{\sigma^2 + R\eps/3} \right).\]
\end{theorem}

We first show $\ell_1$ sensitivity sampling gives an $\ell_1$-subspace embedding. 
\begin{lemma}
\lemlab{lem:l1:subspace}
For $\eps>\frac{1}{n}$, \algref{alg:online:lp} returns a matrix $\M$ such that for all $\x\in\mathbb{R}^d$,
\[|\norm{\M\x}_1-\norm{\A\x}_1|\le\eps\norm{\A\x}_1,\]
with high probability.
\end{lemma}
\begin{proof}
Let $\x\in\mathbb{R}^d$ and suppose without loss of generality that $\norm{\A\x}_1=1$. 
Let $Y_0,Y_1,\ldots,Y_n$ be a martingale with difference sequence $X_1,\ldots,X_n$. 
For $j\ge 1$, we set $X_j=0$ if $Y_{j-1}\ge\eps$ and otherwise if $Y_{j-1}<\eps$, we define
\[X_j=
\begin{cases}
\left(\frac{1}{p_j}-1\right)|\a_j^\top\x| & \text{ if }\a_j\text{ is sampled in }\M\\
- |\a_j^\top\x|  & \text{otherwise}.
\end{cases}
\]
Observe that $\EEx{Y_j|Y_1,\ldots,Y_{j-1}}=Y_{j-1}$ so the sequence is a valid martingale. 
Moreover, the difference sequence defines $Y_j=|\norm{\M_j\x}_1-\norm{\A_j\x}_1|$, where $\M_j$ is the rows of $\M$ sampled at time $j$ and $\A_j=\a_1\circ\ldots\circ\a_j$. 

Thus if $p_j=1$, then $X_j=0$ and otherwise, $\EEx{X_j^2}\le\frac{1}{p_j}|\a_j^\top\x|^2=\frac{1}{\alpha\tau_j}|\a_j^\top\x|^2$. 
Since $\tau_j\ge|\a_j^\top\x|$ and we scaled $\norm{\A\x}_1=1$, then $|\a_j^\top\x|\le 1$ so that $\sum_{j=1}^n\EEx{X_j}^2\le\frac{1}{\alpha}$. 
By Freedman's inequality (\thmref{thm:scalar:freedman}) and $\alpha=\O{\frac{d}{\eps^2}\log n}$, it follows that
\[\PPr{|Y_n|>\eps}\le2\exp\left(-\frac{\eps^2/2}{\alpha + R\alpha/3}\right)\le\frac{1}{2^d\,\poly(n)},\]
for sufficiently large $\alpha$. 
Since we scaled $\norm{\A\x}_1=1$, then $|\norm{\M\x}_1-\norm{\A\x}_1|\le\eps\norm{\A\x}_1$ with probability at least $1-\frac{1}{2^d\,\poly(n)}$. 

Now consider the unit ball $B=\{\A\y\in\mathbb{R}^d\,|\,\norm{\A\y}_1=1\}$ and let $\mathcal{N}$ be an $\eps$-net of $B$, constructed greedily. 
Observe that $\mathcal{N}$ has at most $\left(\frac{3}{\eps}\right)^d$ points, since balls of radius $\frac{\eps}{2}$ around each point cannot overlap, but must all fit into a ball of radius $1+\frac{\eps}{2}$. 
Thus by a union bound for $\frac{1}{\eps}<n$, we have $|\norm{\M\y}_1-\norm{\A\y}_1|\le\eps\norm{\A\y}_1$ for all $\A\y\in\mathcal{N}$, with probability at least $1-\frac{1}{\poly(n)}$. 

For any vector $\z\in\mathbb{R}^d$ such that $\norm{\A\z}_1=1$, we construct a sequence $\A\y_1,\A\y_2,\ldots$ such that $\norm{\A\z-\sum_{j=1}^i\A\y_j}_1\le\eps^i$ and there exists a constant $\gamma_i\le\eps^{i-1}$ such that $\frac{1}{\gamma_i}\A\y_i\in\mathcal{N}$ for all $i$. 
Let $\A\y_1$ be the point in $\mathcal{N}$ closest to $\A\z$ so that $\norm{\A\z-\A\y_1}_1\le\eps$. 
Given a sequence $\A\y_1,\ldots,\A\y_{i-1}$ such that $\gamma_i:=\norm{\A\z-\sum_{j=1}^{i-1}\A\y_j}_1\le\eps^{i-1}$, note that $\frac{1}{\gamma_i}\norm{\A\z-\sum_{j=1}^{i-1}\A\y_j}_1=1$, so there exists a point $\A\y_i\in\mathcal{N}$ that is within distance $\eps$ of $\A\z-\sum_{j=1}^{i-1}\A\y_j$, which completes the induction. 
Hence since $\y_i\in\mathcal{N}$ for all $i$ and $\M=\S\A$ for some matrix $\S$ that samples rows of $\A$, we have 
\begin{align*}
\norm{\M\z}_1&=\norm{\M\gamma_1\y_1+\M\gamma_2\y_2+\M\gamma_3\y_3+\ldots}_1\\
&\le\sum_{i=1}^\infty\gamma_i\norm{\M\y_i}_1\le(1+\eps)+\eps(1+\eps)+\ldots\le(1+\O{\eps}),\\
\norm{\M\z}_1&=\norm{\M\gamma_1\y_1+\M\gamma_2\y_2+\M\gamma_3\y_3+\ldots}_1\\
&\ge\gamma_1\norm{\M\y_1}-\sum_{i=2}^\infty\gamma_i\norm{\M\y_i}_1\ge(1-\eps)-\eps(1+\eps)+\ldots\ge(1-\O{\eps}),
\end{align*}
and thus 
\[|\norm{\M\z}_1-\norm{\A\z}_1|\le\O{\eps}\norm{\A\z}_1.\]
The claim follows from a rescaling of $\eps$. 
\end{proof}
To analyze the space complexity, it remains to bound the sum of the online $\ell_1$ sensitivities. 
Much like the proof of \thmref{thm:sw:spectral}, we first bound the sum of regularized online $\ell_1$ sensitivities and then relate these quantities. 
We then require a few structural results relating online leverage scores and their relationships to regularized online $\ell_1$ sensitivities.
\newcommand{\whackamole}{Let $\A=\a_1\circ\ldots\circ\a_n\in\mathbb{R}^{n\times d}$.  
Let $\lambda>0$, $\B_0=\lambda\I_d$, $\B=\underbrace{\B_0\circ\ldots\circ\B_0}_{n\text{ times}}$ and $T$ be the sum of the online leverage scores of $\B\circ\A$. 
For any $\gamma\ge 1$, there exists a diagonal matrix $\W$ with entries in $[0,1]$ that contains at most $\frac{n}{\gamma}$ entries strictly less than $1$ so that the online leverage score of any row of $\W\A$ with respect to $\C:=\B\circ(\W\A)$ is at most $\frac{\gamma T}{n}$.  
That is, if $\tau_i^{\OL}(\C)$ denotes the leverage score of the $i\th$ row of $\W\A$ with respect to $\C$, then $\tau_i^{\OL}(\C)\le\frac{\gamma T}{n}$ for all $i\in[n]$.} 
\begin{lemma}[Online whack-a-mole]
\cite{CohenLMMPS15}
\lemlab{lem:online:reweighting}
\whackamole
\end{lemma}
It should be noted that \cite{CohenLMMPS15} prove a slightly different version of \lemref{lem:online:reweighting} so that the reweighted input matrix $\W\A$ has all of its leverage scores uniformly bounded. 
However, given the regularization matrix $\B$, the same proof demonstrates that the reweighted input matrix $\W\A$ has all of its online leverage scores uniformly bounded. 
For completeness, we provide the modified proof of \cite{CohenLMMPS15} in \appref{app:whackamole}. 

We now show that if the rows of a matrix has uniformly bounded leverage scores, then the $\ell_1$ sensitivities must also be uniformly bounded. 
\begin{lemma}[Uniform leverage scores imply uniform $\ell_1$ sensitivities]
\lemlab{lem:uniform_to_uniform}
Given a matrix $\A=\a_1\circ\ldots\circ\a_n$, let $\tau_i(\A)$ and $\zeta_i(\A)$ denote the leverage score and $\ell_1$ sensitivity of $\a_i$, respectively, for each $i\in[n]$. 
Let $\gamma>1$ and suppose that $\zeta_i(\A)\le\frac{\gamma d}{n}$ for all $i\in[n]$. 
Then $\tau_i(\A)\le\frac{\gamma d}{n}$ for all $i\in[n]$. 
\end{lemma}
\begin{proof}
For any given $i$, let $\x^* = \argmax_{\x}\frac{|\a_i^\top\x|}{\norm{\A\x}_1}$. 
If $\frac{\max_j|[\A\x^*]_j|}{\norm{\A\x^*}_2} > \sqrt{\frac{\gamma d}{n}}$, where $[\A\x^*]_j$ denotes coordinate $j$ of vector $\A\x^*$, then row $\a_j$ would have leverage score $\tau_j(\A)>\frac{\gamma d}{n}$, which violates the assumption of the claim. 
Thus, $\frac{\max_j|\A\x^*|_j}{\norm{\A\x^*}_2}\le \sqrt{\frac{\gamma d}{n}}$. 

Now for any vector $\z\in\mathbb{R}^n$, we have $\frac{\norm{\z}_2}{\max_j|[\z]_j|}\cdot\norm{\z}_2\le\norm{\z}_1$. 
Therefore, 
\begin{align*}
\zeta_i(\A)=\frac{|\a_i^\top\x^*|}{\norm{\A\x^*}_1}\le|\a_i^\top\x^*|\frac{\max_j|[\A\x^*]_j|}{\norm{\A\x^*}_2^2}=\max\frac{\max_j|\A\x^*|_j}{\norm{\A\x^*}_2}\cdot \sqrt{\frac{\gamma d}{n}}\le\frac{\gamma d}{n}.
\end{align*}
\end{proof}
\noindent
We now bound the sum of the online $\ell_1$ sensitivities of a matrix. 
\begin{lemma}[Bound on Sum of Online $\ell_1$ Sensitivities]
\lemlab{lem:ol:sen:bound}
Let $\A=\a_1\circ\ldots\circ\a_n\in\mathbb{R}^{n\times d}$. 
Let $\zeta^{\OL}_i(\A)$ be the online $\ell_1$ sensitivity of $\a_i$ with respect to $\A$, for each $i\in[n]$. 
Then $\sum_{i=1}^n\zeta^{\OL}_i(\A)=\O{d\log n\log(\kappa nd)}$, where $\kappa$ is the subset condition number of $\A$. 
\end{lemma}
\begin{proof}
Let $\lambda>0$, $\B_0=\lambda\I_d$, $\B=\underbrace{\B_0\circ\ldots\circ\B_0}_{n\text{ times}}$, and $T$ be an upper bound on the sum of the online leverage scores of any submatrix of $\X:=\B\circ\A$. 
For each $i\in[n]$, let $\zeta_i(\X)$ and $\tau_i(\X)$ be the $\ell_1$ sensitivity and leverage score of row $\a_i$ with respect to $\X$, respectively. 
Similarly, let $\zeta^{\OL}_i(\X)$ denote the online $\ell_1$ sensitivity of row $\a_i$ with respect to $\X$. 
By \lemref{lem:online:space} and \lemref{lem:online:reweighting}, there exist a set $S\subseteq[n]$ of size at most $d$ and a diagonal matrix $\W$ with $\frac{n}{2}$ non-unit entries such that $\tau^{\OL}_i(\B\circ(\W\A))\le\frac{2T}{n}$ for all $i\in[n]$, where $T$ is the sum of the online ridge leverage scores of $\A$ with regularization $\lambda>0$. 
Let $\C=\B\circ(\W\A)$ and $S$ be the set of all indices $i$ such that $\W_{i,i}=1$. 
For each $i\in[n]$, let $\C_{i+nd}$ be the matrix formed by the first $i+nd$ rows of $\C$. 
Since the online leverage score of each row $\a_j$ with $j\le i$ in $\C$ is at most $\frac{2T}{n}$ and the leverage scores of the first $nd$ rows of $\C_{i+nd}$ is at most $\frac{1}{n}$, then all rows of $\C_{i+nd}$ have leverage score at most $\frac{2T}{n}$. 
By \lemref{lem:uniform_to_uniform}, all rows of $\C_{i+nd}$ have $\ell_1$ sensitivity at most $\frac{2T}{n}$. 

Now for any $i\in S$, we have from the definition of online $\ell_1$ sensitivity that
\begin{align*}
\zeta_i^{\OL}(\X)\le\zeta_i^{\OL}(\C)=\zeta_i^{\OL}(\C_{i+nd})=\zeta_i(\C_{i+nd})\le\frac{2T}{n}
\end{align*}
so that $\sum_{i\in S}\zeta_i^{\OL}(\X)\le2T$. 
Let $U= \{1,\ldots, n\} \setminus S$ and let $\X_U$ be the matrix $\X$ restricted to the rows whose indices are in $U$. 
By monotonicity of sensitivities, $\zeta_i^{\OL}(\X)\le \zeta_i^{\OL}(\X_U)$.  
By induction, 
\begin{align*}
\sum_{i\in U} \zeta_i^{\OL}(\X)\le \sum_{i\in U} \zeta_i^{\OL}(\X_U) \le 2T\log n
\end{align*}
Thus we have $\sum_i\zeta_i^{\OL}(\X)=\O{T\log n}$, which is $\O{d\log n\log\kappa'}$ by \lemref{lem:online:space}, where $\kappa'$ denotes the largest condition number of any submatrix of $\X$. 
Now suppose we set $\lambda=\frac{\gamma}{\sqrt{d}n}$, where $\gamma$ is the minimum of the nonzero singular values across all submatrices of $\A$. 
We also have for any $\x\in\mathbb{R}^d$ that 
\[\norm{\x}_1\le\sqrt{d}\norm{\x}_2\le\frac{\sqrt{d}}{\sigma_{\min}(\A_i)}\norm{\A_i\x}_2\le\frac{\sqrt{d}}{\sigma_{\min}(\A_i)}\norm{\A_i\x}_1.\]
Thus for our choice of $\lambda$, we have
\[\frac{|\a_i^\top\x|}{\norm{(\B\circ\A_i)\x}_1}\ge\frac{|\a_i^\top\x|}{\norm{\A_i\x}_1+\frac{\lambda \sqrt{d}n}{\sigma_{\min}(\A_i)}\norm{\A_i\x}_1}\ge\frac{|\a_i^\top\x|}{2\norm{\A_i\x}_1},\]
so that the online $\ell_1$ sensitivity of $\a_i$ with respect to $\A$ is within a constant factor of the online $\ell_1$ sensitivity of $\a_i$ with respect to $\X=\B\circ\A$. 
That is, $\sum_{i=1}^n\zeta^{\OL}_i(\A)=\O{d\log n\log\kappa}$. 
Moreover, for our choice of $\lambda$, we have $\log\kappa'=\log(\kappa nd)$, where $\kappa$ denotes the largest condition number of any subset of rows of $\A$, i.e., the subset condition number of $\A$. 
\end{proof}
\noindent
We now give the full guarantees of \algref{alg:online:lp}. 
\begin{theorem}[Online $\ell_1$-Subspace Emebedding]
\thmlab{thm:online:lp}
Given $\eps>0$ and a matrix $\A\in\mathbb{R}^{n\times d}$ with subset condition number $\kappa$ whose rows $\a_1,\ldots,\a_n$ arrive sequentially in a stream, there exists an online algorithm that outputs a matrix $\M$ with $\O{\frac{d^2}{\eps^2}\log^2 n\log(\kappa nd)}$ (rescaled) rows of $\A$ such that
\[(1-\eps)\norm{\A\x}_1\le\norm{\M\x}_1\le(1+\eps)\norm{\A\x}_1,\]
for all $\x\in\mathbb{R}^d$ with high probability.  
\end{theorem}
\begin{proof}
Note that for $\eps<\frac{1}{n}$, a trivial algorithm is just to sample every row. 
For $\eps>\frac{1}{n}$, consider \algref{alg:online:lp} and recall that the property that $\M$ gives an $\ell_1$-subspace embedding follows from \lemref{lem:l1:subspace} with high probability. 
Moreover by a union bound, $\M_t$ is a $\ell_1$-subspace embedding of $\A_t$ with approximation $(1+\eps)$, where $\M_t$ are the rows of $\M$ stored at time $t$ and $\A_t=\a_1\circ\ldots\circ\a_t$. 
It suffices to consider $\eps<\frac{1}{2}$, for which each row $\a_i$ is thus sampled with probability at most $4\alpha\zeta^{\OL}_i$, where $\zeta^{\OL}_i(\A)$ is the online $\ell_1$ sensitivity of $\a_i$ with respect to $\A$. 
By \lemref{lem:ol:sen:bound}, we have $\sum_{i=1}^n\zeta^{\OL}_i(\A)=\O{d\log n\log(\kappa nd)}$. 
Since we oversample by a factor of $\alpha=\O{\frac{d}{\eps^2}\log n}$, then the space complexity of the algorithm follows from a coupling argument and standard Chernoff bounds.
\end{proof}

Finally, note that we can use the reverse online $\ell_1$ sensitivities in the framework of \algref{alg:sampling:framework} to obtain an $\ell_1$-subspace embedding in the sliding window model. 
\begin{algorithm}[t]
\caption{$\score(\r,\A)$ function for $\ell_1$-subspace embedding}
\alglab{alg:lp:score}
\begin{algorithmic}[1]
\Require{A row $\r\in\mathbb{R}^d$ and a matrix $\A\in\mathbb{R}^{m\times d}$.}
\Ensure{Scaled $\ell_1$ sensitivity of $\r$ with respect to $\A$.}
\If{$\rank(\A)=\rank(\A\circ\r)$}
\State{\Return $4d\cdot\max_{\x\in\mathbb{R}^d}\frac{|\r^\top\x|}{\norm{\A\x}_1}$}
\Else
\State{\Return $1$}
\EndIf
\end{algorithmic}
\end{algorithm}

Since we are considering a sliding window algorithm, we consider the reverse online $\ell_1$ sensitivities rather than using the online $\ell_1$ sensitivities as for the online $\ell_1$-subspace embedding algorithm in \algref{alg:online:lp}. 
For a matrix $\A=\a_1\circ\ldots\circ\a_n\in\mathbb{R}^{n\times d}$, the reverse online $\ell_1$ sensitivity of row $\a_i$ is defined in the natural way, by $\max_{\x\in\mathbb{R}^d}\frac{\a_i^\top\x}{\norm{\Z_{i-1}\x}_1}$ if $\rank(\Z_{i-1})=\rank(\Z_i)$ and by $1$ otherwise, where $\Z_i=\a_n\circ\ldots\circ\a_i$. 
Note that \algref{alg:sampling:framework} with the $\score$ function in \algref{alg:lp:score} evaluates the importance of each row compared to the following rows, so we are approximately sampling by reverse online $\ell_1$ sensitivities, as desired. 
The proof follows along the same lines as \thmref{thm:sw:spectral} and \thmref{thm:sw:pcp}, using a martingale argument to show that the approximations for the reverse online $\ell_1$ sensitivities induce sufficiently high sampling probabilities, while still using the sum of the online $\ell_1$ sensitivities to bound the total number of sampled rows. 
We sketch the proof below, as \secref{sec:coreset} presents an improved algorithm for $\ell_1$-subspace embedding in the sliding window model that is nearly space optimal, up to lower order factors.   
\begin{theorem}[Randomized $\ell_1$-Subspace Embedding Sliding Window Algorithm]
\thmlab{thm:sw:lp}
Let $\r_1,\ldots,\r_n\in\mathbb{R}^{d}$ be a stream of rows and $\kappa$ be the subset condition number of the matrix $\r_1\circ\ldots\circ\r_n$. 
Let $W>0$ be a window size parameter and $\A=\r_{n-W+1}\circ\ldots\circ\r_n$ be the matrix consisting of the $W$ most recent rows. 
Given a parameter $\eps>0$, there exists an algorithm that outputs a matrix $\M$ with a subset of (rescaled) rows of $\A$ such that $(1-\eps)\norm{\A\x}_1\le\norm{\M\x}_1\le(1+\eps)\norm{\A\x}_1$ for all $\x\in\mathbb{R}^d$ and stores $\O{\frac{d^2}{\eps^2}\log^2 n\log(\kappa nd)}$ rows at any time, with high probability.
\end{theorem}
\begin{proof}
Consider \algref{alg:sampling:framework} using the $\score$ function of \algref{alg:lp:score}. 
Since we have already established the martingale argument of the online $\ell_1$ sensitivities in \lemref{lem:l1:subspace}, the same argument follows for the reverse online $\ell_1$ sensitivities. 
It follows that each row $\r_i$ that increases the rank of $\Z_i$ for $i\ge n-W+1$ is sampled with probability within a constant factor of $\alpha d\cdot\max_{\x\in\mathbb{R}^d}\frac{|\r_i^\top\x|}{\norm{\Z_{i-1}\x}_1}$, where $\Z_i=\r_{n-W+1}\circ\ldots\circ\r_i$ for each $i\in[n]$. 
There are at most $d$ rows that increase the rank of $\Z_i$ from $\Z_{i+1}$. 
Since $\alpha=\O{\frac{d}{\eps^2}\log n}$ and the sum of the reverse online $\ell_1$ sensitivities is bounded by $\O{d\log n\log(\kappa nd)}$ by \lemref{lem:ol:sen:bound}, then the space complexity of the algorithm follows from a coupling argument and standard Chernoff bounds.
\end{proof}

\paragraph{Approximating the Sensitivities.} 
We observe that explicitly computing the (online) $\ell_1$ sensitivity $\max_{\x\in\mathbb{R}^d}\frac{|\a_i^\top\x|}{\norm{\A\x}_1}$ for a row $\a_i$ in $\A=\a_1\circ\ldots\circ\a_n\in\mathbb{R}^{n\times d}$ may be infeasible. 
However, it suffices to approximate the $\ell_1$ sensitivity to within an additive $\frac{1}{n^2}$ in polynomial time using linear programming. 
The oversampling parameter $\alpha$ can then be scaled by a factor of two to handle any row with online $\ell_1$ sensitivity at least $\frac{2}{n^2}$ and those rows with online $\ell_1$ sensitivity less than $\frac{2}{n^2}$ will only incur $\O{\log n}$ additional samples in total, with high probability.   
\section{A Coreset Framework for Deterministic Sliding Window Algorithms}
\seclab{sec:coreset}
In this section, we give a framework for deterministic sliding window algorithms based on the merge and reduce paradigm and the concept of online coresets. 
We define an online coreset for a matrix $\A$ as a weighted subset of rows of $\A$ that also provides a good approximation to prefixes of $\A$:
\begin{definition}[Online Coreset]
An \emph{online coreset} for a function $f$, an approximation parameter $\eps>0$, and a matrix $\A\in\mathbb{R}^{n\times d}=\a_1\circ\ldots\circ\a_n$ is a subset of weighted rows of $\A$ such that for any $\A_i=\a_1\circ\ldots\circ\a_i$ with $i\in[n]$, we have $f(\M_i)$ is a $(1+\eps)$-approximation of $f(\A_i)$, where $\M_i$ is the matrix that consists of the weighted rows of $\A$ in the coreset that appear at time $i$ or before. 
\end{definition}

We can use deterministic online coresets for deterministic sliding window algorithms using a merge and reduce framework. 
The idea is to store the $\mspace$ most recent rows in a block $\B_0$, for some parameter $\mspace$ related to the coreset size.  
Once $\B_0$ is full, we reduce $\B_0$ to a smaller number of rows by setting $\B_1$ to be a $\left(1+\frac{\eps}{\log n}\right)$ coreset of $\B_0$, \emph{starting with the most recent row}, and then empty $\B_0$ so that it can again store the most recent rows. 
Subsequently, whenever $\B_0$ is full, we merge successive non-empty blocks $\B_0,\ldots,\B_i$ and reduce them to a $\left(1+\frac{\eps}{\log n}\right)^i$ coreset, indexed as $\B_{i+1}$. 
Since the entire stream has length $n$, then by using $\O{\log n}$ blocks $\B_i$, we will have a $\left(1+\frac{\eps}{\log n}\right)^{\log n}$ coreset, starting with the most recent row. 
Rescaling $\eps$, this gives a merge and reduce based framework for $(1+\eps)$ deterministic sliding window algorithms based on online coresets. 
We give the framework in full in \algref{alg:deterministic:framework}.

\begin{algorithm}[!htb]
\caption{Merge and reduce framework for deterministic sliding window matrix algorithms using online coresets.}
\alglab{alg:deterministic:framework}
\begin{algorithmic}[1]
\Require{A matrix function $f$ that admits an online coreset, a stream of rows $\r_1,\ldots,\r_n\in\mathbb{R}^{1\times d}$, approximation parameter $\eps>0$, and a parameter $W>0$ for the window}
\Ensure{An approximation to $f(\A)$, where $\A\in\mathbb{R}^{W\times d}=\r_{n-W+1}\circ\ldots\circ\r_n$}
\State{Initialize blocks $\B_0,\B_1,\ldots,\B_{\log n}\gets\emptyset$}
\For{each row $\r_t$}
\If{$\B_0$ does not contain $\mspace$ rows}
\Comment{$\mspace$ is the coreset size}
\State{$\B_0\gets\r_t\circ\B_0$}
\Comment{Keep timestamps for all stored rows}
\Else
\State{Let $i>0$ be the minimal index such that $\B_i=\emptyset$.}
\Comment{Prepare to merge and reduce}
\State{$\B_i\gets\coreset\left(\M,\frac{\eps}{\log n}\right)$, where $\M=\B_0\circ\ldots\circ\B_{i-1}$}
\Comment{Online coreset for function $f$}
\For{$j=0$ to $j=i-1$}
\State{$\B_j\gets\emptyset$}
\EndFor
\State{$\B_0\gets\r_t$}
\EndIf
\If{there exists a row $\r$ in a block $\B_i$ with timestamp before $t-W+1$}
\State{Delete $\r$ from $\B_i$}
\EndIf
\EndFor
\State{\Return $\B_{\log n}\circ\ldots\circ\B_1\circ\B_0$}
\end{algorithmic}
\end{algorithm}

\begin{lemma}
\lemlab{lem:coreset:block}
$\B_i$ in \algref{alg:deterministic:framework} is a $\left(1+\frac{\eps}{\log n}\right)^i$ online coreset for $2^{i-1}\mspace$ rows.
\end{lemma}
\begin{proof}
Note that $\B_i$ can only be non-empty if at some point $\B_0$ contains $\mspace$ rows and $\B_1,\ldots,\B_{i-1}$ are all non-empty. 
By induction, $\B_j$ is a $\left(1+\frac{\eps}{\log n}\right)^j$ online coreset for $2^{j-1}\mspace$ rows for each $1\le j<i$. 
$\B_i$ is then a $\left(1+\frac{\eps}{\log n}\right)$ online coreset for the rows in $\B_0\circ\ldots\circ\B_{i-1}$. 
Thus for $i=1$, then $\B_i$ is a $\left(1+\frac{\eps}{\log n}\right)$ online coreset for $\mspace$ rows and for $i>1$, $\B_i$ is a $\left(1+\frac{\eps}{\log n}\right)\left(1+\frac{\eps}{\log n}\right)^{i-1}=\left(1+\frac{\eps}{\log n}\right)^i$ online coreset for $\mspace+\sum_{j=0}^{i-2}2^j\mspace=2^{i-1}\mspace$ rows.
\end{proof}

\begin{theorem}
\thmlab{thm:coreset:framework}
Let $\r_1,\ldots,\r_n\in\mathbb{R}^{1\times d}$ be a stream of rows, $\eps>0$, and $\A=\r_{n-W+1}\circ\ldots\circ\r_n$ be the matrix consisting of the $W$ most recent rows. 
If there exists a deterministic online coreset algorithm for a matrix function $f$ that stores $S(n,d,\eps)$ rows, then there exists a deterministic sliding window algorithm that stores $\O{S\left(n,d,\frac{\eps}{\log n}\right)\log n}$ rows and outputs a matrix $\M$ such that $f(\M)$ is a $(1+\eps)$-approximation of $f(\A)$. 
\end{theorem}
\begin{proof}
Let $\coreset$ be a deterministic online coreset algorithm for a function $f$ that stores $S(n,d,\eps)$ rows. 
Consider \algref{alg:deterministic:framework}, setting $\mspace=S\left(n,d,\frac{\eps}{\log n}\right)$. 
Let $t\le\mspace$ and suppose $\B_0$ contains the $t$ rows $\r_{m-t+1},\ldots,\r_m$ at any time $m$. 
We show the stronger property that at any time $m$, the rows of $\B_0\circ\ldots\circ\B_i$ can be simultaneously used to provide a $\left(1+\frac{\eps}{\log n}\right)^i$ approximation to $f(\A_i)$ for all integers $1\le i\le t+(2^{i-1})\mspace$, where $\A_i=\r_{m-i+1}\circ\ldots\circ\r_m$. 

Let $\B_{i_1},\ldots,\B_{i_x}$ be the nonempty blocks after $B_0$, with $0<i_1<\ldots<i_x$. 
By \lemref{lem:coreset:block}, $\B_{i_1}$ is a $\left(1+\frac{\eps}{\log n}\right)^{i_1}$ online coreset for the $2^{i_1-1}\mspace$ rows before time $m-t+1$. 
Similarly, $\B_{i_y}$ is a $\left(1+\frac{\eps}{\log n}\right)^{i_y}$ coreset for the $2^{i_y-1}\mspace$ rows before the rows represented by $\B_{i_{y-1}}$. 
If $t<\mspace$, then no coresets are merged and thus the hypothesis holds by the definition of online coreset. 

Let $j$ be the maximal index such that $i_j=j$. 
If $t=\mspace$, then $\B_{i_j+1}$ is formed by using $\coreset$ to reduce from the merging of $\B_0\circ\B_{i_j}$. 
Thus, $\B_{i_j+1}$ is a $\left(1+\frac{\eps}{\log n}\right)^{j+1}$ online coreset for the $2^{i_j}\mspace$ rows that appear before the rows represented by $\B_{i_{j+1}}$, so the hypothesis again holds by the definition of online coreset. 
Since $\A=\r_{n-W+1}\circ\ldots\circ\r_n$, then \algref{alg:deterministic:framework} outputs a matrix $\M$ such that $f(\M)$ is a $\left(1+\frac{\eps}{\log n}\right)^{\log n}$ approximation of $f(\A)$. 
The correctness then follows by rescaling $\eps$. 

$\B_0$ stores at most $\mspace=S\left(n,d,\frac{\eps}{\log n}\right)$ rows. 
The algorithm stores at most $\log n$ blocks and since the error parameter is $\frac{\eps}{\log n}$ to each call of $\coreset$ in 
each block $\B_i$ with $i>1$ stores at most $S\left(n,d,\frac{\eps}{\log n}\right)$ rows by the definition of $\mspace$. 
Hence, the algorithm stores $\O{S\left(n,d,\frac{\eps}{\log n}\right)\log n}$ rows at any given time. 
\end{proof}

The online row sampling algorithm of~\cite{CohenMP16} shows the existence of an online coreset for spectral approximation. 
Note that if runtime and space are not issues, this coreset can be explicitly computed by computing the online leverage scores, enumeration over sufficiently small subsets of scaled rows, and checking whether a subset is an online coreset for spectral approximation. 
\begin{theorem}[Online Coreset for Spectral Approximation]
\cite{CohenMP16}
\thmlab{thm:l2:coreset}
For a matrix $\A\in\mathbb{R}^{n\times d}=\a_1\circ\ldots\circ\a_n$, there exists a constant $C>0$ and a deterministic algorithm $\coreset(\A,\eps)$ that outputs an online coreset of $\frac{Cd}{\eps^2}\log n\log\kappa$ weighted rows of $\A$. 
For any $i\in[n]$, let $\M_i$ be the weighted rows of $\A$ in the coreset that appear at time $i$ or before. 
Then $(1-\eps)\norm{\A_i\x}_2\le\norm{\M_i\x}_2\le(1+\eps)\norm{\A_i\x}_2$ for all $\x\in\mathbb{R}^d$, where $\A_i=\a_1\circ\ldots\circ\a_i$. 
\end{theorem}
\noindent
Then \thmref{thm:coreset:framework} and \thmref{thm:l2:coreset} imply:
\begin{theorem}[Deterministic Sliding Window Algorithm for Spectral Approximation]
\thmlab{thm:det:spectral}
Let $\r_1,\ldots,\r_n\in\mathbb{R}^{1\times d}$ be a stream of rows and $\kappa$ be the condition number of the stream. 
Let $\eps>0$ and $\A=\r_{n-W+1}\circ\ldots\circ\r_n$ be the matrix consisting of the $W$ most recent rows. 
There exists a deterministic algorithm that stores $\O{\frac{d}{\eps^2}\log^4 n\log\kappa}$ rows and outputs a matrix $\M$ such that $(1-\eps)\norm{\A\x}_2\le\norm{\M\x}_2\le(1+\eps)\norm{\A\x}_2$ for all $\x\in\mathbb{R}^d$. 
\end{theorem}
\noindent
\thmref{thm:online:pcp} shows that the existence of an online coreset for computing a rank $k$ projection-cost preservation. 
\begin{theorem}[Online Coreset for Rank $k$ Projection-Cost Preservation]
\thmlab{thm:pcp:coreset}
For a matrix $\A\in\mathbb{R}^{n\times d}=\a_1\circ\ldots\circ\a_n$, there exists a constant $C>0$ and a deterministic algorithm $\coreset(\A,\eps)$ that outputs an online coreset of $\frac{Ck}{\eps^2}\log n\log\kappa$ weighted rows of $\A$. 
For any $i\in[n]$, let $\M_i$ be the weighted rows of $\A$ in the coreset that appear at time $i$ or before. 
Then $\M_i$ is a $(1+\eps)$ projection-cost preservation for $\A_i:=\a_1\circ\ldots\circ\a_i$. 
\end{theorem}
\noindent
Thus \thmref{thm:coreset:framework} and \thmref{thm:pcp:coreset} give:
\begin{theorem}[Deterministic Sliding Window Algorithm for Rank $k$ Projection-Cost Preservation]
\thmlab{thm:det:pcp}
Let $\r_1,\ldots,\r_n\in\mathbb{R}^{1\times d}$ be a stream of rows and $\kappa$ be the condition number of the stream. 
Let $\eps>0$ and $\A=\r_{n-W+1}\circ\ldots\circ\r_n$ be the matrix consisting of the $W$ most recent rows. 
There exists a deterministic algorithm that stores $\O{\frac{k}{\eps^2}\log^4 n\log\kappa}$ rows and outputs a matrix $\M$ such that $(1-\eps)\norm{\A-\A\P}_F\le\norm{\M-\M\P}_F\le(1+\eps)\norm{\A-\A\P}_F$ for all rank $k$ orthogonal projection matrices $\P\in\mathbb{R}^{d\times d}$.
\end{theorem}

For $\ell_1$-subspace embeddings, we can use our online coreset from \thmref{thm:online:lp}, but in fact \cite{CohenP15} showed the existence of an offline coreset for $\ell_1$-subspace embeddings that stores a smaller number of rows. 
The offline coreset of \cite{CohenP15} is based on sampling rows proportional to their Lewis weights. 
We define a corresponding online version of Lewis weights:
\begin{definition}[(Online) Lewis Weights]
\deflab{def:lewis}
For a matrix $\A=\a_1\circ\ldots\circ\a_n\in\mathbb{R}^{n\times d}$, let $\A_i=\a_1\circ\ldots\circ\a_i$. 
Let $w_i(\A)$ denote the Lewis weight of row $\a_i$. 
Then the Lewis weights of $\A$ are the unique weights such that $w_i(\A)=\left(\a_i(\A^\top\W^{-1}\A)^{-1}\a_i^\top\right)^{1/2}$, where $\W$ is a diagonal matrix with $\W_{i,i}=w_i(\A)$. 
Equivalently, $w_i(\A)=\tau_i(\W^{-1/2}\A)$, where $\tau_i(\W^{-1/2}\A)$ denotes the leverage score of row $i$ of $\W^{-1/2}\A$. 
We define the \emph{online Lewis weight} of $\a_i$ to be the Lewis weight of row $\a_i$ with respect to the matrix $\A_i$. 
\end{definition}
The results of \cite{CohenP15} hold with high probability and they also show that Lewis weights cannot increase with the addition of new rows, so perhaps it is not surprising that their construction can be easily modified to form an online coreset based on sampling rows proporitional to their \emph{online} Lewis weights. 
\begin{lemma}[Monotonicity of Lewis Weights]
\cite{CohenP15}
\lemlab{lem:lewis:monotonic}
For a matrix $\A=\a_1\circ\ldots\circ\a_n\in\mathbb{R}^{n\times d}$ and $i\in[n]$, let $w_i(\A)$ denote the Lewis weight of row $\a_i$ with respect to $\A$ and let $\tau_i(\B)$ denote the Lewis weight of row $\a_i$ with respect to $\B:=\A\circ\r$ for any row $\r\in\mathbb{R}^{d}$. 
Then $\tau_i(\A)\ge\tau_i(\B)$. 
\end{lemma}
\begin{theorem}[Online Coreset for $\ell_1$-Subspace Embedding]
\thmlab{thm:lp:coreset}
\cite{CohenP15}
Let $\A\in\mathbb{R}^{n\times d}=\a_1\circ\ldots\circ\a_n$,
If there exists an upper bound $ C$ on the sum of the online Lewis weights of $\A$, then there exists a constant $C>0$ and a deterministic algorithm $\coreset(\A,\eps)$ that outputs an online coreset of $\frac{C}{\eps^2}\log n$ weighted rows of $\A$. 
For any $i\in[n]$, let $\M_i$ be the weighted rows of $\A$ in the coreset that appear at time $i$ or before. 
Then $\M_i$ is an $\ell_1$-subspace embedding for $\A_i:=\a_1\circ\ldots\circ\a_i$ with approximation $(1+\eps)$. 
\end{theorem}
It thus suffices to analyze the sum of the online Lewis weights. 
We first require a few structural results on Lewis weights. 
\cite{CohenP15} gives an iterative algorithm (\figref{fig:iter:lewis}) that converges toward the Lewis weights. 
\begin{figure}[!htb]
\begin{mdframed}
\begin{enumerate}
\item
Given input matrix $\A\in\mathbb{R}^{n\times d}$, initialize $\W^{(0)}\gets\I_n$
\item
Repeat: $\W^{(j)}_{i,i}\gets\left(\a_i(\A^\top(\W^{(j-1)})^{-1}\A)^{-1}\a_i^\top\right)^{1/2}$ for all $i\in[n]$
\end{enumerate}
\end{mdframed}
\caption{Iterative algorithm for computing Lewis weights.}
\figlab{fig:iter:lewis}
\end{figure}

The algorithm in \figref{fig:iter:lewis} has the following property, which shows that it is a contraction mapping and thus converges to the Lewis weights due to the Banach fixed point theorem. 
\begin{lemma}
\lemlab{lem:lewis:prop}
\cite{CohenP15}
If $\W$ is diagonal matrix representing the Lewis weights and $\frac{1}{\gamma}\W^{(j)}\preceq\W\preceq\gamma\W^{(j)}$ for some $j>0$ and $\gamma>1$ in the algorithm in \figref{fig:iter:lewis}, then $\frac{1}{\sqrt{\gamma}}\W^{(j+1)}\preceq\W\preceq\sqrt{\gamma}\W^{(j+1)}$. 
Thus, $\W_{(j)}$ converges to the Lewis weights.  
\end{lemma}
We now show that if the rows of a matrix has uniformly bounded leverage scores, then the Lewis weights must also be uniformly bounded. 
Although the statement is similar to \lemref{lem:uniform_to_uniform}, the structure of the Lewis weights requires a different strategy of iteratively computing values that converge to the Lewis weights. 
\begin{lemma}[Uniform leverage scores imply uniform Lewis weights]
\lemlab{lem:uniform_lewis}
If the leverage scores of $\A$ are at most $\frac{\gamma d}{n}$ for some $\gamma\ge1$, then the Lewis weights of $\A$ are at most $\frac{\gamma d}{n}$. 
\end{lemma}
\begin{proof}
Let $ C>0$ and suppose $\a_i(\A^\top\A)^{-1}\a_i^\top\le C$ for all $i\in[n]$. 
Then for iteration $j$ in the above procedure, we have $\W^{(j)}\preceq C^{1-2^{-j}}\I_n$.

For base case $j=1$, we have $\W^{(1)}_{i,i}=\left(\a_i(\A^\top\A)^{-1}\a_i^\top\right)^{1/2}\le C^{1/2}$ so that $\W^{(1)}\preceq C^{1/2}\I_n$.  
For iteration $j$, we have $\W^{(j)}_{i,i}=\left(\a_i(\A^\top(\W^{(j-1)})^{-1}\A)^{-1}\a_i^\top\right)^{1/2}$, where $\W^{(j-1)}\preceq C^{1-2^{-(j-1)}}\I_n$ by the inductive hypothesis. 
Thus $(\W^{(j-1)})^{-1}\succeq\frac{1}{ C^{1-2^{-(j-1)}}}\I_n$, so 
\begin{align*}
\A^\top(\W^{(j-1)})^{-1}\A&\succeq\frac{1}{ C^{1-2^{-(j-1)}}}\A^\top\A\\
(\A^\top(\W^{(j-1)})^{-1}\A)^{-1}&\preceq C^{1-2^{-(j-1)}}(\A^\top\A)^{-1}\\
(\W^{(j)}_{i,i})^2=\a_i(\A^\top(\W^{(j-1)})^{-1}\A)^{-1}\a_i^\top&\le C^{1-2^{-(j-1)}}\a_i(\A^\top\A)^{-1}\a_i^\top\le C^{2-2^{-(j-1)}}
\end{align*}
Thus $\W^{(j)}_{i,i}\le C^{1-2^{-j}}$ so that $\W^{(j)}\preceq C^{1-2^{-j}}\I_n$. 
The claim follows from setting $C=\frac{\gamma d}{n}$ and the convergence of the algorithm to the Lewis weights by \lemref{lem:lewis:prop}. 
\end{proof}
We now show that splitting a row of a matrix into two rows that sum up to the original row alters the Lewis weights in the natural way. 
\begin{lemma}[Splitting Invariance of Lewis Weights]
\lemlab{lem:lewis:split}
Given $\A=\a_1\circ\ldots\circ\a_n\in \mathbb{R}^{n\times d}$, let 
\[\B\in\mathbb{R}^{(n+1) \times d}=\a_1\circ\ldots\circ\a_{j-1}\circ\b_j\circ\a_{j+1}\circ\ldots\circ\a_n\circ\b_{n+1}\]
have the same rows but with row $j$ have $\b_j=(1-\gamma)\cdot\a_j$ and row $n+1$ have $\b_{n+1}=\gamma\a_j$ for some $\gamma \in [0,1]$ and $j\in[n]$. 
, 
Then we have $w_i(\A)=w_i(\B)$ for any $i \notin \{j,n+1\}$, $w_j(\B)=(1-\gamma)w_j(\A)$ and $w_{n+1}(\B)=\gamma w_j(\A)$. 
\end{lemma}
\begin{proof}
Suppose without loss of generality that $j = n$.
Let $\W\in\mathbb{R}^{n \times n}$ be the diagonal Lewis weight scaling matrix with $\W_{i,i} = w_i(\A)$. 
Let $\barW\in\mathbb{R}^{(n+1) \times (n+1)}$ match $\W$ on the first $n-1$ rows. 
Let $\barW_{n,n} = (1-\gamma) w_n(\A)$ and $\barW_{n+1,n+1} = \gamma w_n(\A)$. 
To prove the claim, it suffices by the uniqueness of Lewis weights to show that $\tau_i(\barW^{-1/2}\B) = \barW_{i,i}$ for all $i \in [n+1]$:

Note that the first $n-1$ rows of $\barW^{-1/2}\B$ are the same as those of $\W^{-1/2}\A$. 
The last two rows are $w_n(\A)^{-1/2} (1-\gamma)^{-1/2} \cdot (1-\gamma)\a_n = w_n(\A)^{-1/2} (1-\gamma)^{1/2} \a_n$ and $w_n(\A)^{-1/2} \gamma^{1/2} \a_n$, respectively.
That is, the last two rows are scaled by $(1-\gamma)^{1/2}$ and $\gamma^{1/2}$, respectively, compared to $\W^{-1/2}\A$. 
Thus we can see that for any vector $\y$, $\norm{\W^{-1/2}\A\y}_2^2 = \norm{\barW^{-1/2}\B\y}_2^2$.
By the maximization characterization of leverage scores, the leverage scores of the first $n-1$ rows of $\W^{-1/2}\A$ are thus identical to those of $\barW^{-1/2}\B$, so that $\W_{i,i}=\barW_{i,i}$ for $1\le n-1$. 

For the last two rows, we have $\tau_n(\barW^{-1/2}\B) = (1-\gamma)\cdot\tau_n(\W^{-1/2}\A) = (1-\gamma)w_n(\A) = \barW_{n,n}$. 
Similarly, $\tau_{n+1}(\barW^{-1/2}\B) = \gamma\cdot\tau_n(\W^{-1/2}\A) = \gamma w_n(\A) = \barW_{n+1,n+1}$. 
Hence, $\tau_i(\barW^{-1/2}\B) = \barW_{i,i}$ for all $i \in [n+1]$, which implies the claim by the uniqueness of Lewis weights. 
\end{proof}

\begin{corollary}[Monotonicity of Lewis Weights II]
\corlab{cor:lewis_mono}
For any $\A\in \mathbb{R}^{n \times d}$, let $\B\in\mathbb{R}^{n \times d}$ have the same rows but with row $j$ reweighted by a factor $(1-\gamma)$ for some $\gamma \in [0,1]$. 
Then for all $i \neq j$, $w_i(\B)\ge w_i(\A)$. 
\end{corollary}
\begin{proof}
Let $\barB\in \R^{(n+1) \times d}$ have row $j$ set to $(1-\gamma)\cdot\a_j$ and row $n+1$ row set to $\gamma \cdot\a_j$. 
Then by \lemref{lem:lewis:split}, for all $i \neq j$, $w_i(\barB) = w_i(\A)$. 
By the monotonicity of Lewis weights through the addition of a new row from \lemref{lem:lewis:monotonic}, we thus have $w_i(\B) \ge w_i(\barB) = w_i(\A)$.
\end{proof}

We now bound the sum of the online Lewis weights by first considering a regularization of the input matrix, which we show only slightly alters each score. 
\begin{lemma}[Bound on Sum of Online Lewis Weights]
\lemlab{lem:online:lewis:sum}
Let $\A=\a_1\circ\ldots\circ\a_n\in\mathbb{R}^{n\times d}$. 
Let $w^{\OL}_i(\A)$ be the online Lewis weight of $\a_i$ with respect to $\A$, for each $i\in[n]$. 
Then $\sum_{i=1}^n w^{\OL}_i(\A)=\O{d\log n\log(\kappa nd)}$, where $\kappa$ is the subset condition number of $\A$. 
\end{lemma}
\begin{proof}
For the first part of the proof, we use the same argument as \lemref{lem:ol:sen:bound}. 
Let $\lambda>0$, $\B_0=\lambda\I_d$, $\B=\underbrace{\B_0\circ\ldots\circ\B_0}_{n\text{ times}}$, and $T$ be an upper bound on the sum of the online leverage scores of any submatrix of $\X:=\B\circ\A$. 
For each $i\in[n]$, let $w_i(\X)$ and $\tau_i(\X)$ be the Lewis weight and leverage score of $\a_i$ with respect to $\X$, respectively. 
By \lemref{lem:online:space} and \lemref{lem:online:reweighting}, there exist a set $S\subseteq[n]$ of size at most $d$ and a diagonal matrix $\W$ with $\frac{n}{2}$ non-unit entries such that $\tau^{\OL}_i(\B\circ(\W\A))\le\frac{2T}{n}$ for all $i\in[n]$, where $T$ is the sum of the online leverage scores of $\A$. 
Let $\C=\B\circ(\W\A)$ and $S$ be the set of all indices $i$ such that $\W_{i,i}=1$. 
For each $i\in[n]$, let $\C_{i+nd}$ be the matrix formed by the first $i+nd$ rows of $\C$. 
Since the online leverage score of each row $\a_j$ with $j\le i$ in $\C$ is at most $\frac{2T}{n}$ and the leverage scores of the first $nd$ rows of $\C_{i+nd}$ is at most $\frac{1}{n}$, then all rows of $\C_{i+nd}$ have leverage score at most $\frac{2T}{n}$. 
By \lemref{lem:uniform_lewis}, all rows of $\C_{i+nd}$ have Lewis weight at most $\frac{2T}{n}$. 

Now for any $i\in S$, we have from the monotonicity of Lewis weights from decreasing the weight of other rows in \corref{cor:lewis_mono} that
\begin{align*}
w_i^{\OL}(\X)\le w_i^{\OL}(\C)= w_i^{\OL}(\C_{i+nd})=w_i(\C_{i+nd})\le\frac{2T}{n}
\end{align*}
so that $\sum_{i\in S}w_i^{\OL}(\X)\le2T$. 
Let $U= \{1,\ldots, n\} \setminus S$ and let $\X_U$ be the matrix $\X$ restricted to the rows whose indices are in $U$. 
By monotonicity of Lewis weights from the addition of new rows in \lemref{lem:lewis:monotonic}, $w_i^{\OL}(\X)\le w_i^{\OL}(\X_U)$.  
By induction, 
\begin{align*}
\sum_{i\in U}w_i^{\OL}(\X)\le\sum_{i\in U}w_i^{\OL}(\X_U) \le 2T\log n
\end{align*}
Thus we have $\sum_i w_i^{\OL}(\X)=\O{T\log n}$, which is $\O{d\log n\log(\kappa nd)}$ by \lemref{lem:online:space}. 

We now assume without loss of generality that all rows of $\A$ have Lewis weight at least $\frac{1}{n}$, since if $\a_n$ has Lewis weight less than $\frac{1}{n}$, then its online Lewis weight is similarly bounded. 
Thus, these rows contribute at most $\O{\frac{\log n}{\eps^2}}$ samples with high probability and removing these rows would only increase the Lewis weights of the other rows by monotonicity~\lemref{lem:lewis:monotonic}. 
Let $\lambda=\frac{\gamma}{n}$, where $\gamma$ is the minimum of the nonzero singular values across all submatrices of $\A$. 
Consider approximating the Lewis weights for $\A$ through \figref{fig:iter:lewis}. 
Since all Lewis weights of $\A$ are at least $\frac{1}{n}$, then by \lemref{lem:lewis:prop}, at most $\O{\log n}$ iterations of the algorithm in \figref{fig:iter:lewis} are necessary to obtain a constant factor approximation to the true Lewis weights. 
The intuition is that the approximation for the Lewis weight of each row $\a_i$ with respect to $\A$ in iteration $j$ of the algorithm is within $\left(1+\frac{j}{n}\right)$ of the approximation for the Lewis weight of row $\a_i$ with respect to $\X$. 
Hence in $\O{\log n}$ iterations, the respective approximations will be within a constant factor approximation. 

Formally, let $\W^{(j)}_{\A}$ be the weight matrix on iteration $j$ of the algorithm in \figref{fig:iter:lewis} on input $\A$ and let $\U^{(j)}_{\X}\circ\W^{(j)}_{\X}$ be the weight matrix on iteration $j$ of the algorithm on input $\X$, so that $\U^{(j)}_{\X}$ has $nd$ rows and $\W^{(j)}_{\X}$ has $n$ rows. 
Thus, we have $\W^{(0)}_{\A}=\W^{(0)}_{\X}=\I_n$ from the initialization of the algorithm in \figref{fig:iter:lewis}. 
The weight for row $\a_i$ is then updated by $\left(\a_i(\A^\top(\W^{(j-1)}_{\A})^{-1}\A)^{-1}\a_i^\top\right)^{1/2}$ for $\W^{(j)}_{\A}$ and by $\left(\a_i(\A^\top(\W^{(j-1)}_{\X})^{-1}\A+\B^\top(\U^{(j-1)}_{\X})^{-1}\B)^{-1}\a_i^\top\right)^{1/2}$. 
Since each Lewis weight is at most $1$ and the algorithm is a contraction mapping, then we certainly have $\U^{(j-1)}_{\X}\preceq\I_{nd}$ so that 
\[\B^\top(\U^{(j-1)}_{\X})^{-1}\B\preceq n\lambda^2\I_{nd}\preceq\frac{\sigma_{\min}(\A^\top\A)}{n}\I_{nd}\preceq\frac{\sigma_{\min}(\A^\top(\W^{(j-1)}_{\A})^{-1}\A)}{n}\I_{nd}.\]
Thus $\W^{(1)}_{\A}\preceq\left(1+\frac{1}{n}\right)\W^{(1)}_{\X}$ and by induction, $\W^{(j)}_{\A}\preceq\left(1+\frac{j}{n}\right)\W^{(j)}_{\X}$. 

Hence the Lewis weight of $\a_n$ with respect to $\A$ is within a constant factor of the Lewis weight of $\a_n$ with respect to $\X$, and thus the \emph{online} Lewis weight of $\a_n$ with respect to $\A$ is within a constant factor of the \emph{online} Lewis weight of $\a_n$ with respect to $\X$. 
Since we have $\sum_i w_i^{\OL}(\X)=\O{d\log n\log(\kappa nd)}$, then it follows that $\sum_i w_i^{\OL}(\A)=\O{d\log n\log(\kappa nd)}$. 
\end{proof}
\noindent
Then from \thmref{thm:coreset:framework}, \thmref{thm:lp:coreset}, and \lemref{lem:online:lewis:sum}, we have the following:
\begin{theorem}[Deterministic Sliding Window Algorithm for $\ell_1$-Subspace Embedding]
\thmlab{thm:det:lp}
Let $\r_1,\ldots,\r_n\in\mathbb{R}^{1\times d}$ be a stream of rows with the subset condition number $\kappa$. 
Let $\eps>0$ and $\A=\r_{n-W+1}\circ\ldots\circ\r_n$ be the matrix consisting of the $W$ most recent rows. 
There exists a deterministic algorithm that stores $\O{\frac{d}{\eps^2}\log^5 n\log(\kappa nd)}$ rows and outputs a matrix $\M$ such that $(1-\eps)\norm{\A\x}_1\le\norm{\M\x}_1\le(1+\eps)\norm{\A\x}_1$ for all $\x\in\mathbb{R}^d$. 
\end{theorem}
Note that \thmref{thm:coreset:framework} also provides an approach for a \emph{randomized} $\ell_1$-subspace embedding sliding window algorithm that improves upon the space requirements of \thmref{thm:sw:lp}, by using online coresets \emph{randomly} generated sampling rows with respect to their online Lewis weights. 
Moreover, recall that in some settings, the online model does not require algorithms to use space sublinear in the size of the input. 
In these settings, \lemref{lem:online:lewis:sum} could also potentially be useful in a row-sampling based algorithm for online $\ell_1$-subspace embedding that improves upon the sample complexity of \thmref{thm:online:lp}.

\paragraph*{Fast Online Coreset Construction.}
To quickly obtain an online coreset for spectral approximation, we derandomize the $\onlinebss$ algorithm of \cite{CohenMP16} with a $\O{\log n}$ overhead. 
Their algorithm requires an input $\lambda\ge 0$ and the rows of $\A=\a_1\circ\ldots\circ\a_n\in\mathbb{R}^{n\times d}$ in a stream and uses upper barrier and lower barrier method of \cite{BSS12} to design a probability distribution for each row that depends on the previously sampled rows but will \emph{always} output a matrix $\M$ such that $(1-\eps)(\A^\top\A+\lambda\I)\preceq\M^\top\M+\lambda\I\preceq(1+\eps)(\A^\top\A+\lambda\I)$. 
Instead, \cite{CohenMP16} bounds the expected number of rows sampled by $\onlinebss$ based on the sampling probabilities of each row. 

First, we observe that the sampling probability for each row in $\onlinebss$ can be rounded up to a power of $2$ between $\frac{1}{n}$ and $1$. 
As before, the rows with sampling probability less than $\frac{1}{n}$ incur only $\O{1}$ additional rows sampled in expectation. 
Although the sampling probability $p_i$ of row $\a_i$ is a random variable that depends on the previous rows sampled, \cite{CohenMP16} bounds the expectation of $p_i$ by the online (ridge) leverage score $\tau^{\OL}_i$, so that the expected number of sampled rows is at most $\sum_{i=1}^n\tau_i$. 
Now consider a process in which we pick a random threshold $t\in [0,1]$ and then sample row $\a_i$ if $p_i>t$. 
It can be shown that the expectation of $p_i$ is still at most the online (ridge) leverage score $\tau^{\OL}_i$. 
Since there are only $\log n$ distinct sampling probabilities, we can derandomize this process by trying all $\log n$ values of $t$. 
For completeness, we describe $\onlinebss$ and our derandomization in full detail in \appref{sec:onlinebss}.

\section*{Acknowledgements}
Vladimir Braverman is supported in part by the National Science Foundation under Grants No. 1447639, 1650041, and 1652257, CISCO faculty award, and by the ONR Award N00014-18-1-2364. 
Petros Drineas is supported in part by NSF IIS-1661760 and NSF IIS-1319280.
Jalaj Upadhyay is supported in part by NSF 1546482, NSF IIS-1838139 and CISCO faculty award. 
David P. Woodruff is supported by the National Science Foundation under Grant No. CCF-1815840, the National Institute of Health (NIH) grant 5R01 HG 10798-2, subaward through Indiana University Bloomington, as well as Office of Naval Research (ONR) grant N00014-18-1-2562. 
Samson Zhou is supported in part by NSF CCF-1649515 and by a Simons Investigator Award of D. Woodruff. 
We thank Yi Li for pointing out an issue in an earlier version regarding subset condition number dependency. 

\def\shortbib{0}
\bibliography{references}
\bibliographystyle{alpha}

\appendix
\section{Smooth Histograms and RandNLA Functions}
\applab{sec:smooth-histogram}
In this section, we define the popular smooth histogram framework for sliding window algorithms and then give counterexamples showing the approach is not amenable to many interesting linear algebraic functions. 

\subsection{Background on Smooth Histograms}
\applab{sec:smooth-background}
An initial framework for approaching problems in the sliding window model is the {\em exponential histogram} data structure, introduced by Datar\etal~\cite{DatarGIM02}. 
Given a function $f$ to approximate in the sliding window model, the exponential histogram partitions the data stream into ``buckets'', time intervals for which the evaluation of $f$ on the data in each partition is exponentially increasing. 
For example, suppose we are given a data stream of integers and we want to approximate the number of ones in the sliding window within a factor of $2$. 
In the exponential histogram data structure, the smallest bucket consists of all elements in the data stream from the most recent element to the most recent element whose value is one. 
The next bucket would consist of all previous elements until two elements whose values are one are seen. 
Similarly, the $i$\th bucket consists of the previous elements until $2^i$ instances of ones are seen. 
The key observation is that because the buckets are exponentially increasing by powers of two, the starting point of the sliding window falls inside some bucket, and it will provide a $2$-approximation to the number of ones seen in the sliding window even though it does not know exactly where the starting point is. 
Datar and Motwani \cite{DatarM07} show that the exponential histogram framework is applicable to the class of ``weakly additive'' functions. 
Namely, if we let $A$ and $B$ be adjacent buckets and assume that $0 \leq f(A)\le \poly(N)$, where $N$ is the length of the data stream, then a function is weakly additive if there exists some fixed constant $C_f\ge 1$ such that $f(A)+f(B) \leq f(A\cup B) \le C_f(f(A)+f(B))$ holds for arbitrary $A$ and $B$, where we recall that $A\cup B$ represents the concatenation of $A$ and $B$. 
Moreover, if there exists a sketch of $f$, as well as a ``composition'' function that computes the sketch of $f(A\cup B)$ from the sketches of $f(A)$ and $f(B)$, then the exponential histogram framework provides a sliding window algorithm to approximate $f$. 

Since the buckets in the exponential histogram data structure consist of disjoint elements, a crucial underlying requirement is 
that an approximation of $f$ must be deducible from the merger of the information from these buckets. 
For example, it is not clear how to maintain buckets for the goal of approximating the geometric mean of a sliding window. 
To that effect, Braverman and Ostrovsky \cite{BravermanO07} define the notion of a smooth function, and provide the smooth histogram data structure as a framework for approximating smooth functions of sliding windows. 
They also show that the class of smooth functions contains the class of weakly additive functions, as well as a number of other functions, such as the geometric mean. 

\paragraph{Smooth Histogram.} 
Given adjacent buckets $A$, $B$, and $C$, a smooth function demands that if $(1-\beta)f(A\cup B)\le f(B)$, then $(1-\alpha)f(A\cup B\cup C)\le f(B\cup C)$ for some constants $0<\beta\le\alpha<1$. 
Informally, a smooth function has the property that once a suffix of a data stream becomes a good approximation, then it \emph{always} remains a good approximation, even with the arrival of new elements in the stream. 
With this definition of smooth function in mind, the smooth histogram data structure maintains a number of ``checkpoints'' throughout the data stream. 
Each checkpoint corresponds to a sketch of all the elements seen from the checkpoint until the most recently arrived element. 
Unlike the exponential histogram, the most recently arrived element  impacts all sketches in the smooth histogram. 
A checkpoint is created with the arrival of each new element and checkpoints are discarded when their corresponding sketches get ``too close'' to the next checkpoint. 
That is, when the corresponding sketches of two checkpoints produce values that are within $(1-\beta)$ of each other, the later checkpoint is discarded, since by the property of smooth functions, the two checkpoints would thereafter always produce values that are within $(1-\alpha)$ of each other. 
This implies that, if the function is polynomially bounded, then the smooth histogram data structure only needs a logarithmic number of checkpoints.
Moreover, Braverman and Ostrovsky \cite{BravermanO07} extend their results to the case where the sketch only provides an approximation to the evaluation of the function.


\paragraph{Smooth Functions.}
We use the notation $B \subseteq_s A$ if the stream of elements indexed by $B$ are a suffix of stream of elements indexed by $A$ (see~\figref{fig:smooth} for an example). 

\begin{figure*}[!htb]
\centering
\begin{tikzpicture}[scale=0.5]
\draw (-8cm,0cm) rectangle+(21cm,1cm);
\filldraw[shading=radial,inner color=white, outer color=gray!90, opacity=0.2] (-8cm,0cm) rectangle+(21cm,1cm);
\draw[decorate,decoration={brace,amplitude=10pt}](-8cm,1.2cm) -- (13cm,1.2cm);
\node at (2.5cm,2.6cm){$A$};
\draw (-5cm,-1cm) rectangle+(18cm,1cm);
\filldraw[shading=radial,inner color=white, outer color=gray!90, opacity=0.2] (-5cm,-1cm) rectangle+(18cm,1cm);
\draw[decorate,decoration={brace,mirror,amplitude=10pt}](-5cm,-1.2cm) -- (13cm,-1.2cm);
\node at (4cm,-2.6cm){$B$};

\draw (13cm,-0cm) rectangle+(6cm,1cm);
\filldraw[thick, top color=white,bottom color=blue!50!] (13cm,-0cm) rectangle+(6cm,1cm);
\draw[decorate,decoration={brace,amplitude=10pt}](13cm,1.2cm) -- (19cm,1.2cm);
\node at (16cm,2.6cm){$C$};
\draw (13cm,-1cm) rectangle+(6cm,1cm);
\filldraw[thick, top color=white,bottom color=blue!50!] (13cm,-1cm) rectangle+(6cm,1cm);
\draw[decorate,decoration={brace,mirror,amplitude=10pt}](13cm,-1.2cm) -- (19cm,-1.2cm);
\node at (16cm,-2.6cm){$C$};
\end{tikzpicture}
\caption{$A$, $B$ and $C$ are substreams, where $B$ is a suffix of $A$ and $C$ is adjacent to $A$ and $B$. 
Smoothness says that if $B$ is a ``good'' approximation of $A$, then $B\cup C$ will be a ``good'' approximation of $A\cup C$ for any $C$.}
\figlab{fig:smooth}
\end{figure*}
\begin{definition}[Smooth function]
\cite{BravermanO07}
\deflab{def:smooth}
A function $f \ge 1$ is $(\alpha,\beta)$-smooth if it has the following properties:
\begin{description}
\item [Monotonicity]
$f(A)\ge f(B)$ for $B\subseteq_s A$ ($B$ is a suffix of $A$)
\item [Polynomial boundedness]
There exists $c>0$ such that $f(A)\le n^c$.
\item [Smoothness]
There exists $\alpha\in(0,1)$, $\beta\in(0,\alpha]$ so that if $B\subseteq_s A$ and $(1-\beta)f(A)\le f(B)$, then $(1-\alpha)f(A\cup C)\le f(B\cup C)$ for any adjacent $C$.
\end{description}
\end{definition}
The {\em smooth histogram} data structure estimates smooth functions in the sliding window model. 

\begin{definition}[Smooth Histogram \cite{BravermanO07}]
Let $g$ be a function that maintains a $(1+\eps)$-approximation of an $(\alpha,\beta)$-smooth function $f$ that takes as input a starting index and ending index in the data stream. 
The approximate smooth histogram is a structure that consists of an increasing set of indices $X_N=\{x_1,\ldots,x_s=N\}$ and $s$ instances of an algorithm $\Lambda$, namely $\Lambda_1,\ldots,\Lambda_s$ with the following properties:
\begin{enumerate}
\item
$x_1$ corresponds to either the beginning of the data stream, or an expired point.
\item
$x_2$ corresponds to an active point. 
\item
For all $i<s$, one of the following holds:
\begin{enumerate}
\item
$x_{i+1}=x_i+1$ and $g(x_{i+1},N)<\left(1-\frac{\beta}{2}\right)g(x_i,N)$.
\item
$(1-\alpha)g_(x_i,N)\le g(x_{i+1},N)$ and if $i+2\le s$ then $g(x_{i+2},N)<\left(1-\frac{\beta}{2}\right)g(x_i,N)$.
\end{enumerate}
\item
$\Lambda_i=\Lambda(x_i,N)$ maintains $g(x_i,N)$.
\end{enumerate}
\end{definition}

Unfortunately, despite being quite general, the smooth histogram frameworks cannot be applied to many interesting problems that have been extensively studied in the streaming model, such as clustering~\cite{BravermanLLM15,BravermanLLM16,BravermanLUZ19}, submodular maximization~\cite{ChenNZ16,EpastoLVZ17}, or heavy-hitter detection~\cite{LeeT06b,BGO14,BravermanGLWZ18,upadhyay2019sublinear}. 
We now show that smooth histogram cannot be applied to many interesting numerical linear algebraic functions. 

\subsection{Lack of $(\alpha,\beta)$-smoothness of RandNLA Functions}
\applab{sec:counters}
In this section, we show that the spectral norm, vector induced matrix norms, linear and generalized regression, and low-rank approximation are not amenable for the smooth histogram framework. 

We first prove that low-rank approximation is not smooth for any meaningful parameters $(\alpha, \beta)$ in~\defref{def:smooth}, even when the best low-rank approximations are nonzero.
\newcommand{\lemlranotsmooth}
{
Low-rank approximation is not smooth for any meaningful parameters $(\alpha, \beta)$  in~\defref{def:smooth} that gives us a constant factor approximation. 
}
\begin{lemma}
\lemlab{ex:lra:not:smooth}
\lemlranotsmooth
\end{lemma}
\begin{proof}
Our proof constructs a matrix and a vector explicitly.  
Let $\e_i$ be the elementary row vector with entry one in the $i\th$ position and zero elsewhere. 
Given $0<\alpha<1$, let $d=\frac{2}{\alpha}$. 
Let $S_A$ be the data stream whose first element is $2d\e_1$, second element is $\e_2$, followed by data stream $S_B$, which is a suffix of $S_A$, i.e., $S_B \subseteq_s S_A$. 
Then suppose $B$ consists of the elements $2d\e_3$, followed by $\e_{i+3}$ for $1\le i\le d$. 
Finally, let $S_C$ be the data stream consisting of the single element $2d\e_j$, where $j=4+d$. 
Then the corresponding matrices appear as below:
\[
\begin{bmatrix}
$2d$ & & & & & & & & &\\
& 1 & & & & & & & &\\\hline
& & $2d$ & & & & & & &\\
& & & 1 & & & & & &\\
& & & & 1 & & & & &\\
& & & & & \ddots & & & &\\
& & & & & & 1 & & &\\\hline
& & & & & & & $2d$ & &\\
& & & & & & & & 0 &\\
& & & & & & & & & \ddots \\
\end{bmatrix}
,\]
where we have the following matrices
\begin{align*}
\A = \begin{bmatrix}
$2d$ & & & & & &  \\
& 1 & & & & &  \\\hline
& & $2d$ & & & & \\
& & & 1 & & &  \\
& & & & 1 & &  \\
& & & & & \ddots &  \\
& & & & & & 1  
\end{bmatrix}, 
\B = \begin{bmatrix}
$2d$ & & & &\\
& 1 & & &\\
& & 1 & &\\
& & & \ddots &\\
& & & & 1 
\end{bmatrix},
\text{ and } 
\C = \begin{bmatrix}
$2d$ & &\\
 & 0 &\\
 & & \ddots \\
\end{bmatrix}.
\end{align*}

Then for $k=2$, the best rank $k$ approximation of $\A$ consists of  two rows containing $2d$ so that
\[\min_{\substack{\X \in \R^{N \times n}\\\rank(\X)=2}}\norm{\A-\X}_F=\sqrt{d+1}.\]
The best rank $k$ approximation of matrix formed by $S_B$ consists of the row containing $2d$ and any other elementary row in $\B$ so that
\[\min_{\substack{\Y \in \R^{N \times n}\\\rank(\Y)=2}}\norm{\B-\Y}_F=\sqrt{d-1}.\]
Hence, the ratio of the best low-rank approximation of $\B$ to the best low-rank approximation of $\A$ is
\[\frac{\sqrt{d-1}}{\sqrt{d+1}}>\frac{d-1}{d+1}=1-\frac{2}{d+1}>1-\alpha.\]

Now, let $\C$ represent the matrix corresponding to stream $S_{A\cup C}$ and let $\mathbf{D}$ represent the matrix corresponding to stream $S_{B\cup C}$.
Then the best rank $k$ approximation of $S_{A\cup C}$ consists of two rows containing $2d$ so that
\[\min_{\substack{\X  \in \R^{N \times n}\\ \rank(\X)=2}}\norm{\C-\X}_F=\sqrt{4d^2+d+1},\]
while the best rank $k$ approximation of $S_{B\cup C}$ consists of two rows containing $2d$ so that
\[\min_{\substack{\Y  \in \R^{N \times n}\\ \rank(\Y)=2}}\norm{\mathbf{D}-\Y}_F=\sqrt{d}.\]
Thus, the ratio of the best low-rank approximation of $S_{B\cup C}$ to the best low-rank approximation of $S_{A\cup C}$ is at least $2$, i.e., $\beta \leq -1$. 
Therefore, low-rank approximation is not smooth.
\end{proof}

We now show that $\ell_p$ regression is not smooth as per the smooth histogram framework. 
Recall that for $\A\in\mathbb{R}^{N\times n}$ and $\B\in\mathbb{R}^{N\times d}$, the generalized $\ell_p$ regression problem is the minimization problem
\[\min_{\X\in\mathbb{R}^{n\times d}}\norm{\A\X-\B}_p,\]
where $\norm{\cdot}_p$ denotes the entrywise $\ell_p$ norm, that is, $\norm{\X}_p = \paren{\sum_{i,j} |\X_{i,j}|^p }^{1/p}$.

In our setting, each update to the data stream consists of a row vector $\a_i\in\mathbb{R}^d$ and an element $b_i$, i.e., the underlying matrix $\A$ represented by a sliding window of size $W$ consists of the rows
\begin{align*}
\A=
\begin{bmatrix}
\a_{t-W+1} & 
\a_{t-W+2} & 
\cdots &
\a_{t}
\end{bmatrix}^\top \\
 \b=
\begin{bmatrix}
b_{t-W+1} & 
b_{t-W+2} & 
\cdots &
b_{t}
\end{bmatrix}^\top.
\end{align*}

We next show that $\ell_p$ regression is not smooth as per~\defref{def:smooth} for any reasonable parameters $(\alpha,\beta)$ (see~\lemref{regression}).
\newcommand{\lemregression}
{
$\ell_p$ regression is not smooth as per~\defref{def:smooth} for any reasonable parameters $(\alpha,\beta)$  that gives us a constant factor approximation.
}

\begin{lemma}
\label{lem:regression}
\lemregression
\end{lemma}

\begin{proof}
Let $0<\alpha<1$. 
Let $A$ be the data stream whose first element is the row $\{\a_1=\{100,0,0,0,0\},b_1=100\}$, second element is the row $\{\a_2=\{0,\alpha,0,0,0\},b_2=0\}$, followed by data stream $B$, which is a suffix of $A$. 
Then suppose $B$ consists of the elements $\{\a_3=\{0,0,1,0,0\},b_3=1\}$, followed by $\{\a_4=\{0,0,0,1,0\},b_4=0\}$. 
Finally, let $C$ be the data stream consisting of the single element $\{\a_5=\{0,0,0,0,1000\},b_5=2000\}$ 
Then the corresponding matrices appear as below:
\[
\A=
\begin{bmatrix}
100 & 0 & 0 & 0 & 0\\
0 & \alpha & 0 &  0 & 0\\\hline
0 & 0 & 1 & 0 & 0\\
0 & 0 & 0 & 1 & 0 \\\hline
0 & 0 & 0 & 0 & 1000
\end{bmatrix},
\quad 
\mathbf b=
\begin{bmatrix}
100\\
0\\\hline
1\\
0\\\hline
2000
\end{bmatrix}
,
\]
where 
\begin{align*}
\A_1 := 
\begin{bmatrix}
100 & 0 & 0 & 0 & 0\\
0 & \alpha & 0 &  0 & 0
\end{bmatrix}, \qquad
\A_2 :=
\begin{bmatrix}
0 & 0 & 1 & 0 & 0\\
0 & 0 & 0 & 1 & 0
\end{bmatrix}, \quad
\text{ and }\quad
\A_3 := \begin{bmatrix}
0 & 0 & 0 & 0 & 1000
\end{bmatrix}
\end{align*}
and $\b_1$, $\b_2$, and $\b_3$ are the corresponding rows of the vector $b$.
Let the matrix $\A_1$ and vector $\b_1$ represent data stream $S_A$, $\A_2$ and $\b_2$ represent data stream $S_B$, $\A_3$ and $\b_3$ represent $A\cup C$, and $\A_4$ and $\b_4$ represent $S_B\cup S_C$. 
Finally, let $\mathcal{Z}_i=\argmin_{\x\in\mathbb{R}^n}\norm{\A_i\x-\b_i}_p$ for $1\le i\le 4$. 
Then one can verify that $\mathcal{Z}_1=1+\alpha$ and $\mathcal{Z}_2=1$. 
On the other hand, $\mathcal{Z}_3>100$ but $\mathcal{Z}_4=\sqrt[p]{2^p+1}$. 
Thus, $\ell_p$ regression is not smooth as per~\defref{def:smooth} for any reasonable parameters $(\alpha,\beta)$  that gives us a constant factor approximation.
\end{proof}

Recall that the vector induced $\ell_p$ matrix norm is defined as 
\[\norm{\A}_{p}=\underset{\norm{\x}_p=1}{\max}\norm{\A\x}_p,\]
and for $p=2$, it is the same as the spectral norm (Schatten-$\infty$ norm). 
We now show that vector induced norms are not smooth for all values of $p$. 
\begin{theorem}
The vector induced matrix norm $\norm{\cdot}_p$ is not a smooth function as per~\defref{def:smooth} for a meaningful parameters $(\alpha,\beta)$ for constant factor approximation. 
\end{theorem}
\begin{proof}
We give explicit construction of streams to show that there exists a stream for which the vector induced matrix norm is not smooth as per~\defref{def:smooth}.  
Let $\e_1, \cdots, \e_4$ be the standard basis of $\R^n$. 
Let $A$ be the data stream consisting of $\{\e_1\}$, followed by the suffix $B$, the data stream consisting of $\{\e_2\}$. 
Let $C$ be the data stream consisting of $\{\e_1\}$. 
Define $\A$ to be the matrix representing $A$ and $\B$ to be the matrix representing $B$. 
Let $\mathbf{R}$ be the matrix representing $A\cup C$ and $\S$ be the matrix representing $B\cup C$.
Thus, 
\[\W:=\A^\top\A=
\begin{pmatrix}
1 & 0\\
0 & 1
\end{pmatrix},
\X:=\B^\top\B=
\begin{pmatrix}
0 & 0\\
0 & 1
\end{pmatrix}\]
\[
\Y:=\mathbf{R}^\top\mathbf{R}=
\begin{pmatrix}
4 & 0\\
0 & 1
\end{pmatrix},
\Z:=\S^\top\S=
\begin{pmatrix}
1 & 0\\
0 & 1
\end{pmatrix}
\]
Observe that
\[\norm{\W}_p=\underset{\norm{\x}_p=1}{\max}\norm{\W\x}_p 
\]
Thus, $\norm{\W}_p=1$. 
Similarly, $\norm{\X}_p=1$ so that 
$(1-\alpha)\norm{\W}_p\le \norm{\X}_p$ for any $0<\alpha<1$. 
On the other hand, $\norm{\Y}_p=4$ and $\norm{\Z}_p=1$. 
Hence, the vector induced matrix norm $\norm{\cdot}_p$ is not a smooth function as per~\defref{def:smooth}.
\end{proof}

\subsection{A Generalization of Smooth Histograms for Spectral Approximation}
\applab{app:histogram}
In this section, we give a deterministic sliding window algorithm for spectral approximation that is neither space nor time optimal, but provides a natural generalization of the smooth histogram framework of~\cite{BravermanO07} to spectral approximation in the sliding window model. 
\begin{algorithm}[!htb]
\caption{Projection-cost preservation for low-rank matrices in the sliding window model}
\alglab{alg:histogram}
\begin{algorithmic}[1]
\Require{A stream of rows $\r_1,\ldots,\r_n\in\mathbb{R}^d$, window size $W$, and an accuracy parameter $\eps>0$}
\Ensure{Projection-cost preservation for low-rank matrices in the sliding window model.}
\State{$\M_0\gets 0^{d\times d}$}
\For{each row $\r_t$}
\State{Suppose $\M_0,\M_1,\ldots,\M_s$ are defined}
\State{$\M_{s+1}\gets\r_t^\top\r_t$, $\t_{s+1}=t$}
\Comment{Keep timestamp for $\M_{s+1}$}
\For{$i=s$ to $i=1$}
\Comment{Update sketches with $\r_t$}
\State{$\M_i\gets\M_i+\r_t^\top\r_t$}
\EndFor
\For{$i=s$ to $i=2$}
\If{$\M_{i-1}^\top\M_{i-1}\preceq(1+\eps)\M_{i+1}^\top\M_{i+1}$}
\State{Delete $\M_i$ and $t_i$ and relabel indices}
\EndIf
\EndFor
\If{$t_2\le t-W+1$}
\Comment{$\M_2$ has expired}
\State{Delete $\M_1$ and $t_1$ and relabel indices}
\EndIf
\EndFor
\State{\Return $\M_1$}
\end{algorithmic}
\end{algorithm}

\begin{theorem}
\thmlab{thm:sw:histogram}
Let $\r_1,\ldots,\r_n\in\mathbb{R}^{d}$ be a stream of rows and $\kappa$ be the condition number of the stream. 
Let $W>0$ be a window size parameter and $\A=\r_{n-W+1}\circ\ldots\circ\r_n$ be the matrix consisting of the $W$ most recent rows. 
Given a parameter $\eps>0$, there exists an algorithm that outputs a matrix $\M$ such that $\A^\top\A \preceq \M \preceq (1+\eps)\A^\top\A$ and uses $\O{\frac{d^3}{\eps}\log\kappa}$ words of space.
\end{theorem}
\begin{proof}
Consider \algref{alg:histogram}. 
Note that either $t_1=n-W+1$ or $t_1<n-W+1$, since if $t_1>n-W+1$, then there would be another matrix with an earlier timestamp.  
In the first case, we have $\M_1=\A^\top\A$ since the outer product of each row $\r_i$ with $i\ge n-W+1$ is added to $\M_1$ at time $i$. 

In the second case, we have that $t_2>n-W+1$ or else $\M_1$ would have been deleted. 
That means at some point $t$ we must have had $\M^{(t)}_1\preceq(1+\eps)\M^{(t)}_2$ to delete all matrices with timestamps between $t_1$ and $t_2$, where $\M^{(t)}_j:=\sum_{i=t_j}^t\r_i^\top\r_i$ is the value at time $t$ of the sketch that is $\M_j$ at the end of the stream.  
Hence, 
\[\M_1=\M^{(t)}_1+\sum_{i=t}^n\r_i^\top\r_i\preceq(1+\eps)\M^{(t)}_2+\sum_{i=t}^n\r_i^\top\r_i\preceq(1+\eps)\M_2.\]
Since $\M_1\succeq\A^\top\A\succeq\M_2$, it follows that $\A^\top\A \preceq \M \preceq (1+\eps)\A^\top\A$. 

Each matrix has dimension $d\times d$. 
Moreover, $\M_i\succeq(1+\eps)\M_{i+2}$ for each index $i$, so that some singular value has increased by a factor of $(1+\eps)$. 
Since there are at most $d$ singular values and they can increase by a factor of at most $\kappa$, then there are at most $\frac{d}{\eps}\log\kappa$ such matrices, and the space complexity follows. 
\end{proof}
\noindent
We now give the proof of \thmref{thm:pcp:low}. 
\newline\noindent
\begin{remindertheorem}{\thmref{thm:pcp:low}}\thmpcplow
\end{remindertheorem}
\begin{proof}
Consider \algref{alg:pcp:low}. 
Similar to \algref{alg:histogram}, $\r_t^\top\r_t$ is added to each sketch and sketch $i$ is deleted if $\M_{i-1}\preceq(1+\eps)\M_{i+1}$. 
Hence we have $\A^\top\A \preceq \M_1\preceq (1+\eps)\A^\top\A$ for the output $\M_1$ of \algref{alg:pcp:low}. 
Moreover, ome singular value has increased by a factor of $(1+\eps)$ between $\M_{i-1}$ and $\M_{i+1}$, by similar reasoning to \thmref{thm:sw:histogram}. 
Since there are at most $2k$ singular values and they can increase by a factor of at most $\kappa$, then there are at most $\frac{2k}{\eps}\log\kappa$ such matrices. 
Thus $\R_1$ contains exactly the row span of $\A$ and since $\R_1\supseteq\R_2\supseteq\ldots$, it suffices to store $2k$ rows of $\A$ to track all the matrices $\R_i$. 
On the other hand, $\U_i\in\mathbb{R}^{2k\times 2k}$, but there are $\O{\frac{k}{\eps}\log\kappa}$ such matrices, so the total additional working space to maintain the matrices is $\O{\frac{k^3}{\eps}\log\kappa}$. 

Finally, note that since $\A$ has integer entries bounded in magnitude by $n^{\O{1}}$, then the characteristic polynomial has coefficients that are bounded in magnitude by $n^{\O{k}}$. 
Moreover, $n^{\O{1}}\ge\norm{\A}_F\ge\sigma_{\max}(\A)$ so then $\sigma_{\min}(\A)\ge\frac{1}{n^{\O{k}}}$. 
Thus, $\log\kappa=\O{k\log n}$ and the space complexity follows. 
\end{proof}
\section{Extended Proofs}
\subsection{Online Whack-a-Mole}
\applab{app:whackamole}
For a given matrix $\A$ prepended by a regularization matrix $\B$, which is some copies of the identity matrix, the goal is to give weights to the rows of $\A$ so that the online leverage score of each row of $\A$ is uniformly bounded by some $\frac{\gamma T}{n}$, where $\gamma>1$ and $T$ is some upper bound on the sum of the online leverage scores of a matrix. 
We give the online whack-a-mole algorithm in \algref{alg:whack:mole}. 
\begin{algorithm}[!htb]
\caption{Online Whack-a-Mole algorithm for uniformly bounded online leverage scores}
\alglab{alg:whack:mole}
\begin{algorithmic}[1]
\Require{Matrix $\A=\a_1\circ\ldots\circ\a_n\in\mathbb{R}^{n\times d}$, regularization parameter $\lambda>0$, parameter $\gamma>1$}
\Ensure{Reweighting matrix $\W$}
\State{$\B_0\gets\lambda\I_d$, $\B=\underbrace{\B_0\circ\ldots\circ\B_0}_{n\text{ times}}$}
\State{Let $T$ be the sum of the online leverage scores of $\B\circ\A$.}
\State{$\W\gets\I_n$}
\For{$i=1$ to $i=n$}
\State{Let $\tau_{nd+i}^{\OL}(\B\circ(\W\A))$ denote the online leverage score of row $nd+i$ of $\B\circ(\W\A)$.}
\If{$\tau_{nd+i}^{\OL}(\B\circ(\W\A))>\frac{\gamma T}{n}$}
\State{Decrease $\W_{i,i}$ so that $\tau_{nd+i}^{\OL}(\B\circ(\W\A))=\frac{\gamma T}{n}$}
\EndIf
\EndFor
\State{\Return $\W$}
\end{algorithmic}
\end{algorithm}
The intuition by the online whack-a-mole algorithm is simple. 
We start with a weight matrix $\W=\I_n$ and iterate through the rows of $\A$. 
If there exists a row $\a_i$ of $\A$ whose online leverage score with respect to $\B\circ(\W\A)$ is larger than $\frac{\gamma T}{n}$, then $\W_{i,i}$ is decreased until $\a_i$ has online leverage score $\frac{\gamma T}{n}$. 
Because (online) leverage scores are lower semi-continuous~\cite{CohenLMMPS15}, such a weight $\W_{i,i}$ must exist. 
Moreover, lowering the weight of a row $\a_j$ does not affect the online leverage score of a row $\a_i$ for $i<j$. 
Thus, it remains to show that the number of rows whose weights are lowered is at most $\frac{n}{\gamma}$. 
\noindent
We now prove \lemref{lem:online:reweighting}. 
\newline\noindent
\begin{remindertheorem}{\lemref{lem:online:reweighting}}
\cite{CohenLMMPS15}
\whackamole
\end{remindertheorem}
\begin{proof}
By the above argument, it suffices to show that the number of entries $i$ such that $\W_{i,i}\neq 1$ is at most $\frac{n}{\gamma}$. 
Let $\X=\B\circ\A$ and for each $i\in[n]$, let $\zeta^{\OL}_i(\X)$ denote the online leverage score of row $\a_i$ with respect to matrix $\X$. 
Note that decreasing the weight of a row $\a_i$ increases the online leverage score of any subsequent rows $\a_j$ with $j>i$. 
Thus we have $\sum_{i\in S}\frac{\gamma T}{n}\le\sum_{i\in S}\zeta_i^{\OL}(\X)\le T$, so that the number of indices in $S$ is at most $\frac{n}{\gamma}$. 
\end{proof}

\subsection{Derandomization of $\onlinebss$}
\applab{sec:onlinebss}
In this section, we formally derandomize the $\onlinebss$ algorithm of \cite{CohenMP16} with a $\O{\log n}$ overhead. 
The $\onlinebss$ algorithm, described in \algref{alg:online:bss}, takes parameters $\eps\in(0,1)$ and $\lambda\ge 0$, as well as the rows of $\A=\a_1\circ\ldots\circ\a_n\in\mathbb{R}^{n\times d}$ in a stream. 
The algorithm uses the upper barrier and lower barrier method of \cite{BSS12} to design a probability distribution for each row that depends on the previously sampled rows but will \emph{always} output a matrix $\M$ such that $(1-\eps)(\A^\top\A+\lambda\I)\preceq\M^\top\M+\lambda\I\preceq(1+\eps)(\A^\top\A+\lambda\I)$. 
Rather than show the correctness over the randomness of the algorithm, \cite{CohenMP16} must instead bound the expected number of rows sampled by $\onlinebss$ over the randomness of the algorithm. 

\begin{algorithm}[!htb]
\caption{$\onlinebss$}
\alglab{alg:online:bss}
\begin{algorithmic}[1]
\Require{Stream of rows $\a_1,\ldots,\a_n\in\mathbb{R}^{1\times d}$, accuracy $\eps\in(0,1)$, regularization $\lambda\ge0$.}
\Ensure{Additive-multiplicative spectral approximation to $\A:=\a_1\circ\ldots\circ\a_n$.}
\State{$c_U\gets\frac{2}{\eps}+1$, $c_L=\gets\frac{2}{\eps}-1$}
\State{$\M\gets\emptyset$, $\B^U\gets\lambda\I$, $\B^L\gets-\lambda\I$}
\For{each row $\a_t$}
\State{$\X^U\gets\B^U-\M^\top\M$, $\X^L\gets\M^\top\M-\B^L$}
\State{$p_t\gets\min\left(1,c_U\a_i(\X^U)^{-1}\a_i^\top+c_L\a_i(\X^L)^{-1}\a_i^\top\right)$}
\State{$\B^U\gets\B^U+(1+\eps)\a_i^\top\a_i$, $\B^L\gets\B^U+(1-\eps)\a_i^\top\a_i$}
\State{With probability $p_t$, $\M\gets\M\circ\frac{\a_t}{\sqrt{p_t}}$}
\EndFor
\State{\Return $\M$.}
\end{algorithmic}
\end{algorithm}

\cite{CohenMP16} start with the upper and lower barrier matrices $\B^U=\lambda\I$ and $\B^L=-\lambda\I$, in the Loewner order sense.  
Upon the arrival of each row $\a_i$, $\B^U$ and $\B^L$ are incremented by $(1+\eps)\a_i^\top\a_i$ and $(1-\eps)\a_i^\top\a_i$ respectively, as to remain upper and lower barriers for $\M^\top\M+\lambda\I$ as long as it remains a good approximation to $\A^\top\A+\lambda\I$ as rows of $\A$ arrive. 
These upper and lower barrier matrices are then used in conjunction with $\M$ to compute the sampling probability $p_i$ of row $\a_i$. 
Namely, $p_i=\min\left(1,c_U\a_i(\X^U)^{-1}\a_i^\top+c_L\a_i(\X^L)^{-1}\a_i^\top\right)$, where $\X^U=\B^U-\M^\top\M$ and $\X^L=\M^\top\M-\B^L$ describe how far $\M^\top\M+\lambda\I$ is from the barriers $\X^U$ and $\X^L$, respectively. 
As $\M^\top\M+\lambda\I$ approaches one of the barriers, $\X^U$ or $\X^L$ approaches zero, so that the following rows have high sampling probability, which forces $\M^\top\M+\lambda\I$ to remain a good approximation to $\A^\top\A+\lambda\I$. 

To analyze the expected number of rows sampled, \cite{CohenMP16} defines 
\begin{align*}
\C^U_{i,j}&=\lambda\I+\frac{\eps}{2}\A_i^\top\A_i+\left(1+\frac{\eps}{2}\right)\A_j^\top\A_j\\
\C^L_{i,j}&=-\lambda\I-\frac{\eps}{2}\A_i^\top\A_i+\left(1-\frac{\eps}{2}\right)\A_j^\top\A_j,
\end{align*}
where $\A_i=\a_1\circ\ldots\circ\a_i$ for each $i\in[n]$. 
Then by definition, $\C^U_{i,i}=\B_i^U$ and $\C^L_{i,i}=\B_i^L$, where $\B_i^U$ and $\B_i^L$ are the matrices $\B^U$ and $\B^L$ at time $i$. 
Finally, denoting $\M$ at time $i$ by $\M_i$, the matrices $\Y^U_{i,j}$ and $\Y^L_{i,j}$ are defined by:
\begin{align*}
\Y^U_{i,j}&=\C^U_{i,j}-\M_j^\top\M_j\\
\Y^L_{i,j}&=\M_j^\top\M_j-\C^L_{i,j},
\end{align*}
so that $\Y^U_{i,i}=\X^U_i$ and $\Y^L_{i,i}=\X^L_i$ are the matrices that represent the ``distance'' of $\M^\top\M+\lambda\I$ from the upper and lower barriers, to determine the sampling probability $p_i$. 
It suffices for \cite{CohenMP16} to bound $\EEx{\a_i^\top(\Y^U_{i-1,j+1})^{-1}\a_i}\le\EEx{\a_i^\top(\Y^U_{i-1,j})^{-1}\a_i}$ and $\EEx{\a_i^\top(\Y^L_{i-1,j+1})^{-1}\a_i}\le\EEx{\a_i^\top(\Y^L_{i-1,j})^{-1}\a_i}$, since $\a_i^\top(\Y^U_{i-1,0})^{-1}\a_i$ and $\a_i^\top(\Y^L_{i-1,0})^{-1}\a_i$ are proportional to the online ridge leverage score $\tau^{\OL}_i$ of $\a_i$. 
Specifically, they show that
\[\a_i^\top(\Y^U_{i-1,0})^{-1}\a_i=\a_i^\top(\Y^L_{i-1,0})^{-1}\a_i\le\frac{2}{\eps}\tau^{\OL}_i.\]
Thus by \lemref{lem:online:space}, the expected number of rows sampled is at most
\begin{align*}
\sum_{i=1}^n c_U\EEx{\a_i^\top(\X^U_{i-1})^{-1}\a_i}+c_L\EEx{\a_i^\top(\X^L_{i-1})^{-1}\a_i}&=c_U\EEx{\a_i^\top(\Y^U_{i-1,i-1})^{-1}\a_i}+c_L\EEx{\a_i^\top(\Y^L_{i-1,i-1})^{-1}\a_i}\\
&\le\sum_{i=1}^n c_U\EEx{\a_i^\top(\Y^U_{i-1,0})^{-1}\a_i}+c_L\EEx{\a_i^\top(\Y^L_{i-1,0})^{-1}\a_i}\\
&\le\frac{2}{\eps}(c_U+c_L)\sum_{i=1}^n\tau^{\OL}_i=\O{\frac{d}{\eps^2}\log\kappa}.
\end{align*}

We make the following modifications to $\onlinebss$. 
We first round the sampling probability for each row up to a power of $2$ between $\frac{1}{n}$ and $1$, which only incurs $\O{1}$ additional rows sampled in expectation with those with sampling probability less than $\frac{1}{n}$ and multiplicative $\O{1}$ for those with samping probability at least $\frac{1}{n}$. 
We then pick a random threshold $t\in [0,1]$ and then sample row $\a_i$ if $p_i>t$. 
Since there are only $\log n$ distinct sampling probabilities, we can derandomize this process by trying all $\log n$ values of $t$. 
It remains to be shown that the expectation of $p_i$ is still at most the online (ridge) leverage score $\tau^{\OL}_i$. 

However, the bounds $\EEx{\a_i^\top(\Y^U_{i-1,j+1})^{-1}\a_i}\le\EEx{\a_i^\top(\Y^U_{i-1,j})^{-1}\a_i}$ and $\EEx{\a_i^\top(\Y^L_{i-1,j+1})^{-1}\a_i}\le\EEx{\a_i^\top(\Y^L_{i-1,j})^{-1}\a_i}$ hold even when conditioned on the first $j$ steps of the algorithm. 
The analysis of \cite{CohenMP16} only requires that row $\a_j$ is sampled with probability $p_j$ at step $j+1$ for whatever value of $p_j$ results from the randomness of the algorithm. 
Thus by the same reasoning, it still follows that the expected number of rows sampled is at most
\begin{align*}
\sum_{i=1}^n c_U\EEx{\a_i^\top(\X^U_{i-1})^{-1}\a_i}+c_L\EEx{\a_i^\top(\X^L_{i-1})^{-1}\a_i}&=c_U\EEx{\a_i^\top(\Y^U_{i-1,i-1})^{-1}\a_i}+c_L\EEx{\a_i^\top(\Y^L_{i-1,i-1})^{-1}\a_i}\\
&\le\sum_{i=1}^n c_U\EEx{\a_i^\top(\Y^U_{i-1,0})^{-1}\a_i}+c_L\EEx{\a_i^\top(\Y^L_{i-1,0})^{-1}\a_i}\\
&\le\frac{2}{\eps}(c_U+c_L)\sum_{i=1}^n\tau^{\OL}_i,
\end{align*}
which is at most $\O{\frac{d}{\eps^2}\log\kappa}$ by \lemref{lem:online:space}.
\end{document}